\documentclass[11pt]{article}

\usepackage{authblk}
\usepackage{amsmath, amssymb, amsthm, mathrsfs, bm}
\usepackage{geometry, extarrows, xcolor, graphicx, enumitem, booktabs, subcaption}
\usepackage[citestyle=authoryear,natbib=true,backend=bibtex]{biblatex}
\usepackage{titlesec}
\usepackage{hyperref}
\hypersetup{
    colorlinks,
    linkcolor={blue!50!black},
    citecolor={blue!50!black},
    urlcolor={blue!80!black}
}
\usepackage{array}
\usepackage[ruled,vlined]{algorithm2e}
\newcolumntype{A}{>{\centering\arraybackslash}m{0.12\columnwidth}}
\newcolumntype{B}{>{\centering\arraybackslash}m{0.4\columnwidth}}

\usepackage{fullpage}

\DeclareMathOperator*{\argmin}{argmin}
\DeclareMathOperator*{\argmax}{argmax}
\newcommand{\indpt}{\perp \!\!\! \perp}
\newcommand{\cond}[2]{P_{#1|#2}}

\theoremstyle{plain}
\newtheorem{theorem}{Theorem}

\newtheorem{lemma}[theorem]{Lemma}
\newtheorem{corollary}[theorem]{Corollary}
\newtheorem{assumption}{Assumption}

\theoremstyle{remark}

\newtheorem{remark}{Remark}

\title{Sufficiently Reduced Distributional Regression}
\author[1]{Alexander Henzi}
\author[2]{Tiange Liu}
\author[2]{Xinwei Shen}
\affil[1]{Tsinghua University, Department of Statistics and Data Science, \href{mailto:henzia@tsinghua.edu.cn}{henzia@tsinghua.edu.cn}}
\affil[2]{University of Washington, \href{mailto:ivytgliu@uw.edu}{ivytgliu@uw.edu} and \href{mailto:xwshen@uw.edu}{xwshen@uw.edu}}
\date{\today}

\begin{document}
\maketitle

\begin{abstract}
We propose Sufficiently Reduced Distributional Regression (SRDR), a generative method that combines conditional distribution estimation with nonlinear sufficient dimension reduction (SDR). It builds on a characterization of sufficiency through strictly proper scoring rules: a dimension reduction is sufficient if and only if predicting the response from the reduced covariates incurs no loss in expected score relative to the full covariates. Sufficient dimension reduction thus becomes a risk minimization problem. SRDR jointly trains a dimension reduction map and a generative prediction model by minimizing the energy score, which can be estimated by sampling without density evaluation or adversarial training. The framework extends to multi-environment data and to classification. We prove that the estimated conditional distributions converge in energy distance to the true ones, which implies that the learned representation is asymptotically sufficient. In simulations and applications to CT slice localization, superconductivity, and digit classification, SRDR recovers low-dimensional sufficient structure and matches or outperforms state-of-the-art nonlinear SDR methods in representation quality and predictive performance.
\end{abstract}

\normalsize


\section{Introduction}

The goal of distributional regression is to estimate the full distribution of an outcome variable $Y \in \mathbb{R}^m$ conditional on covariates $X \in \mathbb{R}^p$. While there are standard tools for estimating conditional distributions for univariate outcomes, such as via their density function, cumulative distribution function (CDF), or quantile function, distributional regression becomes more difficult for multivariate outcomes and high-dimensional covariates. On the one hand, this is due to the curse of dimensionality, making it difficult to estimate a functional response --- the conditional distribution --- given only point-valued observations. On the other hand, the calculation and usage of multivariate densities or distribution functions is often impractical, and it is more attractive to generate samples from the outcome distribution of interest, from which any desired quantity can be approximated via its empirical counterpart from the sample. These difficulties motivated the use of deep learning and generative modeling for conditional distribution estimation, where one estimates a sampling mechanism $d(X, \eta)$ that takes the covariate $X$ and independent noise $\eta \in \mathbb{R}^s$ as input and generates samples that should ideally follow the conditional distribution $\cond{Y}{X}$ of $Y$ given $X$ \citep{Zhou2023, Song2025, ShenMeinshausen2025}. 

From a statistical perspective, approximating conditional distributions with high-dimensional covariates and outcome variables to a reasonable accuracy is typically only feasible if the underlying structure of the data is not truly high-dimensional, but rather governed by a simpler, lower-dimensional function $e(X) \in \mathbb{R}^q$ of the covariates, which contains sufficient information about the outcome $Y$. Formally, this means that $Y$ and $X$ are independent given $e(X)$, in symbols $Y \indpt X \mid e(X)$, and the remaining randomness of $Y$ after conditioning on $X$ is captured via independent noise $\eta$. This is precisely the definition of sufficiency of $e(X)$ in the literature on dimension reduction \citep{Li2018}, and equivalent to $\cond{Y}{X}$ being equal to $\cond{Y}{e(X)}$. In a generative modeling setting, sufficiency requires that the generator is now of the form $d(e(X),\eta)$, taking lower-dimensional features $e(X)$ instead of the full covariate $X$ as input.

While neural networks and generative models have been applied successfully to each of the two problems discussed above, distributional regression and dimension reduction, the combination has not been widely considered in the literature so far. For instance, \citet{Zhou2023,Song2025} propose generative methods for estimating conditional distributions, without explicit dimension reduction. \citet{Liang2022,Chen2024,TangLi2025} propose state-of-the-art methods for nonlinear sufficient dimension reduction (SDR) using neural networks, but they do not consider the problem of generating samples of $Y$ based on the lower-dimensional representation of $X$. Hence, dimension reduction and distributional regression are mostly considered as disjoint tasks. Notable exceptions are the recent work by \citet{XuEtAl2025}, who propose an approach for dimension reduction and conditional sampling based on conditional stochastic interpolation; and the work by \citet{TanLiXue2026}, who simultaneously developed a method that is essentially equal to our proposal. We discuss the latter in more detail in Section \ref{sec:literature}.

The main contribution of this paper is a new characterization of sufficient dimension reduction via strictly proper scoring rules. A proper scoring rule is a function $S = S(P, Y)$ mapping a probability measure $P$ and outcome $Y$ to a numerical score, such that
\begin{equation} \label{eq:proper_scoring_rule}
    \mathbb{E}_{Y \sim P}[S(P,Y)] \leq \mathbb{E}_{Y\sim P}[S(Q,Y)]
\end{equation}
for all $P,Q$ in a certain set $\mathcal{P}$, and $S$ is strictly proper if equality in \eqref{eq:proper_scoring_rule} only holds for $P = Q$ \citep{GneitingRaftery2007,WaghmareZiegel2025}. We prove that a dimension reduction $e(X)$ is sufficient if, and only if, the expected loss $\mathbb{E}[S(\cond{Y}{e(X)}, Y)]$ for predicting $Y$ based on the conditional distribution $\cond{Y}{e(X)}$ equals $\mathbb{E}[S(\cond{Y}{X}, Y)]$, the expected loss based on the full covariate $X$, which is the smallest over all measurable reductions of $X$. This connection makes it possible to perform sufficient dimension reduction via any sufficiently flexible distributional regression method that enforces the compression of the covariates $X$ to some lower dimensional representation $e(X)$. While it is well-known that sufficient dimension reduction can be performed by optimizing various loss functions or dependence measures \citep{Li2018}, the connection to proper scoring rules is new to the dimension reduction literature, to the best of our knowledge. It phrases sufficient dimension reduction as a standard risk minimization problem, and allows us to simultaneously reduce dimension and target the accuracy of probabilistic predictions for $Y$ based on the dimension reduction $e(X)$.

Based on the above insight, we propose a flexible generative method for sufficient dimension reduction and distributional regression, called Sufficiently Reduced Distributional Regression (SRDR). For a sample $(X_1, Y_1), \dots, (X_n, Y_n)$ with unknown conditional distributions $\cond{Y}{X}$, and a class $\mathcal{M}$ of pairs $(e,d)$, which we will take to be deep or wide neural networks, we estimate a representation $d(e(X),\eta)$ of $Y$ by minimizing the average empirical score,
\[
    \min_{(e,d) \in \mathcal{M}} \ \frac{1}{n}\sum_{i=1}^n S(P_{e,d,X_i}, Y_i).
\]
The expectation of this average score is minimal only if the distribution $P_{e,d,X_i}$ of the generated outcomes $d(e(X_i),\eta_i)$, conditional on $X_i$, equals $\cond{Y}{X_i}$; and the dimension of $e(X) \in \mathbb{R}^q$ is chosen smaller than $p$, to enforce dimension reduction. This is in the same spirit as the bottleneck in autoencoders \citep[see, e.g.,][and the references therein]{Ghosh2022} or the belted neural networks proposed by \citet{TangLi2025} for sufficient dimension reduction, and yields the desired lower-dimensional representation of the data. In practice, we take the energy score for estimation, which is defined as
\begin{equation} \label{eq:energy_score}
        \mathrm{ES}(P,Y) =  \mathbb{E}_{Z \sim P}[\|Y - Z\|] - \frac{1}{2}\mathbb{E}_{Z, Z' \sim P}[\|Z - Z'\|],
\end{equation}
where $\|\cdot\|$ denotes the Euclidean norm and $Z,Z'$ in the second expectation are independent \citep{GneitingRaftery2007}. Hence, for the identity dimension reduction map $e(x)=x$, the estimation of $d$ in $d(x,\eta)$ is equivalent to the engression method proposed by \citet{ShenMeinshausen2025}. The energy score has favorable practical properties for estimating generative models, in that it can be easily approximated with an unbiased estimator by sampling, and does not require density estimates or adversarial training. 

From a theoretical perspective, we show that minimizing the energy score yields a statistically consistent estimator of the conditional distributions $\cond{Y}{X}$ in the so-called energy distance \citep{SzekelyRizzo2023}, which is a metric on the space of probability distributions. In concurrent work, \citet{chen2026error,huang2026theoretical} analyze engression, where the noise $\eta$ enters the generator together with the full covariate $X$. In our model, $\eta$ enters only after the dimension reduction $e(X)$, and this is what links consistency of the conditional distribution estimator to consistency of the dimension reduction: the fact that the estimated conditional distributions given the lower-dimensional $\hat{e}_n(X)$ are consistent for $\cond{Y}{X}$ implies that $\hat{e}_n(X)$ is sufficient, in an asymptotic sense; see Section \ref{sec:thmdiscuss} for a comparison. As a relevant extension, we show that our method and the theoretical guarantees also apply to heterogeneous multi-environment data, where both the marginal distribution of $X$ and the conditional distributions $\cond{Y}{X}$ can differ between environments, but a shared dimension reduction $e(X)$ is sufficient in all environments. This corresponds to a nonlinear version of the linear partial sufficient dimension reduction \citep{Chiaromonte2002}. Furthermore, we extend SRDR to classification problems, where a categorical outcome $Y$ is represented by a one-hot vector and the method applies with only minor modifications compared to a regression setting.

We apply SRDR to a wide range of regression and classification tasks. Simulations demonstrate that in problems where the relationship between $X$ and $Y$ is governed by a lower-dimensional sufficient reduction $e(X)$, SRDR outperforms direct engression in terms of out-of-sample energy score as prediction error, highlighting the advantages of incorporating sufficient dimension reduction into distributional regression and generative modeling. Furthermore, the recovered sufficient representation shows a stronger dependence, as measured by distance correlation \citep{Szekely2007}, with the outcome variable than the dimension reduction obtained with competing methods. Empirical studies on real data, including CT slice localization \citep{GrafCavallaro2011}, superconductivity data \citep{Hamidieh2018}, and MNIST classification, show that SRDR outperforms competing SDR methods, both in the quality of the dimension reduction and in predictive accuracy. In transfer learning experiments with the multi-environment SRDR, transferring the pretrained reduction to a new environment gives a consistent gain over training the same architecture from scratch, and SRDR is competitive with the method by \citet{GeZhouHuang2025}.

The structure of the article is as follows. In Section \ref{sec:literature} we provide an overview of related literature. The connections between proper scoring rules and sufficient dimension reduction, which form the basis of our approach, are elaborated in Section \ref{sec:scoring}. Section \ref{sec:srdr} provides the description of our algorithm, and an extension to data grouped into multiple environments. In Section \ref{sec:guarantees} we prove consistency. Empirical results in Section \ref{sec:simulations} and \ref{sec:app} demonstrate that our method is competitive with state-of-the-art methods for distributional regression and for sufficient dimension reduction.

\section{Related Literature} \label{sec:literature}

Distributional regression and sufficient dimension reduction are well-studied problems in the literature. Classical statistical methods estimate the conditional distribution of an outcome variable in terms of the density function, distribution function, and typically include some form of dimension reduction. For instance, in the popular generalized additive models for location, scale, and shape \citep{Rigby2005}, the outcome $Y$ is assumed to follow a parametric family, and the covariates $X$ only affect it through the low-dimensional parameters of the family. Semiparametric models reduce the covariates to a low-dimensional projection or predictor $\theta(X)$, which is assumed to be sufficient, and the conditional distribution of $Y$ given $\theta(X)$ is estimated nonparametrically \citep{HallYao2005,Henzi2023,Balabdaoui2024,Walz2024}. Methods based on the distribution or quantile function model individual probabilities or quantiles, often with parametric dependence on $X$, and can be viewed as semiparametric approaches \citep{Foresi1995,Koenker2005,Chernozhukov2013}. A recent overview of distributional regression is given by \citet{Kneib2023}.
These classical approaches combine distributional modeling with dimension reduction, but are often limited in their ability to handle complex, high-dimensional data. Motivated by such settings, machine learning methods have become increasingly popular, including distributional random forests \citep{Cevid2022}, conditional generative adversarial networks \citep{Mirza2014}, variational autoencoders \citep{Kingma2013}, and conditional stochastic interpolation \citep{Huang2023}. While these methods offer great flexibility and scalability, they typically focus on conditional distribution estimation and do not explicitly recover sufficient low-dimensional representations of the covariates.

Sufficient dimension reduction aims at finding a lower-dimensional representation $e(X)$ of the covariates such that $Y \indpt X \mid e(X)$, thereby containing all relevant information for predicting $Y$ in $e(X)$. Earlier work focused on linear functions, beginning with inverse regression proposed by \citet{Li1991}; see \citet{Li2018} for an overview. In the nonlinear case, \citet{Lee2013} developed a general theory for sufficient dimension reduction based on dimension reduction $\sigma$-fields, which are sub-$\sigma$-fields $\mathcal{G} \subseteq \sigma(X)$ for which
\[
    Y \indpt X \mid \mathcal{G}.
\]
A sufficient dimension reduction $e(X)$ is not identifiable, since any injective transformation of $e(X)$ is still sufficient; however, there exists a unique, minimal dimension reduction $\sigma$-field under mild conditions \citep[Theorem 1]{Lee2013}. Popular methods for nonlinear sufficient dimension reduction include kernel embeddings \citep[e.g.,][]{Yeh2009, LiEtAl2011}, random forests \citep{Loyal2022, Dai2024}, and neural network based methods that optimize dependence measures between $e(X)$ and $Y$ as target function, such as distance covariance \citep{Huang2024, JiaoEtAl2024, GeZhouHuang2025}, or generalized martingale difference divergence \citep{Chen2024}. None of these approaches considers the problem of generating samples for $Y$ based on the dimension reduction $e(X)$.

Existing methods that share some similarity to our approach are by \citet{Liang2022, XuEtAl2025, TangLi2025}. \citet{Liang2022} propose a method based on stochastic neural networks, which, like ours, inserts random noise into layers of the neural network for a stochastic approximation of the data. However, their model assumes additive noise and a layer-wise structure for the function $e$ \citep[Equation 6]{Liang2022}, and their method yields a random dimension reduction, depending on Gaussian noise inserted into the network. \citet{XuEtAl2025} propose a method based on conditional stochastic interpolation, which learns an interpolation path from Gaussian noise to the terminal distribution $\cond{Y}{X=x}$. Under suitable regularity conditions, this yields a sufficient dimension reduction and a consistent estimation of the conditional distribution. In contrast, our approach directly targets conditional predictive accuracy through a loss function, and only requires a single forward pass through the learned model for the generation of samples. \citet{TangLi2025} characterize sufficiency of $e(X)$ through the moment conditions $\mathbb{E}[g(Y,t) | e(X)] = \mathbb{E}[g(Y,t) | X]$, for a family of functions $g(y,t)$ indexed by a parameter $t \in I \subset \mathbb{R}$, for univariate $Y$ and without predicting $Y$. Section \ref{sec:test_functions} gives the close connection between their approach and minimizing the energy score.

Our method builds on two generative methods that are trained by minimizing the energy score. Engression \citep{ShenMeinshausen2025} fits a sampling mechanism $d(X, \eta)$ for $\cond{Y}{X}$ without dimension reduction and is the special case $e(x) = x$ of our method. Distributional principal autoencoders \citep{ShenMeinshausen2024} learn a low-dimensional representation of $X$ from which $X$ itself is generated, and our method is equivalent to them in the special case $X = Y$. However, the notion of sufficiency in the regression setting considered here is not sensible in unsupervised dimension reduction, where it would mean that $X \indpt X \mid e(X)$, i.e., $e(X) \in \mathbb{R}^q$ is a lossless compression of $X \in \mathbb{R}^p$.

During the preparation of the current manuscript, we came across the very recent work by \citet{TanLiXue2026}, who concurrently and independently propose essentially the same method. Their analysis focuses on finite-sample convergence rates and on the parameter efficiency of the bottleneck architecture, whereas our characterization of sufficiency holds for general strictly proper scoring rules, and we also treat multi-environment data and classification.

\section{Proper Scoring Rules and Sufficiency} \label{sec:scoring}

\subsection{Notation}
We assume that $X \in \mathbb{R}^p$ and $Y \in \mathbb{R}^m$, where the underlying probability space is always equipped with the Borel $\sigma$-field. The marginal distribution of $X$ is denoted by $P_X$, and the conditional distributions of $Y$ given $e(X)$ with a measurable function $e\colon \mathbb{R}^p \rightarrow \mathbb{R}^q$ by $\cond{Y}{e(X)}$. For a composition of the form $d(e(X),\eta)$ with $\eta$ independent of $X$, we use $P_{e,d,X}$ to abbreviate the conditional distribution of $d(e(X),\eta)$ given $X$. We write $\mathbb{E}_{U \sim P}[\cdot]$ for the expectation over $U$ with distribution $P$ whenever it is necessary to indicate over which distribution the expectation is calculated. A measurable function $e\colon \mathbb{R}^p \rightarrow \mathbb{R}^q$ is called a dimension reduction of $X$; we interchangeably use the term ``dimension reduction'' to refer to $e$ or to the random vector $e(X)$. A dimension reduction is sufficient if $Y \indpt X \mid e(X)$, or, equivalently, $\cond{Y}{X} = \cond{Y}{e(X)}$ almost surely. 

\subsection{Proper Scoring Rules, Divergence, and Entropy} \label{sec:scoring_rules}

A proper scoring rule \citep{GneitingRaftery2007,WaghmareZiegel2025} is a function $S = S(P, Y)$ for probability measures $P$ and outcomes $Y$, such that
\begin{equation*}
    \mathbb{E}_{Y \sim P}[S(P,Y)] \leq \mathbb{E}_{Y\sim P}[S(Q,Y)]
\end{equation*}
for all $P,Q$ in a certain set $\mathcal{P}$; it is strictly proper if equality only holds for $P = Q$. Well-known examples are the logarithmic score, defined as $S(P,Y) = -\log(f_P(Y))$, where $f_P$ is the density of $P$, or the energy score \eqref{eq:energy_score}. The latter is strictly proper with respect to all distributions $P$ such that $\mathbb{E}_P[\|Z\|] < \infty$.
Associated with every proper scoring rule there is a divergence
\[
    D(P,Q) = \mathbb{E}_P[S(Q,Y)] - \mathbb{E}_P[S(P,Y)] 
\]
which is non-negative and equals zero only if $P = Q$, if $S$ is strictly proper. Divergence measures of proper scoring rules are not metrics, but often behave like a squared metric on probability distributions \citep{WaghmareZiegel2025}. The entropy is
\[
    H(P) = \mathbb{E}_P[S(P,Y)],
\]
and, intuitively, measures how diffuse the distribution $P$ is.

\subsection{Characterization of Sufficiency via Proper Scoring Rules}

We now present our main result on the connection between proper scoring rules and sufficiency. Suppose that we have a candidate dimension reduction $e(X)$ and want to verify if it is sufficient, which amounts to checking $Y \indpt X \mid e(X)$. It turns out that simply comparing the loss in terms of a strictly proper scoring rule answers this question.

\begin{theorem} \label{thm:minimizer_sufficient}
Let $S$ be a strictly proper scoring rule with respect to a set of probability measures $\mathcal{P}$ on $\mathbb{R}^m$. Assume that $P_{Y|X} \in \mathcal{P}$ almost surely and that $\mathbb{E}[|S(P_{Y|X},Y)|] < \infty$. Then for all measurable $e \colon \mathbb{R}^p \rightarrow \mathbb{R}^q$, it holds that
\[
    \mathbb{E}[S(P_{Y|e(X)}, Y)] \geq \mathbb{E}[S(P_{Y|X}, Y)],
\]
with equality if and only if $e$ is sufficient.
\end{theorem}

\begin{proof}
Because $S$ is strictly proper and $\mathbb{E}[|S(\cond{Y}{X},Y)|] < \infty$, we have
\begin{align*}
    \mathbb{E}[S(\cond{Y}{e(X)}, Y)] & = \mathbb{E}_{X \sim P_X}\left[\mathbb{E}_{Y \sim \cond{Y}{X}}[S(\cond{Y}{e(X)}, Y) \mid X]\right] \\
    & \geq  \mathbb{E}_{X \sim P_X}\left[\mathbb{E}_{Y \sim \cond{Y}{X}}[S(\cond{Y}{X}, Y)]\right],
\end{align*}
with equality if and only if $\cond{Y}{e(X)} = \cond{Y}{X}$ almost surely, i.e., if and only if $e$ is sufficient \citep[see, e.g.,][Section 1 (b)]{AdragniCook2009}.
\end{proof}

The statement of Theorem \ref{thm:minimizer_sufficient} also holds for $e(X)$ replaced by any sub-$\sigma$-field of $\sigma(X)$, i.e., in the setting by \citet{Lee2013}. The fact that reducing $X$ by a coarsening $e(X)$ increases the expected score was already known in the statistical forecasting literature \citep[Theorem 3]{Holzmann2014}, but, to the best of our knowledge, the connection to sufficient dimension reduction has not been exploited so far. Compared to existing results on the characterization of sufficiency for nonlinear dimension reduction, our theorem is relatively short and simple to state, because minimality of the expected score directly relates to sufficiency, via conditioning of the expectation on $X$. Thus, sufficiency is characterized entirely in terms of predictive optimality. By comparison, \citet{TangLi2025} require precise definitions of when a function class characterizes sufficiency, and \citet{XuEtAl2025} impose technical assumptions on the continuity of their conditional velocity field. This again highlights that proper scoring rules as an estimation criterion for distributional regression naturally align with the definition of sufficiency. 
To provide more intuition for our approach, we reformulate Theorem \ref{thm:minimizer_sufficient} in terms of the divergence and entropy.

\begin{corollary} \label{cor:divergence_entropy}
Under the assumptions of Theorem \ref{thm:minimizer_sufficient}, for a strictly proper scoring rule with divergence $D$ and entropy $H$, a dimension reduction $e$ is sufficient if and only if $\mathbb{E}[D(\cond{Y}{e(X)}, \cond{Y}{X})] = 0$, or if and only if $\mathbb{E}[H(\cond{Y}{e(X)})] = \mathbb{E}[H(\cond{Y}{X})]$.
\end{corollary}

The interpretation in terms of the divergence is straightforward, in that zero divergence implies equality of $\cond{Y}{X} = \cond{Y}{e(X)}$, almost surely. When minimizing a scoring rule $S$ as target function for distributional regression, one can expect convergence of the conditional distributions in the corresponding divergence measure $D$; we prove this in Theorem \ref{thm:consistencyMULTIY} for our method with the energy score. For the entropy, if $e(X)$ is not sufficient, then there is randomness in $Y$ that can be explained by $X$ even after conditioning on $e(X)$, causing $\cond{Y}{e(X)}$ to be more diffuse than $\cond{Y}{X}$, as measured with respect to $H$. So only a sufficient $e(X)$ achieves minimal entropy.

Theorem \ref{thm:minimizer_sufficient} and Corollary \ref{cor:divergence_entropy} are general, model-agnostic characterizations of sufficiency in terms of prediction optimality. In Section \ref{sec:srdr}, we apply these results to propose a practical estimator based on generative models.

\subsection{Sufficient Dimension Reduction with the Energy Score} \label{sec:test_functions}

The scoring rule of main interest in this article is the energy score \eqref{eq:energy_score}, whose divergence measure equals the so-called energy distance \citep{SzekelyRizzo2023}, up to a scaling factor of $1/2$,
\begin{align}
    \mathcal{D}^2(P,Q) & = \mathbb{E}_{Z \sim P, Y \sim Q}[\|Y - Z\|] - \frac{1}{2}\mathbb{E}_{Z, Z' \sim P}[\|Z - Z'\|] - \frac{1}{2}\mathbb{E}_{Y, Y' \sim Q}[\|Y - Y'\|] \nonumber \\
    & = \frac{\Gamma((m+1)/2)}{2\pi^{(m+1)/2}} \int_{\mathbb{R}^m} \frac{|\phi_P(t) - \phi_Q(t)|^2}{\|t\|^{m+1}} \, dt, \label{eq:energy_divergence}
\end{align}
with the characteristic functions $\phi_P(t) = \mathbb{E}_P[\exp(it^{\top}Z)]$, $\phi_Q(t) = \mathbb{E}_Q[\exp(it^{\top}Z)]$. For the energy score, there is a close connection between our Theorem \ref{thm:minimizer_sufficient} and the approach of \citet{TangLi2025}. Analogously to formula \eqref{eq:energy_divergence} for the divergence, and since $\mathrm{ES}(P,Y) = \mathcal{D}^2(P,\delta_Y)$ for the point mass $\delta_Y$ at $Y$, the energy score can be written as a squared distance on the characteristic functions,
\[
    \mathrm{ES}(P,Y) = \frac{\Gamma((m+1)/2)}{2\pi^{(m+1)/2}} \int_{\mathbb{R}^m} \frac{|\phi_P(t) - \exp(it^{\top}Y)|^2}{\|t\|^{m+1}} \, dt.
\]
Conditioning on $X$ as in the proof of Theorem \ref{thm:minimizer_sufficient} gives $\mathbb{E}[\mathrm{ES}(\cond{Y}{e(X)}, Y)] - \mathbb{E}[\mathrm{ES}(\cond{Y}{X}, Y)] = \mathbb{E}[\mathcal{D}^2(\cond{Y}{e(X)}, \cond{Y}{X})]$, which by Corollary \ref{cor:divergence_entropy} is zero if and only if $e$ is sufficient. Applying \eqref{eq:energy_divergence} conditionally on $X$, with the conditional characteristic functions $\mathbb{E}[\exp(it^{\top}Y) \mid e(X)]$ and $\mathbb{E}[\exp(it^{\top}Y) \mid X]$ in place of $\phi_P$ and $\phi_Q$, yields
\[
    \mathbb{E}[\mathcal{D}^2(\cond{Y}{e(X)}, \cond{Y}{X})] = \frac{\Gamma((m+1)/2)}{2\pi^{(m+1)/2}} \int_{\mathbb{R}^m} \frac{\mathbb{E}\left[\left|\mathbb{E}[\exp(it^{\top}Y) \mid e(X)] - \mathbb{E}[\exp(it^{\top}Y) \mid X]\right|^2\right]}{\|t\|^{m+1}} \, dt.
\]
Hence, the expected energy score is minimal, and $e$ is sufficient, if and only if
\begin{align*}
    & \mathbb{E}\left[g(Y, t) | e(X)\right] = \mathbb{E}\left[g(Y, t) | X\right]
\end{align*}
almost surely for all $t \in \mathbb{R}^m$ and for both $g = g_1$ and $g = g_2$, where $g_1(y, t) = \cos(t^{\top}y)/\|t\|^{(m+1)/2}$ and $g_2(y, t) = \sin(t^{\top}y)/\|t\|^{(m+1)/2}$ are the real and imaginary parts of $\exp(it^{\top}y)/\|t\|^{(m+1)/2}$. \citet{TangLi2025} propose to perform sufficient dimension reduction by minimizing
\[
    \mathbb{E}\left[\int_I (h(e(X),t) - g(Y,t))^2\, d\mu(t)\right]
\]
over $e$ and $h$, for a weighting measure $\mu$ on $I$ and a class of test functions $g$ that is sufficiently rich to characterize the conditional distribution of $Y$ given $X$. At the population level the minimizing $h$ is $h(e(X),t) = \mathbb{E}[g(Y,t) | e(X)]$, so this contains our energy score target above as a special case, up to the constant factor in \eqref{eq:energy_divergence} and an additive term that does not depend on $e$, when $\mu$ is the Lebesgue measure on $I = \mathbb{R}^m$ and the objectives for $g = g_1$ and $g = g_2$ are added.

The crucial difference to \citet{TangLi2025} is that the parameter $t$ of our test functions is in $\mathbb{R}^m$, whereas they require a bounded subset of $\mathbb{R}$. Moreover, \citet{TangLi2025} discretize the integral over $t$ to a finite grid $t_1, \dots, t_M$, and estimate a neural network with output dimension $M$. The network approximates the expectation $(\mathbb{E}[g(Y,t_1)| X], \dots,  \mathbb{E}[g(Y,t_M)| X])$, but without explicit dependence on the parameters $t_j$. In contrast, we estimate the distribution $\cond{Y}{X}$, which yields the expectation of the test functions for any desired $t$, and easily scales to multivariate outcomes. While the characteristic function reformulation is helpful for theoretical analyses of consistency, we do not apply it in the practical implementation, since the expectation-based formulation of the energy score is computationally more efficient.

\section{Sufficiently Reduced Distributional Regression} \label{sec:srdr}

\subsection{Population objective}
We now turn the characterization of sufficiency through proper scoring rules into an estimation method based on generative models. It is well known that for any $(X,Y)$, there exists a function $G$ and a noise variable $\eta \in \mathbb{R}^s$ independent of $X$ such that $Y = G(X, \eta)$ almost surely, where $\eta$ can have any continuous reference distribution, such as Gaussian or uniform. See, e.g., \citet[Lemma 2.1]{Zhou2023} or \citet[Equation (1)]{Song2025}. The following lemma extends this result to a setting with dimension reduction. We tacitly assume that the underlying probability space for $(X,Y)$ is atomless, meaning that there exists a random variable with distribution $\mathrm{Unif}(0,1)$ defined on the probability space, which is always possible on a suitable extension if necessary.

\begin{lemma} \label{lem:generative_sufficiency}
Let $(X, Y) \in \mathbb{R}^{p+m}$ have an arbitrary distribution on $\mathbb{R}^{p+m}$ equipped with the Borel $\sigma$-field $\mathcal{B}(\mathbb{R}^{p+m})$. For any integers $q, s \geq 1$ and for any atomless distribution $\nu$ on $(\mathbb{R}^s, \mathcal{B}(\mathbb{R}^s))$, there exist measurable $e \colon \mathbb{R}^p \rightarrow \mathbb{R}^q$, $d\colon \mathbb{R}^{q+s} \rightarrow \mathbb{R}^m$ and a random variable $\eta \sim \nu$ such that
\[
    Y = d(e(X),\eta) \quad \text{almost surely}.
\]
\end{lemma}

The above result shows that a sufficient dimension reduction exists for any dimension $q$ of $e(X)$, even for $q = 1$, independently of the dimensions of $X$ and $Y$. This is due to the fact that there exist bijective measurable --- but complicated --- mappings between any uncountable Borel spaces \citep[see, e.g.,][]{RaoSrivastava1994}. Clearly, the functions $d, e$ in Lemma \ref{lem:generative_sufficiency} can only be approximated well in practice if they satisfy additional properties, like smoothness, as we impose for our consistency result in Section \ref{sec:guarantees}. Such restrictions rule out arbitrary dimensions $s, q$ and make the formulation $Y = d(e(X),\eta)$ a model assumption.

Motivated by Lemma \ref{lem:generative_sufficiency}, we propose to estimate $\cond{Y}{X=x}$ by approximating a generator of the form $d(e(x),\eta)$, with distribution $P_{e,d,x}$, which is the push-forward measure of $\eta$ under $\eta \mapsto d(e(x),\eta)$. Combined with strictly proper scoring rules, our objective function for distributional regression and sufficient dimension reduction then becomes
\[
    \min_{(e,d)} \ \mathbb{E}[S(P_{e,d,X}, Y)],
\]
which requires specifying a function class for $(e,d)$ and the scoring rule $S$.

\subsection{Method Description}

In our method, we use the energy score for estimation, because it admits the computationally simple unbiased estimator
\[
    \widehat{\mathrm{ES}}_N(P_{e,d,x},Y) = \frac{1}{N}\sum_{i=1}^N \|d(e(x),\eta_i) - Y\| - \frac{1}{2N(N-1)}\sum_{i,j=1}^N \|d(e(x),\eta_i) - d(e(x),\eta_j)\|,
\]
where $\eta_1, \dots, \eta_N$ are sampled from the reference distribution for $\eta$. Optimization of the energy score can be done efficiently via stochastic gradient descent, like in the engression method of \citet{ShenMeinshausen2025}. To parametrize the functions $e, d$, we choose deep or wide rectified linear unit (ReLU) networks. The schematic description of our estimation method is in Algorithm \ref{alg:srdr}. For given dimension $q$ of $e$ and noise $\eta$ with dimension $s$, we initialize the functions $e$, $d$, pass samples $X_i$ and noise variables $\eta_{i,j}$, $i = 1, \dots, n$, $j = 1, \dots, N$ to generate outcomes $d(e(X_i),\eta_{i,j})$ and approximate the energy score, and optimize the neural network parameters via gradient descent until a stopping criterion is reached. In practice, the average over the training sample is replaced by an average over a minibatch in each iteration.
For the noise $\eta$, sensible distributions are independent standard Gaussian variables or the uniform distribution on an interval $[-c, c]$ for some bound $c > 0$. The former is the standard choice in generative modeling and also used by \citet{ShenMeinshausen2025}, whereas the latter is in line with our boundedness assumptions for consistency in Section \ref{sec:guarantees}, but does not produce noticeably different results from Gaussian noise in practice.

\begin{algorithm}[t]
\caption{Sufficiently Reduced Distributional Regression}
\label{alg:srdr}

\KwIn{
\begin{tabular}{@{}ll@{}}
$(q,s)$ 
& dimensions for $e$ and $\eta$ \\

$\theta_e, \theta_d$ 
& hyperparameters (width, depth) for $e$, $d$ \\

$(X_1, Y_1), \dots, (X_n, Y_n)$ 
& training sample \\

$N$ 
& number of samples for energy score approximation
\end{tabular}
}

\KwOut{estimated functions $\hat{e}, \hat{d}$}

\textbf{Initialization:} initialize $e_1, d_1$

\For{$t = 1, 2\dots,$ training iterations, until a stopping criterion at $t = T$ is reached}{
    calculate dimension reduction $e_t(X_1), \dots, e_t(X_n)$\;
    generate noise $\eta_{i,j}$, $i = 1, \dots, n$, $j = 1, \dots, N$\;
    generate samples $d_t(e_t(X_i),\eta_{i,j})$ for $i = 1, \dots, n$, $j = 1, \dots, N$, and compute
    \[
        \begin{aligned}
            & \frac{1}{n}\sum_{i=1}^n \Big(\frac{1}{N}\sum_{j=1}^N \|d_{t}(e_{t}(X_i),\eta_{i,j}) - Y_i\|  - \frac{1}{2N(N-1)}\sum_{k,j=1}^N
            \|d_{t}(e_{t}(X_i),\eta_{i,j}) - d_{t}(e_{t}(X_i),\eta_{i,k})\|\Big);
        \end{aligned}\;
    \]
    perform a gradient step on the parameters of $e_t, d_t$ with the above loss function\;
}
\Return $\hat{e} = e_T$, $\hat{d} = d_T$
\end{algorithm}

\subsection{Dimension Selection} \label{sec:dimension_selection}

In this section, we give guidance on selecting the dimensions $q$ of $e(X)$ and $s$ of the noise $\eta$. In general, dimension selection for nonlinear sufficient dimension reduction is more involved than for linear dimension reduction. The identifiable object is the dimension reduction $\sigma$-field, whose ``size'' is not directly related to the dimension of the vector $e(X) \in \mathbb{R}^q$; indeed, as demonstrated by Lemma \ref{lem:generative_sufficiency}, even one-dimensional $e(X)$ can be sufficient if $e,d$ are allowed to be arbitrary measurable functions.

The dimension $s$ of the noise $\eta$ is generally a less sensitive parameter in practice, and it intuitively measures the intrinsic dimension of the outcome $Y$ given $X$; a reasonably large $s$, such as the default $s = 100$ in \citet{ShenMeinshausen2025}, is often sufficient for most applications. For the dimension of $e(X)$, since our target function directly measures predictive accuracy of the samples $d(e(X),\eta)$ as a probabilistic forecast for $Y|X$, the selection of $q$ and $s$ can be integrated into the standard workflow of neural network hyperparameter tuning on a hold-out validation data set $(X_i', Y_i'), i = 1, \dots, N'$. That is, for a grid $q_1 < \dots < q_u$ and $s_1 < \dots < s_v$, one approximates
\[
    (q_k, s_l) \mapsto  \frac{1}{N'}\sum_{i=1}^{N'} \mathrm{ES}(\hat{P}_{X_i'}^{[q_k, s_l]}, Y_i'),
\]
where $\hat{P}_{X_i'}^{[q_k, s_l]}$ denotes the distribution of $\hat{d}(\hat{e}(X_i'),\eta_i)\mid \hat{e}(X_i')$, with $\hat{e}(x) \in \mathbb{R}^{q_k}$ and $\eta_i \in \mathbb{R}^{s_l}$. A reasonable choice is then to select small $q_k, s_l$ for which the above mapping does not further decrease for $k' > k, l' > l$. This is the elbow pattern seen in Figure \ref{fig:synthetic_sparse}: the out-of-sample score decreases up to the sufficient dimension $q^*$ and flattens afterwards, since a reduction of dimension $q > q^*$ is still sufficient and, by Remark \ref{rem:larger_dimension}, the consistency result continues to apply. An alternative criterion, based on a mixture loss over several dimensions, is described in Appendix \ref{app:mixture}.

\subsection{Extension to Multi-Environment Data} \label{sec:multi_env_data}

In some applications, the data is grouped into multiple heterogeneous environments $W \in \{1, \dots, K\}$, which may correspond to different data sources or populations, and one is interested in obtaining a dimension reduction that is sufficient across all of them. Such settings have been first considered for linear dimension reduction by \citet{Chiaromonte2002}, and more recently by \citet{JiaoEtAl2024,GeZhouHuang2025} for nonlinear models based on neural networks. Formally, we observe $(X, Y, W)$, and are interested in finding $e(X)$ such that
\[
    Y \indpt X \mid (e(X), W),
\]
meaning that sufficiency also holds conditional on the environment. Similarly to the single-environment case, this is equivalent to the existence of a noise variable $\eta$, which is independent of $X$ and $W$, and $K$ functions $d_1, \dots, d_K$ such that
\[
    (X,Y,W) = (X, d_W(e(X),\eta), W)
\]
almost surely. That is, there is a shared dimension reduction $e(X)$ that is valid for all environments, but the $d_W$ depend on the environment $W$, corresponding to partial sufficient dimension reduction as introduced by \citet{Chiaromonte2002}.

An advantage of the scoring-rule formulation for estimation of sufficient dimension reductions is that it easily extends to this multi-environment setting. For a strictly proper $S$, we have
\[
    \mathbb{E}\left[\sum_{k=1}^K S(\cond{Y_k}{e(X_k)}, Y_k)\right] \geq  \mathbb{E}\left[\sum_{k=1}^K S(\cond{Y_k}{X_k}, Y_k)\right],
\]
for $(X_k, Y_k) \sim \cond{(X,Y)}{W=k}$, with equality if and only if $\cond{Y_k}{e(X_k)} = \cond{Y_k}{X_k}$ for $k = 1, \dots, K$. So minimizing the aggregated error with a proper scoring rule over a joint dimension reduction and separate functions $d_k$ recovers a partial sufficient dimension reduction, ensuring independence of $Y$ and $X$ conditional on $(e(X), W)$. On a finite sample with the energy score, this corresponds to the minimization problem
\[
    (\hat{e}, \hat{d}_1, \dots, \hat{d}_K) \in \argmin_{e,d_1, \dots, d_K} \sum_{k=1}^K \sum_{i=1}^{n_k} \mathrm{ES}(P_{e,d_k,X_{i,k}}, Y_{i,k}),
\]
where $n_k$ is the sample size in environment $k$, and the energy score is approximated by sampling from the noise distribution as before. Algorithm \ref{alg:srdr} can be applied analogously in this setting. Consistency of the resulting estimator is established in Appendix \ref{sec:proof_multi}; the transfer learning experiment in Appendix \ref{app:transfer} illustrates the extension on real data.

\subsection{Extension to Classification Problems} \label{sec:classification}
Although SRDR is formulated for Euclidean responses, it can be applied to classification by representing the class label as a one-hot vector. For an $m$-class problem, the response is $Y \in \{v_1,\ldots,v_m\}\subset \mathbb{R}^m$, where $v_k$ is the \(k\)th standard basis vector. The conditional distribution of $Y \mid X$ is then characterized by the conditional class-probability vector
\[
    p(x) = \bigl(\mathbb{P}(Y=v_1\mid X=x),\ldots, \mathbb{P}(Y=v_m\mid X=x)\bigr).
\]
Thus, a dimension reduction $e(X)$ is sufficient for classification if and only if
\[
    \mathbb{P}(Y=v_k\mid X) = \mathbb{P}(Y=v_k\mid e(X)) \quad \text{almost surely}, \qquad k=1,\ldots,m.
\]

During implementation, we use the same SRDR objective, with a softmax output layer on $d$ that constrains the generator output to the probability simplex,
\[
    d(e(X),\eta)\in \Delta^{m-1}, \qquad
    \Delta^{m-1} = \left\{s\in[0,1]^m:\sum_{k=1}^m s_k=1\right\}.
\]
The energy score is computed between the generated simplex-valued vector and the observed one-hot response. This gives a continuous formulation of the classification problem: the model is trained using the same proper scoring rule objective, and $d$ generates probability vectors rather than discrete labels.

This formulation is consistent with the overall framework of SRDR, and strict propriety of the energy score implies that SRDR generated samples $\hat{d}(\hat{e}(X),\eta)$ should approximate the distribution of $Y \in \{v_1, \dots, v_m\}$, i.e., have roughly the same class probabilities. Although generated samples do not fall strictly within the discrete set $\{v_1, \dots, v_m\}$, their components are typically very close to $0$ or $1$ in practice. A unique predicted class label can be obtained by mapping $\mathbb{E}[\hat{d}(\hat{e}(x),\eta)]$ to a hard assignment vector $\hat{v}(x)$ whose non-zero entry corresponds to $\argmax_{k=1,\dots,m} \mathbb{E}[\hat{d}(\hat{e}(x),\eta)]$. In our empirical applications, this classification extension of SRDR is competitive with state-of-the-art methods for sufficient dimension reduction in classification problems. Moreover, the implementation is simple in the sense that essentially the same setup as for the regression SRDR can be applied, which is an advantage over more complex approaches for generative models in classification, such as \citet{Kusner2016}. The reason why the application to classification settings is simpler is that the energy score naturally supports both discrete and continuous variables, whereas methods based on likelihood (or KL-divergence) may encounter difficulties in training when the probabilities of certain outcomes converge to zero or one.

\section{Statistical Guarantees} \label{sec:guarantees}

\subsection{Assumptions} \label{sec:assumptions}

We prove consistency of our method for the estimation of the conditional distribution functions of $Y$ given $X$, showing that $P_{\hat{e}, \hat{d}, X}$ converges to $\cond{Y}{X}$. This also implies consistency of the dimension reduction, since the condition $\cond{Y}{e(X)} = \cond{Y}{X}$ characterizes sufficiency.

Our assumptions can be broadly grouped into model assumptions, regularity assumptions, and assumptions on estimation and approximation. The first assumption below is our model assumption; without any restrictions on the functions $e^*, d^*$, it is strictly speaking not an assumption, because there always exist functions such that the given representation for $Y$ holds.

\begin{assumption} \label{assumption:modelMULTIY}
There exist $e^*\colon \mathbb{R}^p \rightarrow \mathbb{R}^q$, $d^* \colon \mathbb{R}^{q+s} \rightarrow \mathbb{R}^m$ and $\eta \in \mathbb{R}^s$ independent of $X$ so that the outcome variable $Y \in \mathbb{R}^m$ satisfies
\[
    Y = d^*(e^*(X),\eta),
\]
where $\eta$ follows a known reference distribution. The training data consists of independent and identically distributed realizations $(X_1, Y_1), \dots, (X_n, Y_n)$ from this model.
\end{assumption}

The distribution of $\eta$ in Assumption \ref{assumption:modelMULTIY} is not further specified for our theoretical result, but one can think of it as the uniform distribution on some interval $[-B, B]$. The next group of assumptions are regularity assumptions on the distribution of $X$ and on $e^*, d^*$.

\begin{assumption} \label{assumption:supportMULTIY}
The support $\mathcal{X} \times \mathcal{Z}$ of $(X,\eta)$ is contained in a bounded set $[-B_x, B_x]^{p+s}$.
\end{assumption}

\begin{assumption} \label{assumption:encoderMULTIY}
The function $e^*$ has Lipschitz continuous components,
\[
    |e_j^*(x) - e_j^*(x')| \leq L_e\|x-x'\|, \ x, x' \in [-B_x, B_x]^p, \ j = 1, \dots, q.
\]
\end{assumption}

\begin{assumption} \label{assumption:decoderMULTIY}
The function $d^*$ has Lipschitz continuous components,
\[
    |d_j^*(z) - d_j^*(z')| \leq L_d\|z-z'\|, \ \ z, z' \in \mathbb{R}^{q+s}, \ j = 1, \dots, m.
\]
\end{assumption}

Assumptions \ref{assumption:supportMULTIY}, \ref{assumption:encoderMULTIY}, and \ref{assumption:decoderMULTIY} together imply that $Y$ is bounded almost surely, and we let $B_y$ be a constant such that
\[
    B_y \geq \max\{|d_j^*(e^*(x),\eta)|\colon (x,\eta) \in [-B_x, B_x]^{p+s}, j = 1, \dots, m\}.
\]
Such a boundedness assumption is comparable to existing work, see, e.g., \citet[Assumption (A1)]{Zhou2023}, \citet[Condition 1]{Song2025}, \citet{chen2026error}, and \citet{huang2026theoretical}. It also appears implicitly for certain test functions in \citet{TangLi2025}, e.g., for $g(y,t) = 1\{y \leq t\}$, which for bounded $t$ only characterize the conditional distributions of $Y$ in the case of bounded support.

Finally, we impose assumptions on the model class and the estimator. We define the class of ReLU networks as functions $f$ such that
\[
    f = L_{\ell + 1} \circ \sigma \circ L_{\ell} \circ \dots \circ L_2 \circ \sigma \circ L_1,
\]
where $\ell$ is the number of hidden layers,
\[
    L_j \colon \mathbb{R}^{r_{j-1}} \rightarrow \mathbb{R}^{r_j}, \ L_j(x) = A_jx + b_j, \ A_j \in \mathbb{R}^{r_j \times r_{j-1}}, \ b_j \in \mathbb{R}^{r_j}, \ j = 1, \dots, \ell + 1,
\]
with input dimension $p = r_0$, output dimension $q = r_{\ell + 1}$, and $\sigma(x) = \max(0, x)$, to be interpreted componentwise for a vector $x$. Similarly to \citet{TangLi2025}, denote by $\mathcal{F}(p, \ell, (r_1, \dots, r_{\ell}), q, B)$ the class of neural networks with input dimension $p$, output dimension $q$, $\ell$ hidden layers each with $r_j$ neurons, and weight and bias entries bounded to $[-B,B]$ for some $B \in [0,\infty]$. If $r_1 = \dots = r_{\ell} = r$, we abbreviate the definition as $\mathcal{F}(p, \ell, r, q, B)$, and we omit the bound $B$ if it is not relevant.

For consistency, we need that $\hat{e}_n, \hat{d}_n$ are a minimizer of the energy score on the training data. While this is admittedly a strong assumption, it is common in similar approaches in the literature, such as in \citet{Zhou2023, TangLi2025, Song2025}. We enforce the neural network output to be bounded, via the operator $T_B x = (\max(-B, \min(B,x_j)))_{j=1}^r$ for vectors $x \in \mathbb{R}^r$ and $B > 0$. We simplify the proof by assuming that within each of the two networks for $e$ and $d$, all hidden layers have the same number of neurons, and denote by $P_{e,d,X,B_y}$ the conditional distribution of $T_{B_y}d(e(X),\eta)$ given $X$.

\begin{assumption} \label{assumption:estimatorMULTIY}
The estimator $(\hat{e}, \hat{d})$ satisfies
\begin{align}
    & (\hat{e}_n, \hat{d}_n) \in \argmin_{(e,d) \in \mathcal{M}} \sum_{i=1}^n \mathrm{ES}(P_{e,d,X_i,B_y}, Y_i),
\end{align}
for the model class
\[
    \mathcal{M} = \{(e,d)\colon e \in \mathcal{F}(p, \ell_1, r_1, q, B_w), \ d \in \mathcal{F}(q+s, \ell_2, r_2, m, B_w)\},
\]
and for some bound $B_w > 0$ on the neural network parameters.
\end{assumption}

Our consistency result requires that the width and/or depth of the neural networks diverge at a not too fast rate, as given in the assumption below.

\begin{assumption} \label{assumption:parametersMULTIY}
The bound $B_w$ is non-decreasing in $n$ with $B_w = \mathcal{O}(n^{\gamma_B})$ for some $\gamma_B > 0$, and the neural network hyperparameters satisfy
\begin{align*}
    & \ell_1 = 12L_1 + 14, & \ r_1 = \max\{4p\lfloor N_1^{1/p}\rfloor + 3p, 12qN_1+8q\}, \\
    & \ell_2 = 12L_2 + 14, & \ r_2 = \max\{4(q+s)\lfloor N_2^{1/(q+s)} \rfloor + 3(q + s), 12N_2 + 8m\},
\end{align*}
where $L_1, N_1, L_2, N_2$ are non-decreasing in $n$ with limits
\begin{align*}
    \lim_{n\rightarrow \infty} \min(L_1N_1, L_2N_2) = \infty, \quad 
        \lim_{n\rightarrow \infty} \frac{\log(n)(L_1N_1^2 + L_2N_2^2)(L_1\log(N_1) + L_2\log(N_2))}{n} = 0.
\end{align*}
\end{assumption}

The expressions for $\ell_i$ and $r_i$ are the ones in the neural network approximation result of Lemma \ref{lem:nn_approximation}, which we apply in the proof. The assumption therefore amounts to letting the depth and the width of the two networks grow with $L_i$ and $N_i$.

\subsection{Theorem and Discussion} \label{sec:thmdiscuss}

Under the assumptions from the previous section, our estimator is consistent in the energy distance, also implying asymptotic recovery of a sufficient dimension reduction.

\begin{theorem} \label{thm:consistencyMULTIY}
If the assumptions from Section \ref{sec:assumptions} hold, then
\[
    \lim_{n \rightarrow \infty} \mathbb{E}\left[\int_{\mathcal{X}}\mathcal{D}^2(P_{\hat{e}_n, \hat{d}_n, x, B_y}, \cond{Y}{X = x}) \, dP_{X}(x) \right] = 0.
\]
\end{theorem}

Our assumptions are similar to those of \citet{TangLi2025}, and our proof follows theirs, with two differences discussed in Section \ref{sec:test_functions}: the test functions $g(y, t)$ have $t \in \mathbb{R}^m$, and their expectations under our model are of the form $\mathbb{E}_{\eta}[\cos(t^{\top}T_{B_y}d(e(x),\eta))]$, and analogously for the sine, so that their approximation results for neural networks do not apply directly. \citet[Theorem 3, Corollary 3]{XuEtAl2025} prove convergence in the 2-Wasserstein distance under continuity, log-concavity and log-convexity assumptions on the conditional density of $Y$. We show convergence in the energy distance and do not require the existence of a density; for instance, the distribution of $Y$ can have atoms through constant regions of $d^*$. Since our assumptions ensure that both $\cond{Y}{X=x}$ and the estimated conditional distributions have compact support, Proposition 16 of \citet{ModesteDombry2024} implies that our estimator also converges in the $1$-Wasserstein distance, meaning that
\[
   \mathbb{E}\left[ \int_{0}^1|\hat{F}_{X,n}^{-1}(\alpha) - F_{X}^{-1}(\alpha)| \, d\alpha\right] \rightarrow 0, \ n \rightarrow \infty,
\]
for univariate outcomes, where $\hat{F}_{X,n}^{-1}$ denotes the quantile function of $P_{\hat{e}_n, \hat{d}_n, X, B_y}$. \citet{chen2026error,huang2026theoretical} derive error bounds and convergence rates for engression, that is, for a full generator $d(X,\eta)$ in which the noise enters together with the covariate; the main result of \citet{chen2026error} also implies consistency of engression when the data are generated as $d^*(e^*(X),\eta)$. Our theorem concerns the reduced model, where $\eta$ enters only after $e(X)$, so that consistency of the conditional distributions is tied to sufficiency of $\hat{e}_n(X)$. It covers multivariate outcomes and the multi-environment setting of Section \ref{sec:multi_env_data}, but it gives consistency without a rate.

We close this section with remarks on extensions and refinements of our consistency result.

\begin{remark} \label{rem:larger_dimension}
To ensure convergence, it is not necessary to know the exact dimensions $q$ for $e^*$ and $s$ for the noise $\eta$. If our method is applied with $q' > q$ and $s' > s$, then our convergence result still applies, because $e^*$ and $\eta$ can always be extended to length $q'$ or $s'$, while all assumptions for the theorem continue to hold.
\end{remark}

\begin{remark}
Our consistency result can easily be adapted to estimators trained with other kernel scores, such as with the Gaussian or Laplace kernel, for which the score and the divergence function have a similar mixture representation as the energy score \citep[see][Section 3.1.2]{WaghmareZiegel2025}. An attractive property of the energy score is that the estimator is preserved under translation, scaling, and multiplication of $Y$ by a rotation matrix, making it a neutral choice if there is no reason to put weight on particular regions on the domain of $Y$.
\end{remark}

\begin{remark}
For univariate $Y$, one can show that with a similar choice of optimal neural network hyperparameters as in Theorems 2 and 3 of \citet{TangLi2025}, our estimator achieves the convergence rate $\mathcal{O}(\log(n)n^{-2/(2+p)})$ in a slightly modified distance. Namely, the energy distance $\mathcal{D}^2$ needs to be replaced by a suitable frequency-truncated version of it, defined as
\[
    \mathcal{D}^2_{\tau}(P,Q) = \frac{1}{2\pi}\int_{-\tau}^{\tau}\frac{|\phi_P(t) - \phi_Q(t)|^2}{t^2} \, dt,
\]
with $0 < \tau < \infty$, and analogously, the energy score for training is modified as
\[
    \mathrm{ES}_{\tau}(P,y) = \frac{1}{2\pi}\int_{-\tau}^{\tau}\frac{|\phi_P(t) - \exp(ity)|^2}{t^2} \, dt.
\]
The frequency truncated energy score is still a strictly proper scoring rule, for the restricted family of distributions for which the moment-generating function exists in a neighborhood of zero. Indeed, such distributions are characterized by the moment generating function, which is in turn characterized by the derivatives of the characteristic function at zero. The more detailed asymptotic analysis in the concurrent work by \citet{TanLiXue2026} reveals that convergence rates also hold in the usual energy distance, and that the rates can scale with the intrinsic dimension of $X$ rather than with the ambient dimension $p$.
\end{remark}

\begin{remark}
In a multi-environment setting (Section \ref{sec:multi_env_data}), our method remains consistent under similar assumptions as for Theorem \ref{thm:consistencyMULTIY}. See the detailed result and proof in Appendix \ref{sec:proof_multi}.
\end{remark}

\section{Simulations} \label{sec:simulations}

In this section, we evaluate SRDR in a series of synthetic experiments designed to assess both distributional prediction and the recovery of low-dimensional sufficient representations. Results are averaged over Monte Carlo replications, with error bars representing the standard deviation across replications.

\subsection{Comparison with engression} \label{sec:Engression}

We conduct several synthetic experiments to evaluate the ability of SRDR to identify a low-dimensional sufficient representation in high-dimensional nonlinear regression. More simulations are presented in Appendix \ref{app:Engression}.

The experiments in this section illustrate the benefits of dimension reduction for improving the performance of distributional regression. In the following simulation, we consider a high-dimensional sparse dependence structure, where covariates $X \in \mathbb{R}^{80}$ are generated from a standard multivariate Gaussian distribution. For the first simulation, a four-dimensional $Y$ is generated as
\[
Y = \mu(X_{1:3}) + \sigma(X_{1:3})\,\varepsilon, \qquad \varepsilon \sim \mathcal{N}(0,I_4),
\]
with nonlinear functions
\[ 
\mu(X_{1:3})= 
\begin{pmatrix} 
    \sin(X_1+0.5X_2)\\ 
    X_1X_3\\ 
    \cos(X_2-X_3)+0.25X_1\\ 
    \tanh(X_1X_2)+0.5X_3 
\end{pmatrix}, 
\qquad
\sigma(X_{1:3})= 
\begin{pmatrix} 
    0.10+0.45\,\mathrm{sigmoid}(X_1)\\ 
    0.10+0.35\,\mathrm{sigmoid}(X_2-X_3)\\ 
    0.12+0.40\,X_3^2/(1+X_3^2)\\ 
    0.10+0.30\,\mathrm{sigmoid}(X_1+X_2+X_3) 
\end{pmatrix}. 
\]
For the second simulation, we replace $X_{1:3}$ by $Z_{1:3}$, where
\[ 
Z_1=X_1^2+0.5\sin(X_2), \qquad Z_2=X_2X_4, \qquad Z_3=\cos(X_3-X_5)+0.25X_4^2. 
\]
In both simulations, the sufficient dimension is $q^* = 3$. We call the first setting ``sparse-linear'', since the sufficient reduction $X_{1:3}$ is linear in $X$, and the second setting ``sparse-nonlinear'', since the sufficient reduction $Z_{1:3}$ is a nonlinear function of $X_{1:5}$. Correspondingly, we train SRDR with both linear and nonlinear dimension reduction maps $e$, i.e., $e(x) = \Theta x$ with $\Theta \in \mathbb{R}^{q \times p}$, and a nonlinear $e$ parametrized by neural networks; the dimension of $e$ ranges over $q\in\{1,2, \ldots,6\}$. Results are averaged over 10 random replicates. Performance is evaluated using the energy score computed on an independent test set. As a baseline we train an engression model without any dimension reduction, using the same training budgets and hyperparameters.

\begin{figure}[!ht]
\centering
\begin{subfigure}[t]{\textwidth}
    \centering
    \includegraphics[width=\textwidth]{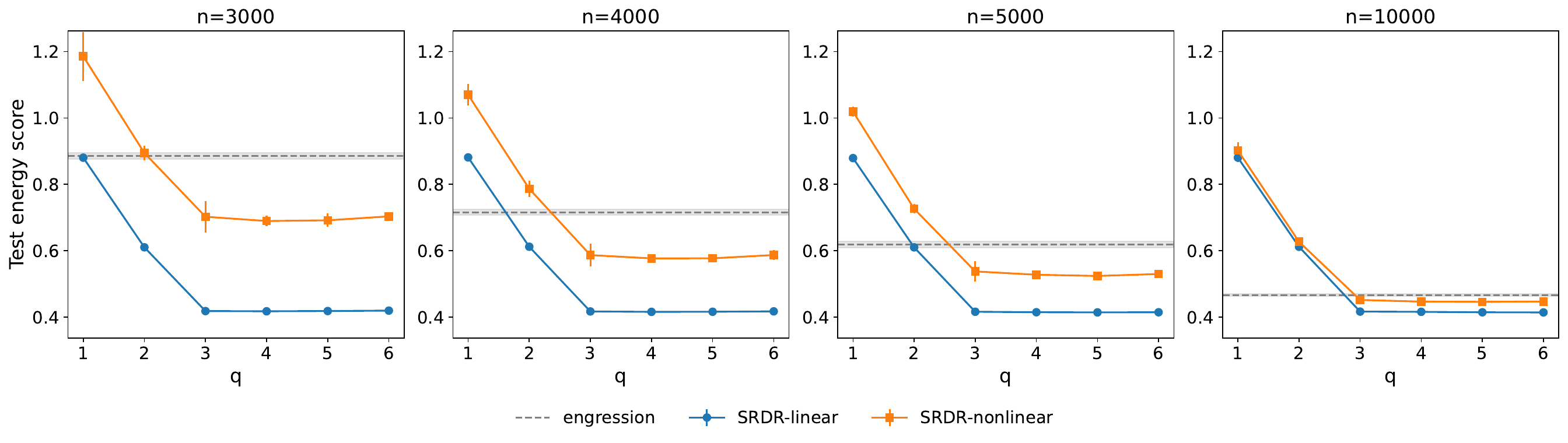}
    \caption{Sparse-linear.}
    \label{fig:sparse_linear}
\end{subfigure}

\vspace{0.5em}

\begin{subfigure}[t]{\textwidth}
    \centering
    \includegraphics[width=\textwidth]{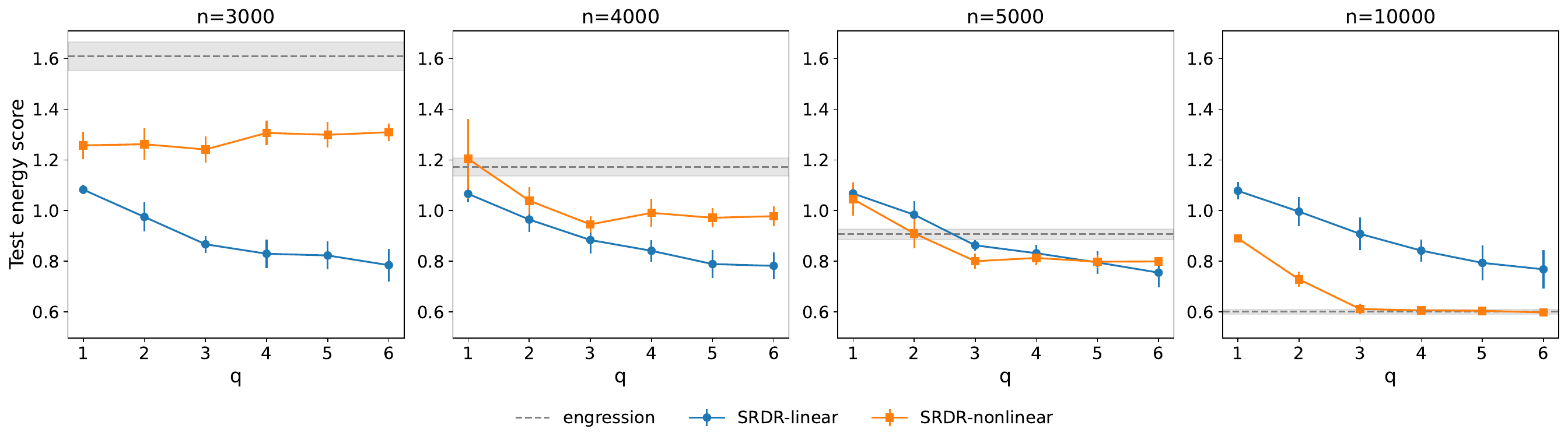}
    \caption{Sparse-nonlinear.}
    \label{fig:sparse_nonlinear}
\end{subfigure}

\caption{Test energy score in the sparse-linear and sparse-nonlinear settings of Section \ref{sec:Engression}. SRDR-linear and SRDR-nonlinear refer to SRDR with a linear and a nonlinear dimension reduction map, respectively. Engression does not incorporate dimension reduction, whence its test scores do not depend on $q$.}
\label{fig:synthetic_sparse}
\end{figure}

We vary the training sample size over \(n \in \{3000,4000,5000,10000\}\) and report the results in Figure~\ref{fig:synthetic_sparse}. Overall, both SRDR variants improve upon engression at smaller sample sizes and remain competitive at larger sample sizes when \(q \ge q^*\). Both variants of SRDR, with linear and nonlinear dimension reduction, also display a clear elbow pattern: the test energy score decreases as \(q\) increases up to \(q^*=3\), and then stabilizes for larger \(q\). This behavior is consistent with the discussion in Section~\ref{sec:thmdiscuss} and suggests that SRDR learns the response-relevant sufficient reduction.

In Figure~\ref{fig:sparse_linear}, where the true sufficient reduction is linear in \(X\), both SRDR variants identify \(q^*\), while SRDR-linear performs better due to correct specification. In contrast, Figure~\ref{fig:sparse_nonlinear} shows that when the true reduction is nonlinear, SRDR with a linear dimension reduction map is misspecified and struggles to recover \(q^*\).

\subsection{Comparison with Competing Methods for Sufficient Dimension Reduction} \label{sec:comp}

We next compare SRDR with several existing nonlinear SDR methods using a benchmark model proposed by \citet{Kapla2022}. The competing methods include GSIR \citep{Lee2013}, GMDDNet \citep{Chen2024}, StoNet \citep{Liang2022}, BENN \citep{TangLi2025}, and GenSDR \citep{XuEtAl2025}. 

The data-generating mechanism is a one-dimensional mean model of the form
\begin{equation}
    Y = \cos(b_1^\top X) + \varepsilon, \qquad \text{where }
X \sim \mathcal{N}_{20}(0,\Sigma), \quad \Sigma_{ij} = 0.5^{|i-j|}.
\label{eq:m1}
\end{equation}
A true sufficient reduction is \(U = b_1^\top X\) with \(b_1 = \frac{1}{\sqrt{6}}(1,1,1,1,1,1,0,\ldots,0)^\top.\) The noise term \(\varepsilon\) follows a generalized normal distribution with shape parameter \(c=0.5\) and variance \(0.25\). We call this setting heavy-tailed. It is challenging because the predictors are strongly correlated and the noise distribution has heavy tails.

We evaluate performance across sample sizes $n \in \{1000, 2000, \ldots, 8000\}$ with 10 replicates for each configuration, using a fixed test set of 2000 observations. Performance is measured by $\mathrm{dCor}(U, \hat e(X))$, the distance correlation \citep{Szekely2007} between the true sufficient reduction and the low-dimensional representation computed on the test set. Figure~\ref{fig:heavy_tailed_dcor} reports the average results over the 10 replicates.

\begin{figure}[!ht]
    \centering
    \begin{subfigure}{0.48\textwidth}
        \centering
        \includegraphics[width=\textwidth]{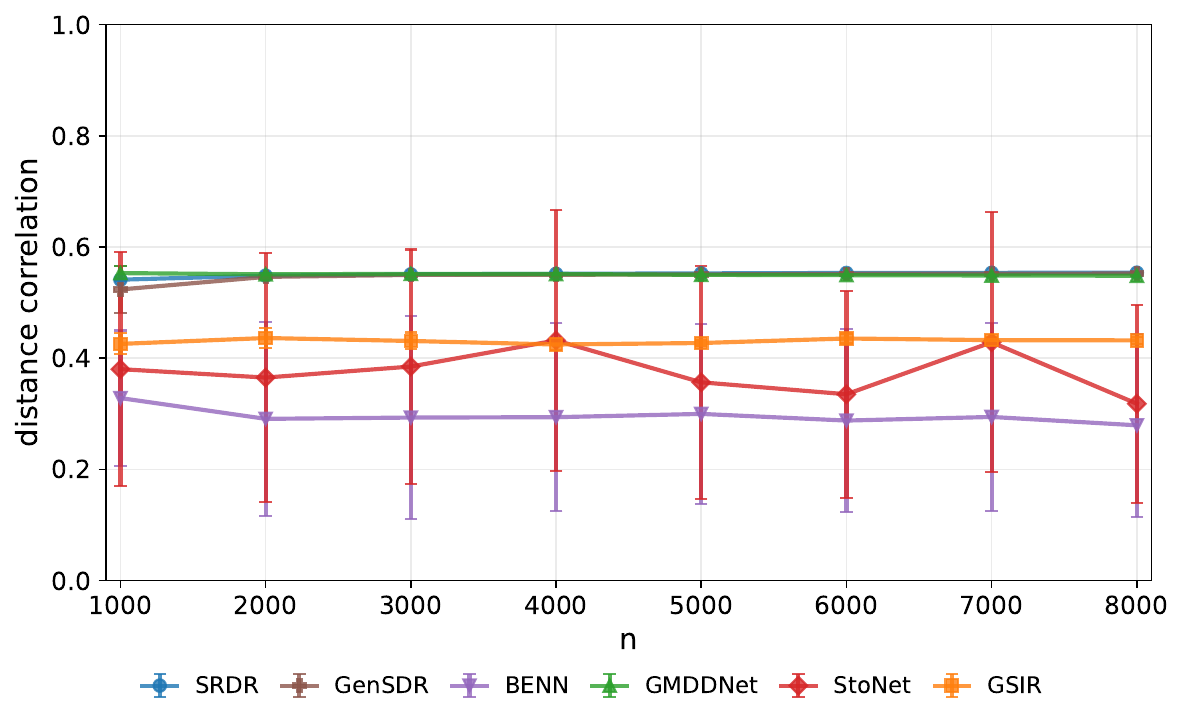}
        \caption{Distance correlation with \(U\)}
        \label{fig:heavy_tailed_dcor_u}
    \end{subfigure}
    \hfill
    \begin{subfigure}{0.48\textwidth}
        \centering
        \includegraphics[width=\textwidth]{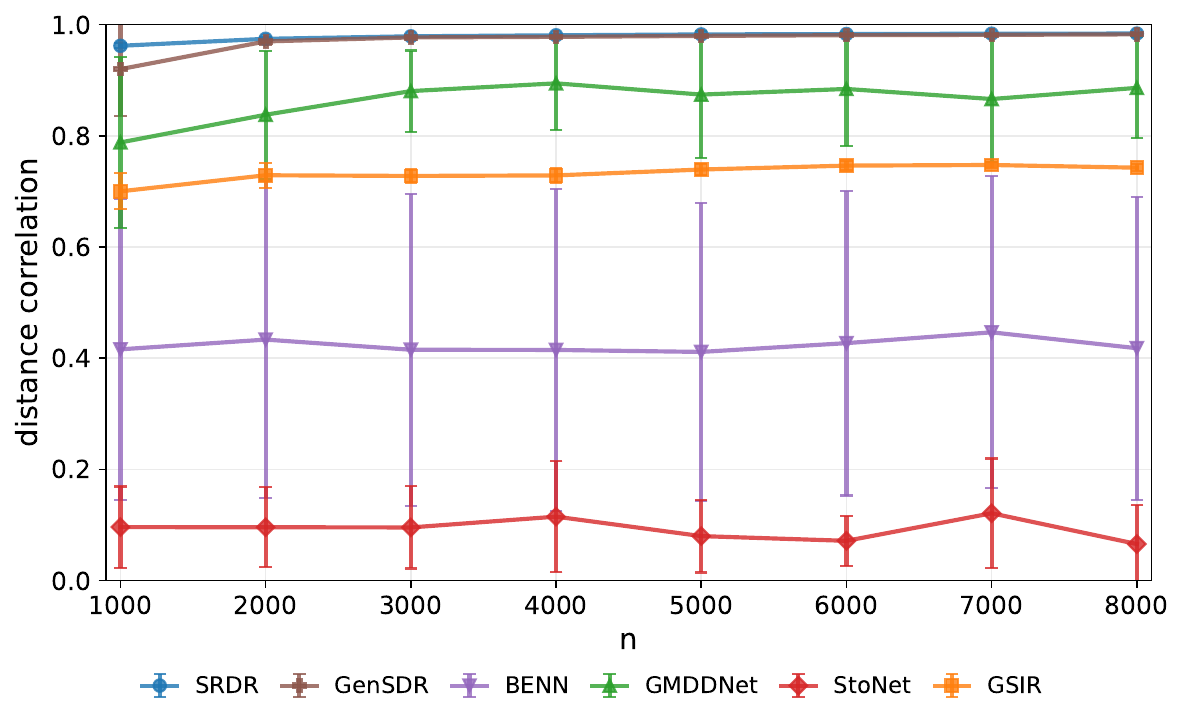}
        \caption{Distance correlation with \(\cos(U)\)}
        \label{fig:heavy_tailed_dcor_cos}
    \end{subfigure}
    
    \caption{Distance correlation of the learned representation with two reference sufficient reductions, $U = b_1^{\top}X$ and $\cos(U)$, in the heavy-tailed setting \eqref{eq:m1} of Section \ref{sec:comp}.}
    \label{fig:heavy_tailed_dcor}
\end{figure}

In Figure~\ref{fig:heavy_tailed_dcor_u}, GMDDNet, GenSDR and SRDR reach similar distance correlation with \(U\). This measure cannot separate the methods well, because the response depends on \(X\) only through \(\cos(b_1^{\top}X)\). A nonlinear SDR method that learns a reduction close to \(\cos(U)\) recovers a valid sufficient reduction, but its distance correlation with \(U\) is below $1$, so the attainable maximum of this measure is not $1$. Figure~\ref{fig:heavy_tailed_dcor_cos} therefore reports the distance correlation with \(\cos(U)\), the actual response signal. Here GenSDR and SRDR achieve by far the highest values and the smallest variability across replicates, so the two generative methods consistently capture the sufficient structure.

This simulation shows that for nonlinear SDR, distance correlation with a fixed reference representation should be interpreted with caution. It measures geometric alignment with the chosen reference reduction rather than sufficiency itself, and its actual range in specific settings does not necessarily reach up to $1$, hiding differences for methods close to the actual upper bound.

\section{Applications} \label{sec:app}
We next evaluate our method on sufficient reduction tasks using three real datasets: two nonlinear regression problems and one classification problem. For all applications, we use a fixed training-test split and evaluate predictive performance on the held-out test set. Results are averaged over 10 repeated model fits on the same split, with error bars representing the standard deviation across replicates. Network architectures and training details for all methods are given in Appendix \ref{app:implementation}.

\subsection{Application to CT Slice Localization} \label{sec:ct}
We evaluate SRDR on the CT slice localization dataset \citep{GrafCavallaro2011}. The dataset consists of $53,500$ computed tomography (CT) images collected from 74 patients. Each observation is represented by a 384-dimensional feature vector constructed from two histograms in polar space: the spatial distribution of bone structures and the location of air inclusions inside the body. The response is the relative axial location of the slice within the CT volume.

We compare SRDR with the competing network-based SDR methods, using full-feature engression as an additional benchmark. For GMDDNet, StoNet, and BENN, we fit a downstream prediction network on the learned reduction, trained with $L_2$ loss when reporting MSE and with $L_1$ loss when reporting MAE. SRDR, GenSDR, and engression directly produce distributional predictions. For MSE evaluation, their point predictions are given by the predictive mean. Distributional prediction is evaluated using the continuous ranked probability score (CRPS), which reduces to the MAE for point forecasts as issued by GMDDNet, StoNet, and BENN. Although StoNet is a stochastic network, its learned reduction is deterministic given the fitted parameters, since randomness enters only through the training procedure, so it is evaluated as a point-prediction method.

\begin{figure}[!t]
    \centering
    \begin{subfigure}{0.48\textwidth}
        \centering
        \includegraphics[width=\textwidth]{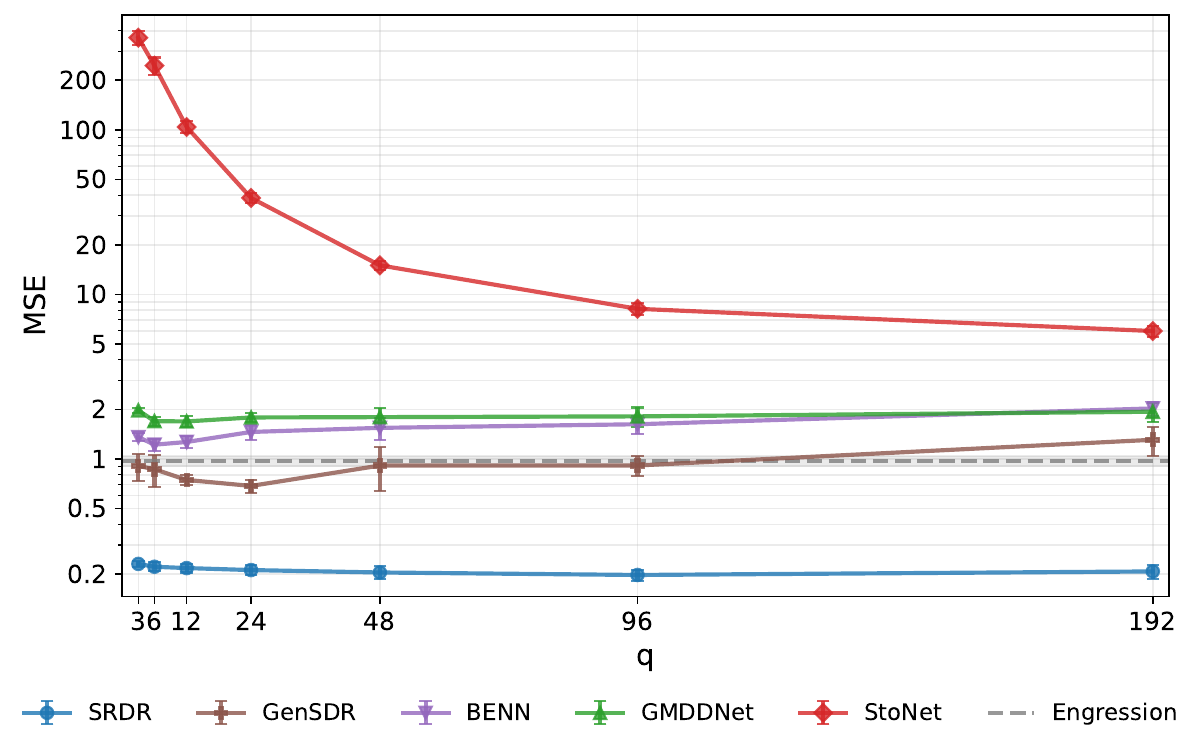}
        \caption{MSE}
        \label{fig:ct_mse}
    \end{subfigure}
    \hfill
    \begin{subfigure}{0.48\textwidth}
        \centering
        \includegraphics[width=\textwidth]{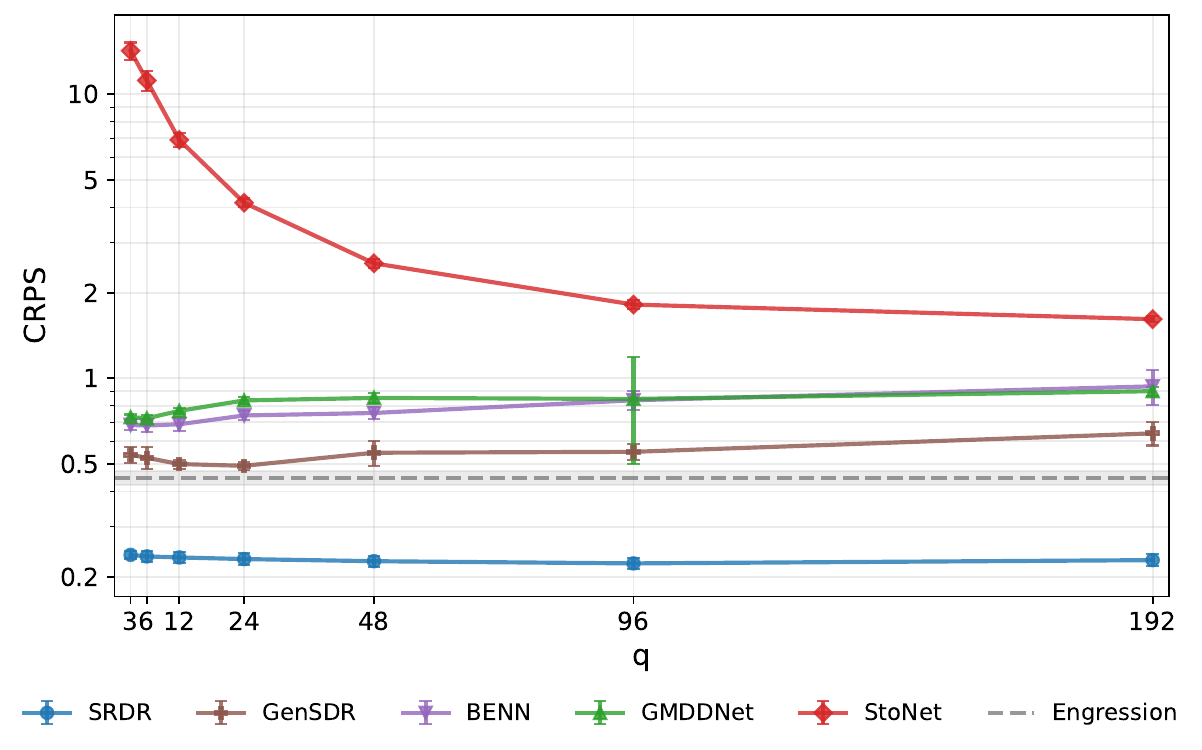}
        \caption{CRPS}
        \label{fig:ct_dist}
    \end{subfigure}
    
    \caption{CT localization results by representation dimension $q$. Panel \ref{fig:ct_mse} reports MSE, and panel \ref{fig:ct_dist} reports CRPS for generative methods and MAE for deterministic methods.}
    \label{fig:ct}
\end{figure}

The results for \(q \in \{3,6,12,24,48,96,192\}\) are reported in Figure~\ref{fig:ct}, with MSE in Figure~\ref{fig:ct_mse} and CRPS, or MAE for the deterministic methods, in Figure~\ref{fig:ct_dist}; the latter two are directly comparable because CRPS reduces to the absolute error for a degenerate predictive distribution. SRDR is the best-performing SDR method across \(q\), with stable performance even at small representation dimensions, while additionally providing a full predictive distribution. The advantage of SRDR over the other dimension reduction approaches highlights that if a dimension reduction is used for the downstream task of prediction, it is beneficial to estimate $e$ and $d$ jointly, like in the SRDR framework, rather than one after another. Moreover, the lower CRPS of SRDR and GenSDR highlights the advantage of distributional predictions, and, hence, generative approaches for dimension reduction, over deterministic point forecasts.

\begin{figure}[!t]
    \centering
    \includegraphics[width=\textwidth]{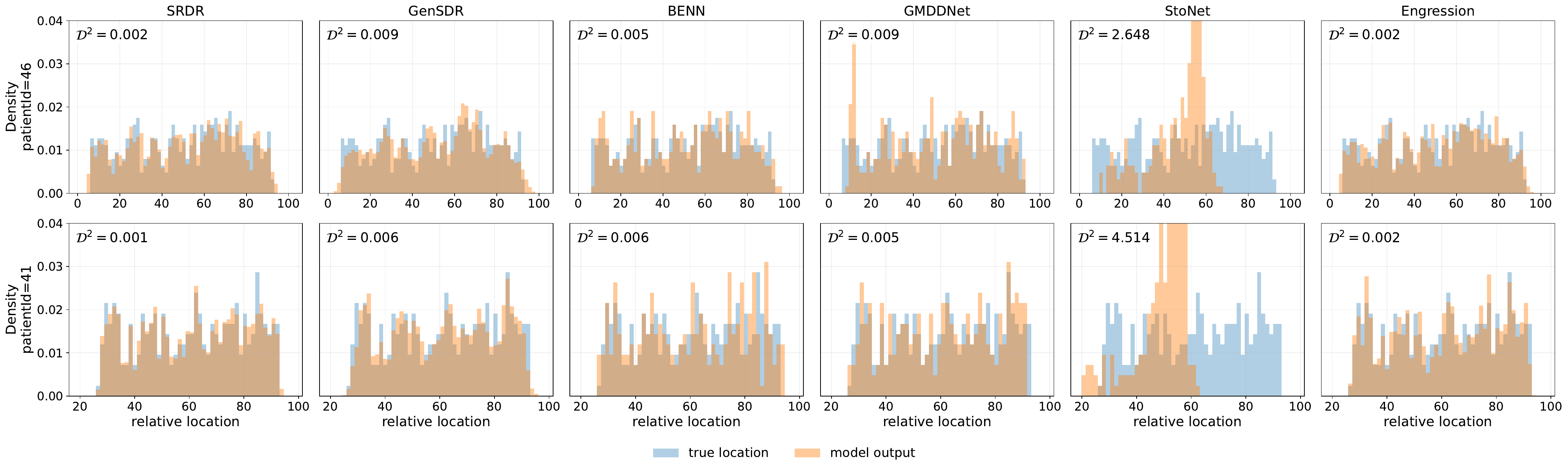}
    \caption{Patient-level distributional comparison for the CT localization dataset at $q=3$. Each panel shows the histograms of the observed slice location and model output (generated samples or point predictions) distributions, together with their squared energy distance $\mathcal{D}^2$.}
    \label{fig:ct_hist_main}
\end{figure}

To further assess how well the methods learn the distribution, Figure~\ref{fig:ct_hist_main} compares patient-level histograms of the observed slice localizations and model outputs for two randomly selected patients at $q=3$. For point-prediction methods, we plot one prediction per observation; for SRDR, GenSDR, and engression, we pool 100 predictive samples per observation. For observed values $y_1, \dots, y_{n_1}$ and model outputs $\tilde{y}_1, \dots, \tilde{y}_{n_2}$, the squared energy distance is estimated by the empirical counterpart of \eqref{eq:energy_divergence}, $\widehat{\mathcal{D}}^2 = \frac{1}{n_1 n_2}\sum_{i=1}^{n_1}\sum_{j=1}^{n_2}|y_i-\tilde{y}_j| - \frac{1}{2n_1^2}\sum_{i,i'=1}^{n_1}|y_i-y_{i'}| - \frac{1}{2n_2^2}\sum_{j,j'=1}^{n_2}|\tilde{y}_j-\tilde{y}_{j'}|$, on the original scale of the response. At \(q=3\), all methods except StoNet broadly overlap with the empirical distributions, and SRDR has the smallest $\mathcal{D}^2$ for both patients. At larger $q$ (in Appendix~\ref{app:ct}), the methods are more visually similar, although GenSDR and GMDDNet show somewhat larger dispersion and occasionally place mass beyond the main empirical support. Overall, SRDR continues to provide reliable patient-level distributional fits, which generally cannot be expected for non-distributional methods.

\subsection{Application to Superconductivity} \label{sec:supercond}

We evaluate SRDR on the superconductivity dataset \citep{Hamidieh2018}, which provides another regression setting with a continuous response. The dataset consists of $n=21{,}263$ observations of superconducting materials, each described by $p=82$ real-valued features derived from elemental properties of their chemical composition. The response variable is the critical temperature $T_c$, a continuous measure indicating the temperature below which a material becomes superconducting. The features are engineered summaries (e.g., means, weighted means, variances, entropies, and ranges) of atomic-level attributes such as atomic mass, valence, electron affinity, and thermal conductivity, computed over the constituent elements of each compound. 

\begin{figure}[!t]
    \centering
    \begin{subfigure}{0.48\textwidth}
        \centering
        \includegraphics[width=\textwidth]{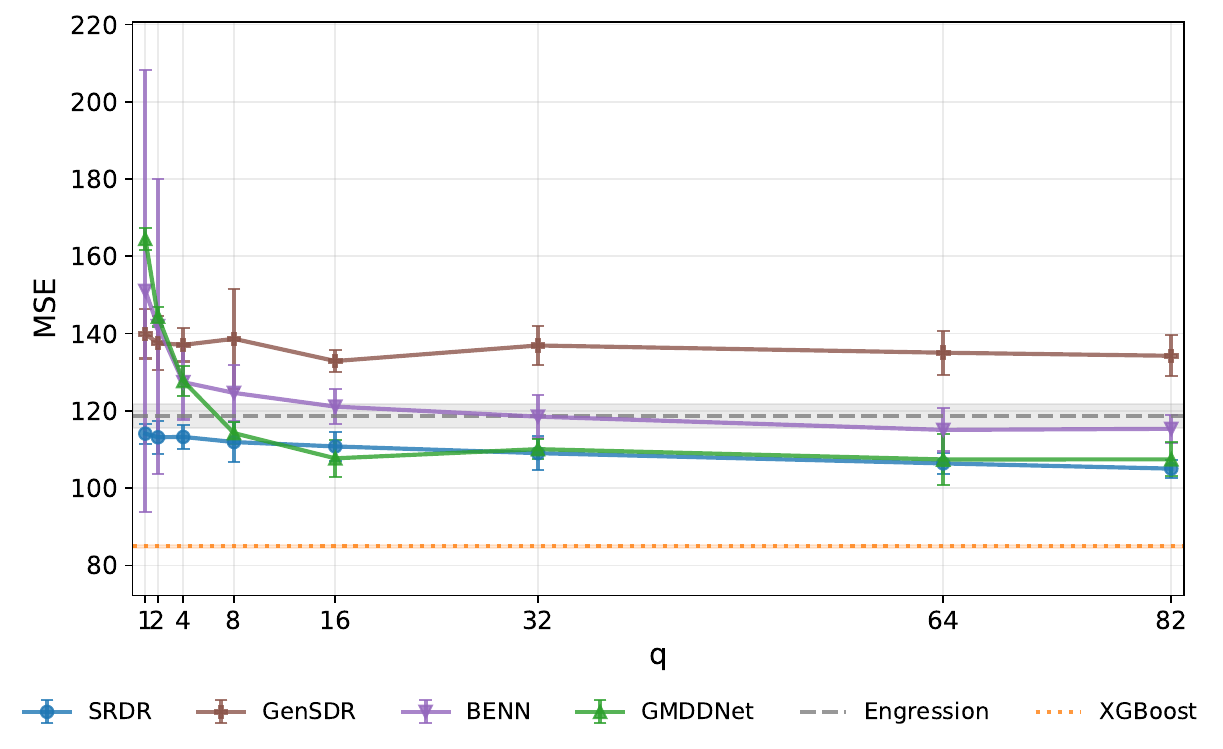}
        \caption{MSE}
        \label{fig:supercond_mse}
    \end{subfigure}
    \hfill
    \begin{subfigure}{0.48\textwidth}
        \centering
        \includegraphics[width=\textwidth]{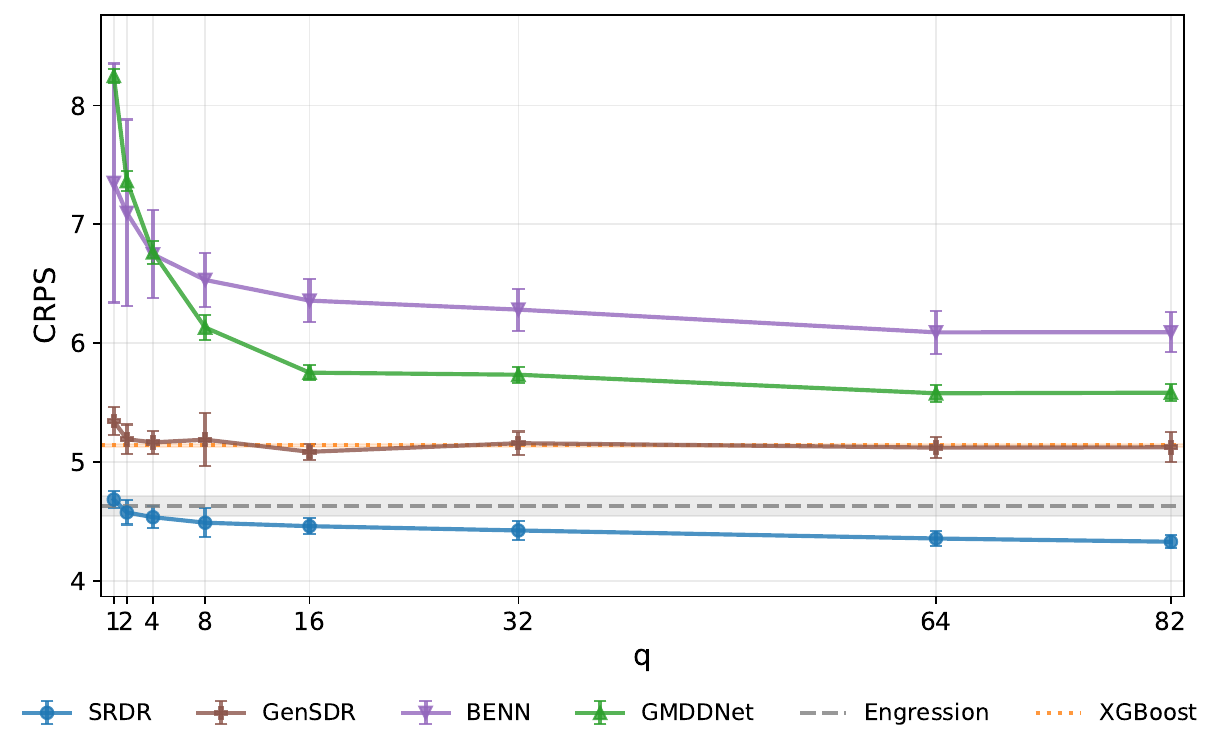}
        \caption{CRPS}
        \label{fig:supercond_dist}
    \end{subfigure}
    
    \caption{Superconductivity results by representation dimension $q$. Panel \ref{fig:supercond_mse} reports MSE, and panel \ref{fig:supercond_dist} reports CRPS for generative methods and MAE for deterministic methods.}
    \label{fig:supercond}
\end{figure}

As benchmarks, we use XGBoost and engression, both fitted on the full set of covariates; the former follows the recommendation of \citet{Hamidieh2018}. We do not report StoNet because, despite repeated tuning and substantial computational cost, its prediction errors remain outside the scale of the other methods. Figure~\ref{fig:supercond} reports the test MSE and CRPS, analogous to Section \ref{sec:ct}. SRDR is among the best performing SDR methods in terms of MSE across \(q\), and its performance stabilizes as \(q\) increases, suggesting that the learned reduction captures most of the predictive information relevant to \(T_c\). The full-feature XGBoost benchmark achieves the lowest MSE overall. A plausible explanation is that the original covariates contain well-separated clusters, which the tree-based XGBoost method exploits better than the neural network methods. This advantage vanishes under distributional evaluation: the CRPS of engression and SRDR is lower than the MAE of the XGBoost point prediction.

\begin{figure}[!ht]
    \centering
    \includegraphics[width=\linewidth]{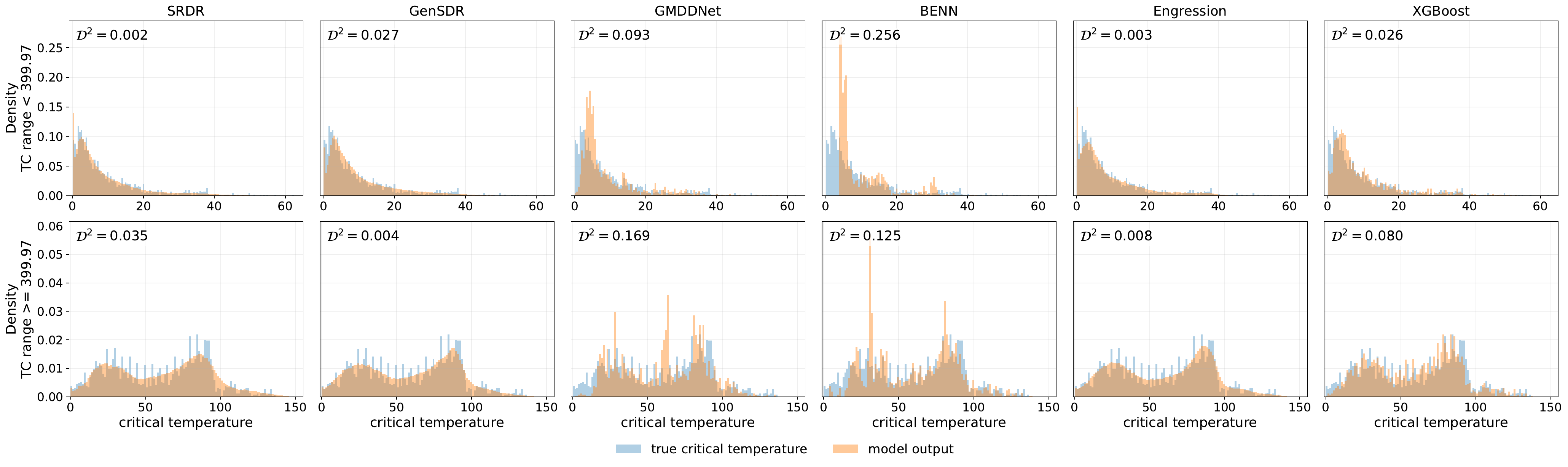}
    \caption{Distributional comparison for the superconductivity dataset at $q=1$. Each panel shows the histograms of the observed critical temperature and model output (generated samples or point predictions) distributions for the two subgroups defined by the range of thermal conductivity, together with their squared energy distance $\mathcal{D}^2$.}
\label{fig:supercond_hist}
\end{figure}

We further examine distributional learning through histograms of model outputs. Following \citet{Hamidieh2018}, we split the test set by the range of thermal conductivity, the most important feature identified by information gain. Figure~\ref{fig:supercond_hist} reports results for \(q=1\), with additional choices of \(q\) in Appendix~\ref{app:supercond}. Both SRDR and GenSDR match the observed distributional shape well in both subpopulations despite the one-dimensional reduction, and their squared energy distances $\mathcal{D}^2$ are among the smallest across the SDR methods. Engression and XGBoost also provide reasonable full-feature benchmarks, with XGBoost slightly underdispersed in the tails. In contrast, GMDDNet and BENN exhibit noticeably larger distributional discrepancies, with more concentrated outputs that fail to capture the full range of $T_c$. Compared with the CT dataset, the larger discrepancy between the observed distributions and the point-prediction outputs points to a lower signal-to-noise ratio in the superconductivity data, which explains why the generative methods gain more here by modeling the full conditional distribution.

\subsection{Application to Handwritten Digit Classification (MNIST)} \label{sec:mnist}

\begin{figure}[!ht]
\centering
\includegraphics[width=\textwidth]{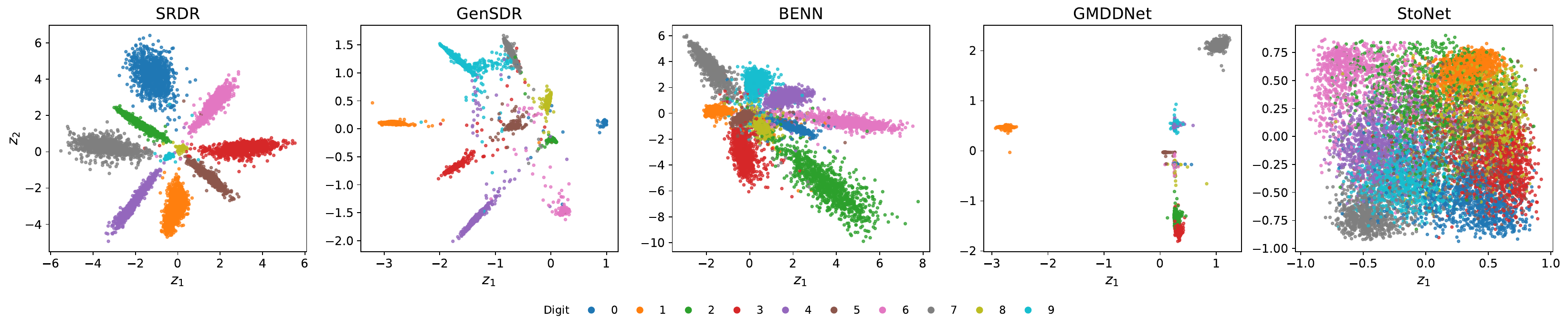}
\caption{Two-dimensional embeddings of MNIST images obtained from different SDR methods.}
\label{fig:mnist_embedding}
\end{figure}

We evaluate SRDR on a handwritten digit classification task using the well-known MNIST dataset, which consists of grayscale images of handwritten digits from 0 to 9. Each image is represented as a $28 \times 28$ pixel grid and is flattened into a 784-dimensional vector. The response variable $Y$ is categorical, with 10 classes corresponding to the digit labels.

We compare SRDR with network-based nonlinear SDR methods, including GenSDR, BENN, GMDDNet, and StoNet. To visualize the learned representations, we first learn a two-dimensional reduction $Z \in \mathbb{R}^2$ and plot the resulting test-data embeddings in Figure~\ref{fig:mnist_embedding}. Each point corresponds to an image and is colored by its digit label.

\begin{figure}[!t]
\centering
\includegraphics[width=0.5\textwidth]{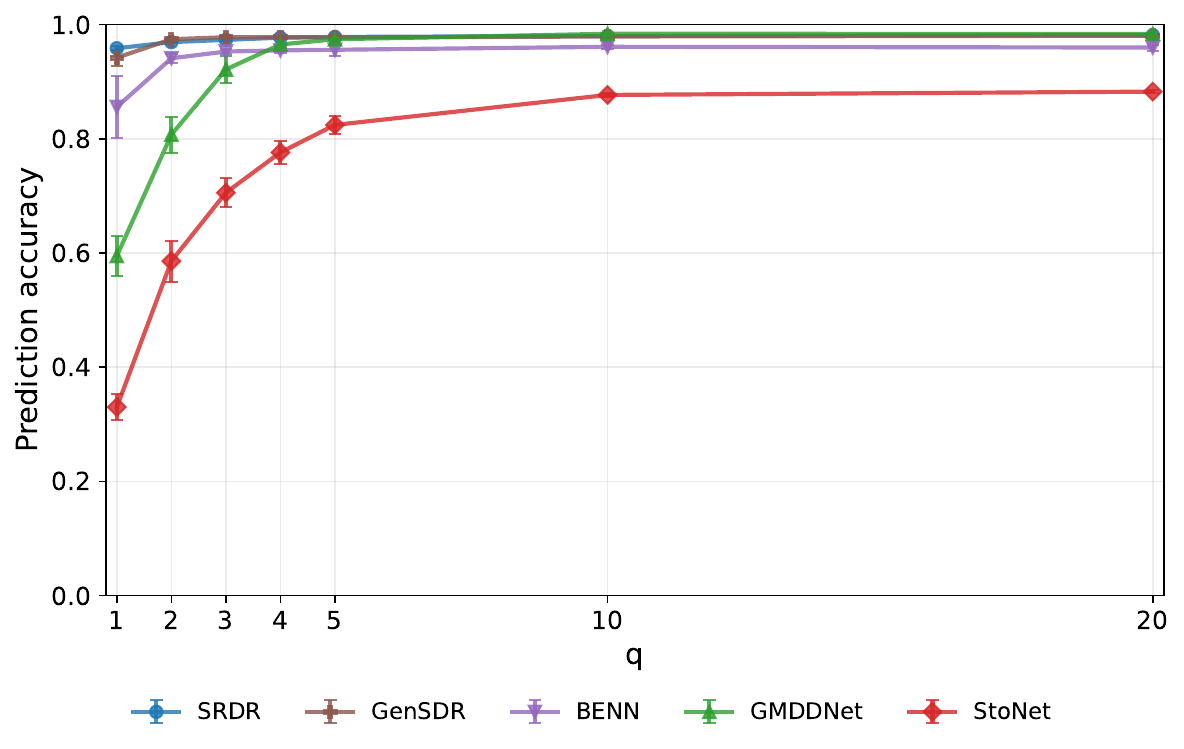}
\caption{Classification accuracy on the MNIST dataset by representation dimension $q$.}
\label{fig:mnist_accuracy}
\end{figure}

The learned embeddings reveal clear differences among the competing methods. StoNet produces embeddings with substantial overlap across digit classes, indicating that it does not fully separate the class-relevant structure in the images when $q$ is relatively small. GMDDNet and BENN yield more visually distinct clusters, although partial overlap remains among some digit classes. Both GenSDR and SRDR produce well-separated clusters for most digits, suggesting that the learned representation preserves discriminative information about the digit labels in a compact low-dimensional representation, with a few observations noticeably away from their primary cluster for GenSDR. We further compare the classification performance of the competing methods across several reduced dimensions. For SRDR and GenSDR, we use their native generative predictions: generated response vectors are averaged for each test observation, and the class corresponding to the largest coordinate is selected. For GMDDNet, StoNet, and BENN, which learn a dimension reduction but do not directly produce class predictions, we fit a logistic regression classifier on the learned representations. The results are summarized in Figure~\ref{fig:mnist_accuracy}. Overall, SRDR achieves the best or tied-best classification accuracy across all $q$. The improvement is particularly pronounced when $q$ is small, which indicates that SRDR is able to retain predictive information in a highly compact representation.

\section{Discussion} \label{sec:discussion}

We proved that sufficient dimension reduction can be performed via distributional regression and minimization of strictly proper scoring rules, thereby simultaneously achieving dimension reduction and targeting downstream predictive accuracy of probabilistic forecasts for the outcome $Y$. Based on this result, we introduced Sufficiently Reduced Distributional Regression (SRDR), a generative method with the energy score as target score. SRDR adapts the engression method of \citet{ShenMeinshausen2025} to sufficient dimension reduction, is computationally efficient, and has statistical consistency guarantees. In empirical applications, it outperforms non-generative methods for sufficient dimension reduction and is generally on par with, and in terms of predictive accuracy often superior to, the state-of-the-art generative method GenSDR by \citet{XuEtAl2025}.

There are several directions for future research, and some limitations of the present work. Our consistency result assumes an exact minimizer of the empirical energy score, whereas training is a nonconvex optimization problem. The reduced dimension $q$ is selected by validation score, which the theory does not cover. Our empirical results demonstrate strong performance of the extension of SRDR to classification problems. However, the general theory only covers regression problems with Lipschitz continuous reduction map and generator, and is not directly applicable to classification settings. Further investigations of the consistency of the classification SRDR are desirable. The multi-environment extension, with a shared reduction map and environment-specific generators, has consistency guarantees, but it is so far illustrated only through the transfer learning experiment in Appendix \ref{app:transfer}, which shows that a representation learned on the training environments is not necessarily sufficient for a new test environment. We investigated refining the reduction map with additional test environment data, in the spirit of the method by \citet{GeZhouHuang2025}. This approach shows strong performance, but the theoretical understanding is yet limited.

\section*{Acknowledgments, disclaimer, and reproducibility}
Alexander Henzi is supported by Shenzhen Science and Technology Program Grant AI2026019. Xinwei Shen is supported by Royalty Research Fund grant GR067344.

Large language models were used to assist with code implementation and to proofread the manuscript for typos, inconsistencies in mathematical notation, and language.

Code and replication material are available on \url{https://github.com/liutiang/SRDR}.

\printbibliography

\clearpage

\appendix

\section{Proof of Lemma \ref{lem:generative_sufficiency}}

\begin{proof}
By the Borel isomorphism Theorem \citep{RaoSrivastava1994}, there exists a Borel isomorphism $e\colon \mathbb{R}^p \rightarrow \mathbb{R}^q$, that is, a bijection that is measurable with measurable inverse. Since $X = e^{-1}(e(X))$, we have $\sigma(e(X)) = \sigma(X)$. Equation (1) of \citet{Song2025} applied to $(e(X),Y)$ shows that there exists a measurable function $\tilde{d}\colon \mathbb{R}^{q+1} \rightarrow \mathbb{R}^m$ and $\tilde{\eta} \sim \mathrm{Unif}(0,1)$ independent of $X$ such that
\[
    Y = \tilde{d}(e(X),\tilde{\eta})
\]
almost surely. Again by the Borel isomorphism Theorem, there exists a Borel isomorphism $h\colon \mathbb{R}^s \rightarrow \mathbb{R}$. Then the distribution of $h(\eta)$ for $\eta \sim \nu$ is atomless, because if $P(h(\eta) = z) > 0$ for some $z \in \mathbb{R}$, this would mean that $P(\eta = h^{-1}(z)) > 0$, implying that $\eta$ has an atom. Hence, the CDF $F$ of $h(\eta)$ is continuous, implying that $F(F^{-1}(u)) = u$ for all $u \in (0,1)$, and $F(h(\eta)) \sim \mathrm{Unif}(0,1)$. Consequently, $h^{-1}(F^{-1}(\tilde{\eta})) =: \eta$ is independent of $X$ and has distribution $\nu$, and we can define 
\[
    d\colon \mathbb{R}^{q+s} \rightarrow \mathbb{R}^m, \ (e,z) \mapsto \tilde{d}(e, F(h(z))),
\]
which is a composition of measurable functions and satisfies
\[
    d(e(X),\eta) = \tilde{d}(e(X),\tilde{\eta}) = Y
\]
almost surely.
\end{proof}

\section{Proof of Theorem \ref{thm:consistencyMULTIY}}

The proof of Theorem \ref{thm:consistencyMULTIY} is structured as follows:
\begin{itemize}
    \item In Section \ref{sec:test_approximation}, we reformulate the energy distance in the form of test functions, as already outlined in Section \ref{sec:test_functions}, and we show that one can restrict the integration domain to a suitable compact set by only adding a small approximation error. 
    \item Notation for the proofs is introduced in Section \ref{sec:proof_notation}.
    \item Like in \citet{TangLi2025}, we decompose the error into estimation and approximation error in Section \ref{sec:decomposition}, which are then bounded in Sections \ref{sec:s21} and \ref{sec:s22}.
    \item The bounds are combined in Section \ref{sec:completing} to complete the convergence result.
\end{itemize} 
While the overall proof closely follows that of \citet{TangLi2025}, the main challenges are the extension to multivariate outcomes and incorporating the noise variable $\eta$ into the calculations.

\subsection{Test functions and Approximation of Energy Distance} \label{sec:test_approximation}
For $y, t \in \mathbb{R}^m$, define the test functions
\[
    g_1(y,t) = \frac{\cos(t^{\top}y)}{\|t\|^{(m+1)/2}}, \ g_2(y,t) = \frac{\sin(t^{\top}y) }{\|t\|^{(m+1)/2}}.
\]
The energy distance can be written as
\[
    \mathcal{D}^2(P,Q) = \frac{\Gamma((m+1)/2)}{2\pi^{(m+1)/2}} \int_{\mathbb{R}^m} \left(\mathbb{E}_P[g_1(Y, t)] - \mathbb{E}_Q[g_1(Y,t)]\right)^2 + \left(\mathbb{E}_P[g_2(Y, t)] - \mathbb{E}_Q[g_2(Y,t)]\right)^2 \, dt,
\]
and the energy score is analogously equal to
\[
    \mathrm{ES}(P,y) = \frac{\Gamma((m+1)/2)}{2\pi^{(m+1)/2}} \int_{\mathbb{R}^m} \left(\mathbb{E}_P[g_1(Y, t)] - g_1(y,t)\right)^2 + \left(\mathbb{E}_P[g_2(Y, t)] - g_2(y,t)\right)^2 \, dt.
\]
For $B > 0$, define the set
\[
    R(B) = [-B, B]^m \setminus (-1/B, 1/B)^m.
\]
\begin{lemma} \label{lem:g_lipschitz_t_MULTIY}
For any $B_y, B > 0$, the functions $g_1$, $g_2$ are bounded and Lipschitz continuous with Lipschitz constant $L_g = L_g(B, B_y, m)$ on the domain $[-B_y, B_y]^m \times R(B)$.
\end{lemma}
The above lemma holds simply because restricted to the compact set $[-B_y, B_y]^m \times R(B)$, the functions $g_1, g_2$ are continuously differentiable. Let now $\mathcal{P}(B_y)$ be the family of distributions $P$ for which $\mathbb{E}_P[\|Y\|] \leq B_y$. In the result below, we show that restricting the integral in the characteristic function representation \eqref{eq:energy_divergence} to $R(B)$ approximates $\mathcal{D}^2(P,Q)$ uniformly over $\mathcal{P}(B_y)$, with an error that depends only on $B$, $B_y$ and $m$.

\begin{lemma} \label{lem:energy_approximation}
For any $\varepsilon > 0$, there exists $B = B(\varepsilon, B_y, m) > 0$ such that
\[
    \sup_{P,Q \in \mathcal{P}(B_y)}\Big\vert\mathcal{D}^2(P,Q) - \frac{\Gamma((m+1)/2)}{2\pi^{(m+1)/2}} \int_{R(B)} \frac{|\phi_P(t) - \phi_Q(t)|^2}{\|t\|^{m+1}} \, dt \Big\vert \leq \varepsilon.
\]
\end{lemma}
\begin{proof}
The set $\mathbb{R}^m \setminus [-B,B]^m$ is contained in $\{t \in \mathbb{R}^m\colon \|t\| > B\}$ and $|\phi_P(t) - \phi_Q(t)| \leq 2$, so
\begin{align*}
    \int_{\mathbb{R}^m \setminus [-B,B]^m} \frac{|\phi_P(t) - \phi_Q(t)|^2}{\|t\|^{m+1}} \, dt \leq 4\int_{\{t\colon \|t\| > B\}} \frac{1}{\|t\|^{m+1}} \, dt.
\end{align*}
Lemma \ref{lem:rotation_symmetric} implies
\[
\int_{\{t\colon \|t\| > B\}} \frac{1}{\|t\|^{m+1}} \, dt = mV(m)\int_{B}^{\infty} u^{m-1-(m+1)} \, du = \frac{mV(m)}{B}.
\]
It is well known that the characteristic function $\phi_P$ of a distribution $P$ with finite first moment satisfies
\[
    |\phi_P(t) - 1| \leq \mathbb{E}_P[|t^{\top}Y|] \leq \|t\| \, \mathbb{E}_P[\|Y\|].
\]
For $P, Q \in \mathcal{P}(B_y)$, this gives $|\phi_P(t) - \phi_Q(t)| \leq 2\|t\|B_y$, and hence
\begin{align*}
    \int_{[-1/B, 1/B]^m} \frac{|\phi_P(t) - \phi_Q(t)|^2}{\|t\|^{m+1}} \, dt & \leq \int_{[-1/B, 1/B]^m} \frac{4\|t\|^2B_y^2}{\|t\|^{m+1}} \, dt \\
    & = 4B_y^2\int_{[-1/B, 1/B]^m} \frac{1}{\|t\|^{m-1}} \, dt.
\end{align*}
If $m = 1$, the integral in the upper bound above equals
\[
    \int_{[-1/B, 1/B]^m} \frac{1}{\|t\|^{m-1}} \, dt = \frac{2}{B}.
\]
For $m > 1$, the set $[-1/B, 1/B]^m$ is contained in a ball with radius $\sqrt{m}B^{-1}$, so
\[
\int_{[-1/B, 1/B]^m} \frac{1}{\|t\|^{m-1}} \, dt \leq \int_{\{t\colon \|t\| \leq \sqrt{m}B^{-1}\}} \frac{1}{ \|t\|^{m-1}} \, dt,
\]
and Lemma \ref{lem:rotation_symmetric} implies
\[
    \int_{\{t\colon \|t\| \leq \sqrt{m}B^{-1}\}} \frac{1}{ \|t\|^{m-1}} \, dt = mV(m)\int_0^{\sqrt{m}B^{-1}}r^{m-(m-1)-1} \, \, dr = \frac{m^{3/2}V(m)}{B}.
\]
Combining the above derivations, we see that
\[
\frac{\Gamma((m+1)/2)}{2\pi^{(m+1)/2}}\int_{\mathbb{R}^m \setminus R(B)} \frac{|\phi_P(t) - \phi_Q(t)|^2}{\|t\|^{m+1}} \, dt \leq \frac{C(B_y, m)}{B}
\]
for a constant $C(B_y, m)$ depending only on $B_y$ and $m$, because $\mathbb{R}^m \setminus R(B)$ is contained in $(\mathbb{R}^m \setminus [-B,B]^m) \cup [-1/B, 1/B]^m$. Choosing $B \geq C(B_y, m)/\varepsilon$ proves the result.
\end{proof}

The following lemma is a standard result about the integral of rotationally symmetric functions, which is used in Lemma \ref{lem:energy_approximation} and in the further steps of the proof below.

\begin{lemma} \label{lem:rotation_symmetric}
Let $f \colon \mathbb{R}^d \rightarrow [0,\infty)$ satisfy $f(x) = h(\|x\|)$ for some function $h$ and all $x \in \mathbb{R}^d$. Then,
\[
    \int_{\mathbb{R}^d} f(x) \, dx = dV(d)\int_{0}^{\infty} h(r) r^{d-1}\, dr,
\]
where $V(d) = \pi^{d/2}/\Gamma(d/2 + 1)$ is the volume of the $d$-dimensional unit ball.
\end{lemma}

\subsection{Notation} \label{sec:proof_notation}

Our target of interest is the expected energy distance,
\begin{align*}
    \mathcal{L} & = \frac{\Gamma((m+1)/2)}{2\pi^{(m+1)/2}}\mathbb{E}\Big[\mathbb{E}_{X\sim P_X}\Big[\int_{\mathbb{R}^m} \frac{(\mathbb{E}_{\eta}[\cos(t^{\top}d^*(e^*(X), \eta)) - \cos(t^{\top}T_{B_y}\hat{d}(\hat{e}(X), \eta))])^2}{\|t\|^{m+1}} \\
    & \qquad + \frac{(\mathbb{E}_{\eta}[\sin(t^{\top}d^*(e^*(X), \eta)) - \sin(t^{\top}T_{B_y}\hat{d}(\hat{e}(X), \eta))])^2}{\|t\|^{m+1}} \,dt\Big]\Big].
\end{align*}
The outer expectation is taken with respect to the training data $(X_i, Y_i)$, $i = 1, \dots, n$, that yields $(\hat{d}, \hat{e})$, where we drop the $n$ in the subscript. To simplify the proof, consider only the cosine part,
\[
    \tilde{\mathcal{L}} = \mathbb{E}\Big[\mathbb{E}_{X\sim P_X}\Big[\int_{\mathbb{R}^m} \frac{(\mathbb{E}_{\eta}[\cos(t^{\top}d^*(e^*(X), \eta)) - \cos(t^{\top}T_{B_y}\hat{d}(\hat{e}(X), \eta))])^2}{\|t\|^{m+1}} \, dt\Big]\Big],
\]
also omitting the multiplicative constant. Exactly the same proof, but with more cumbersome notation, applies to the overall error $\mathcal{L}$, because the test functions with sine and with cosine are equivalent in their relevant smoothness and boundedness properties. We introduce the following notation, which is analogous to that of \citet{TangLi2025},
\begin{align*}
    & g(y,t) = \frac{\cos(t^{\top}y)}{\|t\|^{(m+1)/2}}, \\
    & s^*(x,t) = \mathbb{E}[g(Y,t) \mid X = x] = \mathbb{E}_{\eta}[g(d^*(e^*(x),\eta), t)] =: h^*(e^*(x),t), \\
    & \hat{h}(\hat{e}(x),t) = \mathbb{E}_{\eta}[g(T_{B_y}\hat{d}(\hat{e}(x),\eta), t)].
\end{align*}
With these definitions, we can write $\tilde{\mathcal{L}}$ compactly as
\begin{equation} \label{eq:error_integrated}
    \tilde{\mathcal{L}} = \mathbb{E}\Big[\mathbb{E}_{X\sim P_X}\Big[\int_{\mathbb{R}^m} (s^*(X,t) - \hat{h}(\hat{e}(X),t))^2 \, dt \Big]\Big].
\end{equation}
Accordingly, we treat $(\hat{e}, \hat{d})$ as a minimizer of the empirical loss
\begin{align*}
\mathbb{E}_n\Big[\int_{\mathbb{R}^m} \left(g(Y,t) - \mathbb{E}_{\eta}[g(T_{B_y}d(e(X),\eta),t)]\right)^2 \, dt\Big],
\end{align*}
where $\mathbb{E}_n$ is defined as the expected value over the empirical distribution $\nu_n = n^{-1}\sum_{i=1}^n \delta_{(X_i, Y_i)}$.

\subsection{Decomposition} \label{sec:decomposition}
With Lemma S.1 from Section S.2.3 of \citet{TangLi2025}, we can decompose the error \eqref{eq:error_integrated} as
\[
    \tilde{\mathcal{L}} = \mathbb{E}[S_{21}] + \mathbb{E}[S_{22}],
\]
where we define
\begin{align*}
    & D_n = (X_i, Y_i)_{i=1}^n, \\
    & S_{21} = \Big\{\mathbb{E}\Big[\int_{\mathbb{R}^m} (\hat{h}(\hat{e}(X),t) - g(Y,t))^2 \, dt \mid D_n\Big]  - \mathbb{E}\Big[\int_{\mathbb{R}^m} (s^*(X,t) - g(Y,t))^2 \, dt\Big] \\
    & \qquad - 2\Big[\mathbb{E}_n\Big[\int_{\mathbb{R}^m} (\hat{h}(\hat{e}(X),t) - g(Y,t))^2 \, dt \Big]  - \mathbb{E}_n\Big[\int_{\mathbb{R}^m} (s^*(X,t) - g(Y,t))^2 \, dt\Big]\Big]\Big\}, \\
    & S_{22} = 2\Big[\mathbb{E}_n\Big[\int_{\mathbb{R}^m} (\hat{h}(\hat{e}(X),t) - g(Y,t))^2 \, dt \Big]  - \mathbb{E}_n\Big[\int_{\mathbb{R}^m} (s^*(X,t) - g(Y,t))^2 \, dt\Big]\Big].
\end{align*}
Here $(X,Y)$ is a copy of $(X_1, Y_1)$ independent of the training data $D_n$. The notation $S_{21}, S_{22}$ is chosen to facilitate comparison to the proof of \citet{TangLi2025}, where the same decomposition is made. Let $\varepsilon > 0$. Since $Y$ and $T_{B_y}\hat{d}(\hat{e}(X),\eta)$ are contained in $[-B_y, B_y]^m$ almost surely, all distributions involved belong to $\mathcal{P}(\sqrt{m}B_y)$, and Lemma \ref{lem:energy_approximation} implies that we can choose a constant $B$ sufficiently large such that
\begin{align}
     & S_{21} \leq \Big\{\mathbb{E}\Big[\int_{R(B)} (\hat{h}(\hat{e}(X),t) - g(Y,t))^2 \, dt \mid D_n\Big]  - \mathbb{E}\Big[\int_{R(B)} (s^*(X,t) - g(Y,t))^2 \, dt\Big] \nonumber \\
    & \qquad - 2\Big[\mathbb{E}_n\Big[\int_{R(B)} (\hat{h}(\hat{e}(X),t) - g(Y,t))^2 \, dt \Big]  - \mathbb{E}_n\Big[\int_{R(B)} (s^*(X,t) - g(Y,t))^2 \, dt\Big]\Big]\Big\} + \varepsilon/4 \label{eq:s21_eps_bound},
\end{align}
replacing the integration domain $\mathbb{R}^m$ by $R(B)$ and adding an additional error of $\varepsilon/4$.

\subsection{Discretizing and bounding $\mathbb{E}[S_{21}]$} \label{sec:s21}
Assume that $B$ is a positive integer, which is no restriction since $B$ can be taken arbitrarily large, and define $\kappa = 2B$. For an integer $\tilde{M}$ that is a multiple of $2B^2$, partition $[-B, B]^m$ into $M := \tilde{M}^m$ cubes of side length $\kappa/\tilde{M}$ with centers
\[
    t(i_1, \dots, i_m) = \left(-B + \frac{\kappa(i_1 - 1/2)}{\tilde{M}}, \dots, -B + \frac{\kappa(i_m - 1/2)}{\tilde{M}}\right), \ i_1, \dots, i_m \in \{1, \dots, \tilde{M}\}.
\]
Let $t_1, \dots, t_M$ be an enumeration of these centers, let $R_j$ be the cube with center $t_j$, and let $J \subseteq \{1, \dots, M\}$ be the indices of the cubes that are not contained in $[-B^{-1}, B^{-1}]^m$. Since $\pm B^{-1}$ are grid points, the cubes $R_j$, $j \in J$, are contained in $R(B)$ and their union is $R(B)$. Define
\begin{align*}
    & \mathbf{s}^*(x) = (s^*(x, t_j))_{j \in J}, \\
    & \mathbf{g}(Y) = (g(Y, t_j))_{j \in J}, \\
    & \mathbf{\hat{h}}(\hat{e}(X)) = (\hat{h}(\hat{e}(X), t_j))_{j \in J}.
\end{align*}
Write the term in braces in \eqref{eq:s21_eps_bound} as $\int_{R(B)} \Psi(t) \, dt = \sum_{j \in J} \int_{R_j} \Psi(t) \, dt$, where
\begin{align*}
    & \psi(x,y,t) = (\hat{h}(\hat{e}(x),t) - g(y,t))^2 - (s^*(x,t) - g(y,t))^2, \\
    & \Psi(t) = \mathbb{E}[\psi(X,Y,t) \mid D_n] - 2\mathbb{E}_n[\psi(X,Y,t)].
\end{align*}
The functions $t \mapsto g(y,t)$, $s^*(x,t)$ and $\hat{h}(\hat{e}(x),t)$ are averages of $g(y',t)$ over $y' \in [-B_y, B_y]^m$, and $g$ is twice continuously differentiable on the compact set $[-B_y, B_y]^m \times R(B)$. Hence there is a constant $L_{\Psi} = L_{\Psi}(B, B_y, m)$ such that the Hessian of $\Psi$ satisfies $\|\nabla^2 \Psi(t)\| \leq L_{\Psi}$ for all $t \in R(B)$. Since $t_j$ is the center of $R_j$, the linear term of the Taylor expansion of $\Psi$ around $t_j$ integrates to zero over $R_j$, and we obtain
\begin{align*}
    \int_{R_j} \Psi(t) \, dt & \leq (\kappa/\tilde{M})^m \Psi(t_j) + \frac{L_{\Psi}}{2}\int_{R_j} \|t-t_j\|^2 \, dt \\
    & \leq \kappa^m M^{-1} \Psi(t_j) + \frac{L_{\Psi}}{2}\int_{\{t\colon \|t\| \leq \sqrt{m}\kappa/(2\tilde{M})\}} \|t\|^2\, dt \\
    & = \kappa^m M^{-1} \Psi(t_j) + \frac{L_{\Psi}}{2} \frac{mV(m)}{m+2} (\sqrt{m}\kappa/(2\tilde{M}))^{m+2} \\
    & = \kappa^m M^{-1} \Psi(t_j) + cM^{-(1+2/m)},
\end{align*}
where we apply Lemma \ref{lem:rotation_symmetric} and define $c = L_{\Psi}(\sqrt{m}\kappa/2)^{m+2}mV(m)/(2(m+2))$. Summing over $j \in J$ and using $|J| \leq M$ gives the upper bound
\begin{align} 
    S_{21} & \leq \kappa^m M^{-1}\big\{\mathbb{E}[\|\mathbf{\hat{h}}(\hat{e}(X)) - \mathbf{g}(Y)\|^2 \mid D_n]  - \mathbb{E}[\|\mathbf{s}^*(X)  - \mathbf{g}(Y)\|^2] \nonumber \\
    & \qquad - 2\big[\mathbb{E}_n[\|\mathbf{\hat{h}}(\hat{e}(X)) - \mathbf{g}(Y)\|^2]  - \mathbb{E}_n[\|\mathbf{s}^*(X)  - \mathbf{g}(Y)\|^2]\big]\big\} + cM^{-2/m} + \varepsilon/4 \nonumber \\
    & =: \kappa^m M^{-1}\tilde{S}_{21} + cM^{-2/m} + \varepsilon/4. \label{eq:s21_upper_bound}
\end{align}
To bound $\mathbb{E}[\tilde{S}_{21}]$, we need to bound the complexity of our function class, similarly to Lemma S.2 and Lemma S.3 of \citet{TangLi2025}. The function class relevant for our problem is
\begin{align*}
    & \mathcal{F}^{(t)}(B_y) = \Big\{f\colon \mathbb{R}^{p} \rightarrow \mathbb{R}, x \mapsto \mathbb{E}_{\eta}[g(T_{B_y}d(e(x), \eta), t)], \\
    & \quad \text{ with } e \in \mathcal{F}(p,\ell_1,r_1,q,B_w), \ d \in \mathcal{F}(q+s,\ell_2,r_2,m,B_w)\Big\}.
\end{align*}
We express this function class as a subset of a larger neural network class for which a bound on the covering number is available.
To this end, define the auxiliary class
\begin{align*}
    & \mathcal{F}(B_y) = \big\{ f\colon \mathbb{R}^{p+s} \rightarrow \mathbb{R}, x \mapsto t^{\top}T_{B_y}d(e(x_{1:p}), x_{(p+1):(p+s)}), \text{ with }\\
    & \quad e \in \mathcal{F}(p,\ell_1,r_1,q,B_w), \ d \in \mathcal{F}(q+s,\ell_2,r_2,m,B_w), t \in [-B,B]^m \big\}.
\end{align*}

\begin{lemma} \label{lem:auxiliary_network_class}
For $B_w \geq \max(1,B_y, B)$, we have
\[
    \mathcal{F}(B_y) \subseteq \check{\mathcal{F}} := \mathcal{F}\big(p+s,\ell_1 + \ell_2 + 3, ((r_1 + 2s)_{i=1,\dots,\ell_1}, q + s, (r_2)_{i=1,\dots, \ell_2}, m, 4m), 1, B_w\big).
\]
\end{lemma}
\begin{proof}
For the ReLU activation function and weights and biases bounded from above by $B_w \geq 1$, we can extend each layer of the first network $e$ with neurons containing $(x_{p+i})_+$ or $(-x_{p+i})_+$, and concatenate $e(x)$ with
\[
    x_{p+i} = (x_{p+i})_+ - (x_{p+i})_-, \ i = 1, \dots, s.
\]
This shows that $x \mapsto (e(x_{1:p}), x_{(p+1):(p+s)})$ can be obtained as the output of a network with $\ell_1$ hidden layers and $r_1 + 2s$ neurons in each layer. As shown in Lemma S.2 of \citet{TangLi2025}, in a ReLU network with $B_w \geq \max(1, B_y)$ and output size $d$, the operator $T_{B_y}$ can be represented by adding an additional layer of $4d$ neurons, by using the relationship
\[
    T_B(a) = -(-a+B)_+ + (a+B)_+ + (-a)_+ -a_+.
\]
The scalar product $x \mapsto t^{\top}x$ for $x \in \mathbb{R}^d$ can be represented by the linear output layer to dimension $1$, since $B_w \geq B$.
\end{proof}

Denote by $\mathcal{N}(\varepsilon, \mathcal{F}, \|\cdot\|)$ the $\varepsilon$ covering number of a function class $\mathcal{F}$ with norm $\|\cdot\|$.

\begin{lemma} \label{lem:entropy}
Suppose that
\begin{enumerate}
    \item the support $\mathcal{X}$ of $X$ is a subset of $[-B_x, B_x]^p$,
    \item the support of the noise variable $\eta$ is contained in $[-B_x, B_x]^s$,
    \item and $B_w \geq \max(1, B_y, B)$.
\end{enumerate}
Then for any $t \in R(B)$ and any $\varepsilon > 0$,
\begin{align*}
    & \sup_{x_1, \dots, x_n \in \mathcal{X}} \mathcal{N}(\varepsilon, \mathcal{F}^{(t)}(B_y), \|\cdot\|_{L_1(\nu_n)}) \\
    & \leq \frac{\left(32B^{m+1}(\ell_1 + \ell_2 + 4)(\sqrt{p + s}B_x + 1)(2B_w)^{\ell_1 + \ell_2 + 5}(p + s)(r_1 + 2s)^{\ell_1}(q+s)r_2^{\ell_2}m^2\right)^{\mathcal{S}}}{((r_1+2s)!)^{\ell_1}((q+s)!)(r_2!)^{\ell_2}(m!)((4m)!)}\varepsilon^{-\mathcal{S}},
\end{align*}
with the total number of parameters $\mathcal{S}$ equal to 
\begin{align*}
    \mathcal{S} & = (\ell_1 - 1)(r_1 + 2s)^2 + (p + \ell_1 + q + 2s)(r_1 + 2s) + (\ell_2  - 1)r_2^2 \\
    & \quad + (q + \ell_2 + m + s)r_2 + (q + s + 4m^2 + 9m + 1).
\end{align*}
\end{lemma}
\begin{proof}
As in \citet{TangLi2025}, we provide a bound for the covering number of $\mathcal{F}^{(t)}(B_y)$ in the supremum norm. For $f \in \check{\mathcal{F}}$ and $x \in \mathbb{R}^p$, define
\begin{equation} \label{eq:f_ring}
    \mathring{f}(x) = \mathbb{E}_{\eta}[\cos(f((x,\eta)))/\|t\|^{(m+1)/2}].
\end{equation}
By Lemma \ref{lem:auxiliary_network_class}, the class $\mathring{\mathcal{F}} = \{\mathring{f}\colon f \in \check{\mathcal{F}}\}$ contains $\mathcal{F}^{(t)}(B_y)$, and so
\[
    \mathcal{N}(\varepsilon, \mathcal{F}^{(t)}(B_y), \|\cdot\|_{\infty}) \leq \mathcal{N}(\varepsilon/2, \mathring{\mathcal{F}}, \|\cdot\|_{\infty}).
\]
Let now $f_1, \dots, f_N$ be an $\varepsilon$-covering of $\check{\mathcal{F}}$, and define $\mathring{f}_i$ as in \eqref{eq:f_ring}. Then for any $f \in \check{\mathcal{F}}$ and the $f_i$ with $\|f - f_i\|_{\infty} \leq \varepsilon$, we have
\begin{align*}
    & \sup_{x \in \mathcal{X}}|\mathbb{E}_{\eta}[\cos(f((x,\eta)))/\|t\|^{(m+1)/2}] - \mathbb{E}_{\eta}[\cos(f_i((x,\eta)))/\|t\|^{(m+1)/2}]| \\
    \quad & \leq \sup_{(x,\eta) \in [-B_x, B_x]^{p+s}}B^{m+1}|f((x,\eta)) - f_i((x,\eta))| \\
    & = B^{m+1}\|f_i - f\|_{\infty},
\end{align*}
where we use $|\cos(a) - \cos(b)| \leq |a - b|$ and $\|t\|^{-(m+1)/2} \leq B^{(m+1)/2} \leq B^{m+1}$ for $t \in R(B)$ and $B \geq 1$. This implies that $\mathring{f}_1,  \dots, \mathring{f}_N$ is a $B^{m+1}\varepsilon$ cover of $\mathring{\mathcal{F}}$, so
\begin{align*}
    \sup_{x_1, \dots, x_n \in \mathcal{X}} \mathcal{N}(\varepsilon, \mathcal{F}^{(t)}(B_y), \|\cdot\|_{L_1(\nu_n)}) \leq \mathcal{N}(\varepsilon/2,\mathring{\mathcal{F}}, \|\cdot\|_{\infty}) \leq \mathcal{N}(\varepsilon/(2B^{m+1}), \check{\mathcal{F}}, \|\cdot\|_{\infty}).
\end{align*}
The entropy bound for $\check{\mathcal{F}}$ follows by Theorem 3.5 of \citet{Shen2024}.
\end{proof}

The remainder of the bound for $S_{21}$ follows exactly the same arguments as in the proof of \citet{TangLi2025}, and we show the main steps for completeness. Express $\tilde{S}_{21}$ as
\begin{align*}
    \tilde{S}_{21} & = \sum_{j \in J} \Big\{\mathbb{E}\left[(\hat{h}(\hat{e}(X),t_j) - g(Y, t_j))^2 \mid D_n\right] - \mathbb{E}\left[(s^*(X,t_j) - g(Y,t_j))^2\right] \\ 
     & \qquad - 2 \left[\mathbb{E}_n\left[(\hat{h}(\hat{e}(X),t_j) - g(Y, t_j))^2\right] - \mathbb{E}_n\left[(s^*(X,t_j) - g(Y,t_j))^2\right]\right]\Big\} =: \sum_{j \in J} \tilde{S}_{21}^{(j)},
\end{align*}
such that, since $|J| \leq M$, $\mathbb{P}(\tilde{S}_{21} > u)$ is bounded from above by $\sum_{j \in J}\mathbb{P}(\tilde{S}_{21}^{(j)} > u/M)$. Denote by $A^{(j)}(\varepsilon)$ the set
\begin{align*}
    A^{(j)}(\varepsilon) & = \Big\{\exists \ f \in \mathcal{F}^{(t_j)}(B_y)\colon \mathbb{E}\left[(f(X) - g(Y, t_j))^2 \mid D_n\right] - \mathbb{E}\left[(s^*(X,t_j) - g(Y,t_j))^2\right] \\ 
     & \qquad - 2 \left[\mathbb{E}_n\left[(f(X) - g(Y, t_j))^2\right] - \mathbb{E}_n\left[(s^*(X,t_j) - g(Y,t_j))^2\right]\right] \geq \varepsilon\Big\}.
\end{align*}
Then $\mathbb{P}(\tilde{S}_{21}^{(j)} \geq u/M) \leq \mathbb{P}(A^{(j)}(u/M))$. Let $B_g = \max\{|g(y,t)| \colon (y,t) \in [-B_y, B_y]^m \times R(B)\}$, which bounds $|g(Y,t_j)|$, $|s^*(X,t_j)|$ and $|f(X)|$ for $f \in \mathcal{F}^{(t_j)}(B_y)$. Theorem 11.4 of \citet{GyoerfiEtAl2002} with their $\alpha, \beta$ set to $\varepsilon/2$ and their $\varepsilon$ set to $1/2$ yields
\[
    \mathbb{P}(A^{(j)}(\varepsilon)) \leq 14\sup_{x_1, \dots, x_n \in \mathcal{X}} \mathcal{N}\left(\frac{\varepsilon}{80B_g}, \mathcal{F}^{(t_j)}(B_y), \|\cdot\|_{L_1(\nu_n)}\right)\exp(-C_2n\varepsilon) \leq C_1\varepsilon^{-\mathcal{S}}\exp(-C_2n\varepsilon),
\]
where, by Lemma \ref{lem:entropy},
\begin{align*}
    C_1 & = 14\frac{\left(2560B_gB^{m+1}(\ell_1 + \ell_2 + 4)(\sqrt{p + s}B_x + 1)(2B_w)^{\ell_1 + \ell_2 + 5}(p + s)(r_1 + 2s)^{\ell_1}(q+s)r_2^{\ell_2}m^2\right)^{\mathcal{S}}}{((r_1+2s)!)^{\ell_1}((q+s)!)(r_2!)^{\ell_2}(m!)((4m)!)},
\end{align*}
and
\begin{align*}
    C_2 = \frac{1}{5136B_g^4}.
\end{align*}
This implies
\begin{align*}
    \mathbb{P}(\tilde{S}_{21} > u) & \leq \sum_{j \in J} C_1(u/M)^{-\mathcal{S}}\exp(-C_2nu/M) \\
    & \leq C_1 M^{\mathcal{S}+1}u^{-\mathcal{S}}\exp(-C_2nu/M) \\
    & \leq C_1 M^{\mathcal{S}+1}n^{\mathcal{S}}\exp(-C_2nu/M),
\end{align*}
where the last line holds for $u \geq 1/n$. Lemma S.5 of \citet{TangLi2025} then yields
\[
    \mathbb{E}[\tilde{S}_{21}] \leq \frac{1+\log(C_1) + (\mathcal{S} + 1)\log(M) + \mathcal{S}\log(n)}{C_2n/M},
\]
which, together with \eqref{eq:s21_upper_bound}, gives
\[
   \mathbb{E}[S_{21}] \leq \kappa^{m}n^{-1}C_2^{-1}(1 + \log(C_1) + (\mathcal{S}+1)\log(M) + \mathcal{S}\log(n)) + cM^{-2/m} + \varepsilon/4.
\]

\subsection{Bounding $\mathbb{E}[S_{22}]$} \label{sec:s22}
In this section, we bound the expected value of 
\[
S_{22} = 2\Big[\mathbb{E}_n\Big[\int_{\mathbb{R}^m} (\hat{h}(\hat{e}(X),t) - g(Y,t))^2 \, dt \Big]  - \mathbb{E}_n\Big[\int_{\mathbb{R}^m} (s^*(X,t) - g(Y,t))^2 \, dt\Big]\Big].
\]
Recall that $(\hat{e}, \hat{d})$ are defined as minimizers of this criterion, via the relationship $\hat{h}(\hat{e}(x),t) = \mathbb{E}_{\eta}[g(T_{B_y}\hat{d}(\hat{e}(x),\eta),t)]$. Like at the beginning of Section \ref{sec:s21}, define a discretized version of $S_{22}$ as
\[
    \tilde{S}_{22} = 2\kappa^m M^{-1}\big\{\big[\mathbb{E}_n[\|\mathbf{\hat{h}}(\hat{e}(X)) - \mathbf{g}(Y)\|^2]  - \mathbb{E}_n[\|\mathbf{s}^*(X)  - \mathbf{g}(Y)\|^2]\big]\big\},
\]
and let $(\tilde{e}, \tilde{d})$ and the corresponding $\tilde{h}(\tilde{e}(x),t) = \mathbb{E}_{\eta}[g(T_{B_y}\tilde{d}(\tilde{e}(x),\eta),t)]$ be minimizers of this discretized criterion. Furthermore, choose $B$ large enough such that
\[
    \Big\vert S_{22} - 2\Big(\mathbb{E}_n\Big[\int_{R(B)} (h(e(X),t) - g(Y,t))^2 \, dt \Big]  - \mathbb{E}_n\Big[\int_{R(B)} (s^*(X,t) - g(Y,t))^2 \, dt\Big]\Big) \Big\vert \leq \varepsilon/4,
\]
for all choices of $e$ and $d$, which is possible by Lemma \ref{lem:energy_approximation}. Then,
\begin{align*}
    S_{22} & \leq 2\Big[\mathbb{E}_n\Big[\int_{\mathbb{R}^m} (\tilde{h}(\tilde{e}(X),t) - g(Y,t))^2 \, dt \Big]  - \mathbb{E}_n\Big[\int_{\mathbb{R}^m} (s^*(X,t) - g(Y,t))^2 \, dt\Big]\Big] \\
    & \leq 2\Big[\mathbb{E}_n\Big[\int_{R(B)} (\tilde{h}(\tilde{e}(X),t) - g(Y,t))^2 \, dt \Big]  - \mathbb{E}_n\Big[\int_{R(B)} (s^*(X,t) - g(Y,t))^2 \, dt\Big]\Big] + \varepsilon/4 \\
    & \leq 2\kappa^{m} M^{-1}\big\{\big[\mathbb{E}_n[\|\mathbf{\tilde{h}}(\tilde{e}(X)) - \mathbf{g}(Y)\|^2]  - \mathbb{E}_n[\|\mathbf{s}^*(X)  - \mathbf{g}(Y)\|^2]\big]\big\} + c'M^{-2/m} + \varepsilon/4,
\end{align*}
where the remainder includes a constant $c'$ depending on $B$, $B_y$ and $m$ and is derived with exactly the same arguments as in Section \ref{sec:s21}. Now extend the definition of the function class $\mathcal{F}^{(t)}(B_y)$ to the parameters $t_j$, $j \in J$, as follows,
\begin{align*}
    & \mathcal{F}_M(B_y) = \Big\{f\colon \mathbb{R}^{p} \rightarrow \mathbb{R}^{|J|}, x \mapsto \big(\mathbb{E}_{\eta}\left[g(T_{B_y}d(e(x), \eta), t_j)\right]\big)_{j \in J}, \\
    & \quad \text{ with } e \in \mathcal{F}(p,\ell_1,r_1,q,B_w), \ d \in \mathcal{F}(q+s,\ell_2,r_2,m,B_w)\Big\}.
\end{align*}
Since $(\tilde{d}, \tilde{e})$ are minimizers of the sample criterion $\mathbb{E}_n[\|\mathbf{\tilde{h}}(\tilde{e}(X)) - \mathbf{g}(Y)\|^2]$ and $s^*(x,t) = \mathbb{E}[g(Y,t)\mid X = x]$, we have
\begin{align*}
    \mathbb{E}\left[\big[\mathbb{E}_n[\|\mathbf{\tilde{h}}(\tilde{e}(X)) - \mathbf{g}(Y)\|^2]  - \mathbb{E}_n[\|\mathbf{s}^*(X)  - \mathbf{g}(Y)\|^2]\big]\right] \leq \inf_{f \in \mathcal{F}_M(B_y)}\mathbb{E}\left[\|f(X) - \mathbf{s}^*(X)\|^2\right];
\end{align*}
the detailed steps match those in Section S.2.5 of \citet{TangLi2025}. For $f \in \mathcal{F}_M(B_y)$ represented by $(e,d)$, we have
\begin{align*}
    \left(f_j(X) - s^*_j(X)\right)^2 & = \left(\mathbb{E}_{\eta}\left[g(T_{B_y}d(e(X),\eta), t_j) - g(d^*(e^*(X),\eta), t_j)\right]\right)^2 \\
    & \leq \mathbb{E}_{\eta}\left[\left(g(T_{B_y}d(e(X),\eta), t_j) - g(d^*(e^*(X),\eta), t_j)\right)^2\right] \\
    & \leq L_g^2\mathbb{E}_{\eta}\left[\|T_{B_y}d(e(X),\eta) - d^*(e^*(X),\eta)\|^2\right] \\
    & \leq  L_g^2\mathbb{E}_{\eta}\left[\|d(e(X),\eta) - d^*(e^*(X),\eta)\|^2\right].
\end{align*}
Here we use that $g$ is Lipschitz continuous with constant $L_g$ by Lemma \ref{lem:g_lipschitz_t_MULTIY}, and the fact that $d^*(e^*(X),\eta) \in [-B_y, B_y]^m$ for the last inequality. The above implies
\begin{align*}
    & \inf_{f \in \mathcal{F}_M(B_y)}\mathbb{E}\left[\|f(X) - \mathbf{s}^*(X)\|^2\right] \leq \\
    & \ ML_g^2\inf\Big\{\mathbb{E}\left[\mathbb{E}_{\eta}\left[\|d(e(X),\eta) - d^*(e^*(X),\eta)\|^2\right]\right]\colon
     e \in \mathcal{F}(p,\ell_1,r_1,q,B_w), d \in \mathcal{F}(q+s,\ell_2,r_2,m,B_w)\Big\}.
\end{align*}

The following lemma is a slight simplification of Corollary S.2 of \citet{TangLi2025}.

\begin{lemma} \label{lem:nn_approximation}
Suppose that $\rho \colon \mathcal{X} \subseteq [-B, B]^{r} \rightarrow \mathbb{R}^d$ is Lipschitz continuous with constant $C$, and $B \geq 1/2$. For any $L, N \in \mathbb{N}$ and some sufficiently large $B_w$, there exists $\phi \in \mathcal{F}(r, 12L + 14, \max\{4r\lfloor N^{1/r}\rfloor + 3r, 12dN + 8d\}, d, B_w)$ such that
\[
    |\phi_j(x)| \leq \sup_{x' \in \mathcal{X}}|\rho_j(x')| + 2CB\sqrt{r}, \quad x \in \mathcal{X}, \ j = 1, \dots, d,
\]
and for any random vector $X$ with support in $[-B,B]^r$,
\[
    \mathbb{E}\left[\|\rho(X) - \phi(X)\|^2\right] \leq 4d(324r+1)C^2B^2N^{-4/r}L^{-4/r}.
\]
\end{lemma}

Assume now that $B_x > 1/2$. Lemma \ref{lem:nn_approximation} yields
\begin{align*}
    \mathbb{E}\left[\mathbb{E}_{\eta}\left[\|d^*(\mathring{e}(X),\eta) - d^*(e^*(X),\eta)\|^2\right]\right] & \leq \mathbb{E}\left[\mathbb{E}_{\eta}\left[mL_d^2\|(\mathring{e}(X),\eta) - (e^*(X),\eta)\|^2\right]\right] \\
    & = mL_d^2\mathbb{E}[\|\mathring{e}(X) - e^*(X)\|^2] \\
    & \leq 4qmL_d^2L_e^2(324p+1)B_x^2N_1^{-4/p}L_1^{-4/p}
\end{align*}
for some function
\[
    \mathring{e} \in \mathcal{F}(p, 12L_1 + 14, \max\{4p\lfloor N_1^{1/p}\rfloor + 3p, 12qN_1+8q\}, q, B_{we}),
\]
with sufficiently large $B_{we}$. For $k = 1, \dots, q$ and $x \in \mathcal{X}$, the function satisfies
\[
    |\mathring{e}_k(x)| \leq M_e + 2L_eB_x\sqrt{p},
\]
where $M_e = \max_{j=1, \dots, q}\sup_{x \in \mathcal{X}}|e_j^*(x)|$. Define $B_z := \max(B_x, M_e + 2L_eB_x\sqrt{p})$. Again by Lemma \ref{lem:nn_approximation}, there exists
\[
    \mathring{d} \in \mathcal{F}(q+s, 12L_2 + 14, \max\{4(q+s)\lfloor N_2^{1/(q+s)} \rfloor + 3(q + s), 12N_2 + 8m\}, m, B_{wd})
\]
with sufficiently large $B_{wd}$, such that for $Z = \mathring{e}(X)$,
\begin{align*}
    \mathbb{E}\left[\mathbb{E}_{\eta}\left[\|\mathring{d}(\mathring{e}(X),\eta) - d^*(\mathring{e}(X),\eta)\|^2\right]\right] & = \mathbb{E}\left[\mathbb{E}_{\eta}\left[\|\mathring{d}(Z,\eta) - d^*(Z,\eta)\|^2\right]\right] \\
    & \leq 4mL_d^2(324(q+s) + 1)B_z^2N_2^{-4/(q+s)}L_2^{-4/(q+s)}.
\end{align*}
Choosing $B_w \geq \max(B_{we}, B_{wd})$, we have
\begin{align*}
    \mathbb{E}[S_{22}] & \leq 2\kappa^mL_g^2\left(\mathbb{E}\left[\|\mathring{d}(\mathring{e}(X),\eta) - d^*(e^*(X),\eta)\|^2\right]\right) + c'M^{-2/m} + \varepsilon/4 \\
    & \leq 4\kappa^mL_g^2\left(\mathbb{E}\left[\|\mathring{d}(\mathring{e}(X),\eta) - d^*(\mathring{e}(X),\eta)\|^2\right] + \mathbb{E}\left[\|d^*(\mathring{e}(X),\eta) - d^*(e^*(X),\eta)\|^2\right]\right) \\
    & \qquad + c'M^{-2/m} + \varepsilon/4 \\
    & \leq 16qm\kappa^m L_g^2L_d^2L_e^2(324p+1)B_x^2N_1^{-4/p}L_1^{-4/p} +  c'M^{-2/m} + \varepsilon/4 \\
    & \qquad + 16m\kappa^m L_g^2 L_d^2(324(q+s) + 1)B_z^2N_2^{-4/(q+s)}L_2^{-4/(q+s)}.
\end{align*}

\subsection{Completing the Proof} \label{sec:completing}
Combining the bounds for $S_{21}$ and $S_{22}$ gives
\begin{align}
    \tilde{\mathcal{L}} & \leq \kappa^{m}n^{-1}C_2^{-1}(1 + \log(C_1) + (\mathcal{S}+1)\log(M) + \mathcal{S}\log(n)) + (c+c')M^{-2/m} \label{eq:c_bound} \\
    & \qquad + 16qm\kappa^m L_g^2L_d^2L_e^2(324p+1)B_x^2N_1^{-4/p}L_1^{-4/p} \label{eq:s22_bound1} \\
    & \qquad + 16m\kappa^m L_g^2 L_d^2(324(q+s) + 1)B_z^2N_2^{-4/(q+s)}L_2^{-4/(q+s)} \label{eq:s22_bound2} \\
    & \qquad + \varepsilon/2. \nonumber
\end{align}
In the bound above, \eqref{eq:s22_bound1} and \eqref{eq:s22_bound2} converge to zero if $\min(N_1L_1, N_2L_2) \rightarrow \infty$. For \eqref{eq:c_bound}, we have
\begin{align*}
     & \log(C_1) = \\
     &\quad  \mathcal{S}\log\left[2560B_g B^{m+1}(\ell_1 + \ell_2 + 4)(\sqrt{p + s}B_x + 1)(2B_w)^{\ell_1 + \ell_2 + 5}(p + s)(r_1 + 2s)^{\ell_1}(q+s)r_2^{\ell_2}m^2\right] \\
    & \qquad - \log\left[((r_1+2s)!)^{\ell_1}((q+s)!)(r_2!)^{\ell_2}(m!)((4m)!)\right] + \log(14).
\end{align*}
Here $\ell_1, r_1, \ell_2, r_2$ are defined as
\begin{align*}
    & \ell_1 = 12L_1 + 14, & \ r_1 = \max\{4p\lfloor N_1^{1/p}\rfloor + 3p, 12qN_1+8q\}, \\
    & \ell_2 = 12L_2 + 14, & \ r_2 = \max\{4(q+s)\lfloor N_2^{1/(q+s)} \rfloor + 3(q + s), 12N_2 + 8m\}.
\end{align*}

Assume now that $L_1, N_1, L_2, N_2, B_w$, and $M$ depend on $n$ and are non-decreasing, as $n$ increases, with $M \rightarrow \infty$, $M = \mathcal{O}(n^{\gamma_M})$, $B_{w} = \mathcal{O}(n^{\gamma_B})$ for some $\gamma_M, \gamma_B > 0$. To lighten the notation, we do not indicate the dependence on $n$ explicitly. Under these assumptions, we have
\begin{align*}
    & \ell_1 = \mathcal{O}(L_1), \ r_1 = \mathcal{O}(N_1), \ \ell_2 = \mathcal{O}(L_2), \ r_2 = \mathcal{O}(N_2), \\
    & \mathcal{S} = \mathcal{O}(L_1N_1^2 + L_2N_2^2), \\
    & \log(C_1) = \mathcal{O}\left((L_1N_1^2 + L_2N_2^2)((L_1 + L_2)\log(n) + L_1\log(N_1) + L_2\log(N_2))\right).
\end{align*}
This implies that the part \eqref{eq:c_bound} in the error is of order
\[
    n^{-1}(L_1N_1^2 + L_2N_2^2)((L_1 + L_2)\log(n) + L_1\log(N_1) + L_2\log(N_2)) +  M^{-2/m}.
\]
A sufficient condition for this to converge to zero is
\[
    \frac{\log(n)(L_1N_1^2 + L_2N_2^2)(L_1\log(N_1) + L_2\log(N_2))}{n} \rightarrow 0, \ n \rightarrow \infty.
\]

\section{Consistency in Multi-environment Setting}  \label{sec:proof_multi}

Below we extend Theorem \ref{thm:consistencyMULTIY} to the setting with multiple data environments. The assumptions are analogous to the single-environment case.

\begin{assumption} \label{assumption:modelENV}
There exist $e^*\colon \mathbb{R}^p \rightarrow \mathbb{R}^q$, functions $d_k^* \colon \mathbb{R}^{q+s} \rightarrow \mathbb{R}^m$ and $\eta_k \in \mathbb{R}^s$ independent of $X_k$ so that the outcome variable $Y_k \in \mathbb{R}^m$ satisfies
\[
    Y_k = d_k^*(e^*(X_k),\eta_k), \quad k = 1, \dots, K.
\]
For each $k$, the observations $(X_{i,k}, Y_{i,k})$, $i = 1, \dots, n$, are independent copies of $(X_k, Y_k)$, and the samples from different environments are independent.
\end{assumption}

In the above assumption, we take equal environment sizes $n$ for simplicity.

\begin{assumption}
The covariates and noise variables, $(X_k, \eta_k)$, satisfy Assumption \ref{assumption:supportMULTIY} for $k = 1, \dots, K$. 
\end{assumption}

\begin{assumption}
The function $e^*$ satisfies Assumption \ref{assumption:encoderMULTIY}.
\end{assumption}

\begin{assumption}
The functions $d_k^*$ satisfy Assumption \ref{assumption:decoderMULTIY} for $k = 1, \dots, K$.
\end{assumption}

\begin{assumption}
The estimator $(\hat{e}, \hat{d}_{1}, \dots, \hat{d}_{K})$ satisfies
\begin{align}
    & (\hat{e}_n, \hat{d}_{1, n}, \dots, \hat{d}_{K,n}) \in \argmin_{(e,d_1, \dots, d_K) \in \mathcal{M}} \sum_{k=1}^K \sum_{i=1}^{n} \mathrm{ES}(P_{e, d_k, X_{i,k}, B_y}, Y_{i,k}),
\end{align}
with $B_y \geq \max\{|d_{j,k}^*(e^*(x),\eta)|\colon (x,\eta) \in [-B_x, B_x]^{p+s}, j = 1, \dots, m, k = 1, \dots, K\}$, and
\[
    \mathcal{M} = \{(e,d_1, \dots, d_K)\colon e \in \mathcal{F}(p, \ell_1, r_1, q, B_w), \ d_k \in \mathcal{F}(q+s, \ell_{2,k}, r_{2,k}, m, B_w), k = 1, \dots, K\},
\]
for $B_w > 0$.
\end{assumption}

\begin{assumption}
The bound $B_w$ is non-decreasing in $n$ with $B_w = \mathcal{O}(n^{\gamma_B})$ for some $\gamma_B > 0$, and
\begin{align*}
    & \ell_1 = 12L_1 + 14, & \ r_1 = \max\{4p\lfloor N_1^{1/p}\rfloor + 3p, 12qN_1+8q\} \\
    & \ell_{2,k} = 12L_{2,k} + 14, & \ r_{2,k} = \max\{4(q+s)\lfloor N_{2,k}^{1/(q+s)} \rfloor + 3(q + s), 12N_{2,k} + 8m\},
\end{align*}
where $L_1, N_1, L_{2,k}, N_{2,k}$ are non-decreasing in $n$ with limits
\begin{align*}
    & \lim_{n \rightarrow \infty} \min(L_1N_1, L_{2,k}N_{2,k}) = \infty, \\ 
    & \lim_{n \rightarrow \infty} \frac{\log(n)(L_1N_1^2 + L_{2,k}N_{2,k}^2)(L_1\log(N_1) + L_{2,k}\log(N_{2,k}))}{n} = 0.
\end{align*}
\end{assumption}

\begin{theorem} \label{thm:consistencyENV}
If the assumptions above hold, then
\[
    \lim_{n \rightarrow \infty} \mathbb{E}\left[\sum_{k=1}^K \int_{\mathcal{X}}\mathcal{D}^2(P_{\hat{e}_n, \hat{d}_{k,n}, x, B_y}, \cond{Y_k}{X_k=x}) \, dP_{X_k}(x)\right] = 0.
\]
\end{theorem}

\subsection{Proof of Theorem \ref{thm:consistencyENV}}
The proof for the multi-environment case is a slight modification of the proof for a single environment. We use the same notation, indicating dependence on the environment $k$ with a subscript whenever necessary. Recall the definitions
\begin{align*}
    & g(y,t) = \frac{\cos(t^{\top}y)}{\|t\|^{(m+1)/2}}, \\
    & s_k^*(x,t) = \mathbb{E}[g(Y_k,t) \mid X_k = x] = \mathbb{E}_{\eta_k}[g(d_k^*(e^*(x),\eta_k), t)] =: h_k^*(e^*(x),t), \\
    & \hat{h}_k(\hat{e}(x),t) = \mathbb{E}_{\eta_k}[g(T_{B_y}\hat{d}_k(\hat{e}(x),\eta_k), t)],
\end{align*}
and the error to be bounded is
\begin{equation*}
    \tilde{\mathcal{L}} = \mathbb{E}\Big[\sum_{k=1}^K\mathbb{E}_{X\sim P_{X_k}}\Big[\int_{\mathbb{R}^m} (s_k^*(X,t) - \hat{h}_k(\hat{e}(X),t))^2 \, dt \Big]\Big].
\end{equation*}
We again decompose the error as the sum $\tilde{\mathcal{L}} = \mathbb{E}[S_{21} + S_{22}]$, where
\begin{align*}
    & S_{21} = \sum_{k=1}^K\Big\{\mathbb{E}\Big[\int_{\mathbb{R}^m} (\hat{h}_k(\hat{e}(X_k),t) - g(Y_k,t))^2 \, dt \mid D_{n}\Big]  - \mathbb{E}\Big[\int_{\mathbb{R}^m} (s_k^*(X_k,t) - g(Y_k,t))^2 \, dt\Big] \\
    & \qquad - 2\Big[\mathbb{E}_{n,k}\Big[\int_{\mathbb{R}^m} (\hat{h}_k(\hat{e}(X_k),t) - g(Y_k,t))^2 \, dt \Big]  - \mathbb{E}_{n,k}\Big[\int_{\mathbb{R}^m} (s_k^*(X_k,t) - g(Y_k,t))^2 \, dt\Big]\Big]\Big\}, \\
    & S_{22} = 2\sum_{k=1}^K\Big[\mathbb{E}_{n,k}\Big[\int_{\mathbb{R}^m} (\hat{h}_k(\hat{e}(X_k),t) - g(Y_k,t))^2 \, dt \Big]  - \mathbb{E}_{n,k}\Big[\int_{\mathbb{R}^m} (s_k^*(X_k,t) - g(Y_k,t))^2 \, dt\Big]\Big].
\end{align*}
Here $D_{n,k} = ((X_{1,k}, Y_{1,k}), \dots, (X_{n, k}, Y_{n, k}))$, $D_n = (D_{n,1}, \dots, D_{n,K})$, the expectation $\mathbb{E}_{n,k}[\cdot]$ is over the empirical distribution $n^{-1}\sum_{i=1}^{n} \delta_{(X_{i,k}, Y_{i,k})}$, and $(X_k, Y_k)$ has the same distribution as $(X_{1,k}, Y_{1,k})$ and is independent of all other random quantities.

The bound on $\mathbb{E}[S_{21}]$ from the single-environment proof can be applied separately to all $K$ summands in the definition of $S_{21}$ above, which gives the upper bound
\[
       \mathbb{E}[S_{21}] \leq \sum_{k=1}^K \kappa^{m}n^{-1}C_2^{-1}(1 + \log(C_{1,k}) + (\mathcal{S}_k+1)\log(M) + \mathcal{S}_k\log(n)) + KcM^{-2/m} + K\varepsilon/4,
\]
where the constants $C_{1,k}$ and $\mathcal{S}_k$ depend on the environment through $\ell_{2,k}$ and $r_{2,k}$. For $S_{22}$, the same arguments as at the beginning of Section \ref{sec:s22} show that $\mathbb{E}[S_{22}]$ is bounded from above by $K(c'M^{-2/m} + \varepsilon/4)$ plus
\[
    2\kappa^mL_g^2 \cdot \inf_{e, d_1, \dots, d_K}\sum_{k=1}^K \mathbb{E}\left[\mathbb{E}_{\eta_k}[\|d_k(e(X_k), \eta_k) - d^*_k(e^*(X_k),\eta_k)\|^2]\right],
\]
where the infimum is over functions in the class
\[
    \{(e,d_1, \dots, d_K)\colon e \in \mathcal{F}(p, \ell_1, r_1, q, B_w), \ d_k \in \mathcal{F}(q+s, \ell_{2,k}, r_{2,k}, m, B_w), k = 1, \dots, K\}.
\]
For any fixed $\mathring{e}, \mathring{d}_1, \dots, \mathring{d}_K$ in the function class, we have 
\begin{align*}
    & \inf_{e, d_1, \dots, d_K}\sum_{k=1}^K \mathbb{E}\left[\mathbb{E}_{\eta_k}[\|d_k(e(X_k), \eta_k) - d^*_k(e^*(X_k),\eta_k)\|^2]\right] \\
    & \leq \sum_{k=1}^K \mathbb{E}\left[\mathbb{E}_{\eta_k}[\|\mathring{d}_k(\mathring{e}(X_k), \eta_k) - d^*_k(e^*(X_k),\eta_k)\|^2]\right] \\
    & \leq 2\sum_{k=1}^K \Big(\mathbb{E}\left[\mathbb{E}_{\eta_k}[\|\mathring{d}_k(\mathring{e}(X_k), \eta_k) - d^*_k(\mathring{e}(X_k),\eta_k)\|^2]\right] \\
    & \qquad + \mathbb{E}\left[\mathbb{E}_{\eta_k}[\|d^*_k(\mathring{e}(X_k), \eta_k) - d^*_k(e^*(X_k),\eta_k)\|^2]\right]\Big) \\
    & \leq 2\sum_{k=1}^K \mathbb{E}\left[\|\mathring{d}_k(\mathring{e}(X_k), \eta_k) - d^*_k(\mathring{e}(X_k),\eta_k)\|^2\right] + 2mKL_d^2 \mathbb{E}[\|\mathring{e}(X_W) - e^*(X_W)\|^2],
\end{align*}
where $X_W$ has the mixture distribution with $X_W \mid W = k \sim P_{X_k}$ and $W$ uniformly distributed on $\{1, \dots, K\}$. Lemma \ref{lem:nn_approximation} can now be applied to find a suitable $\mathring{e}$ to bound the approximation error $\mathbb{E}[\|\mathring{e}(X_W) - e^*(X_W)\|^2]$. Given $\mathring{e}$, each term $\mathbb{E}\left[\|\mathring{d}_k(\mathring{e}(X_k), \eta_k) - d^*_k(\mathring{e}(X_k),\eta_k)\|^2\right]$ in the upper bound is again bounded by Lemma \ref{lem:nn_approximation}, with neural network parameters $\ell_{2,k}, r_{2,k}$ depending on $L_{2,k}, N_{2,k}$. The proof is concluded with the same steps as in the case $K = 1$, where we use that the rates on the parameters apply to all $K$ environments.

\section{Alternative Dimension Selection Criterion} \label{app:mixture}

A computationally attractive alternative to the criterion of Section \ref{sec:dimension_selection} has been suggested by \citet{ShenMeinshausen2024}: for a given maximal $D$ for the input $(e(X), \eta)$ of $d$, they compute a weighted mixture loss
\begin{align*}
    & \sum_{q=1}^Dw_q\Big( \frac{1}{n}\sum_{i=1}^n\Big(\frac{1}{N}\sum_{j=1}^N \|d(e_{1:q}(X_i),\eta_{i,j; \, (q+1):D}) - Y_i\| -  \\
    & \qquad \frac{1}{2N(N-1)}\sum_{j,k=1}^N \|d(e_{1:q}(X_i),\eta_{i,j; \, (q+1):D}) - d(e_{1:q}(X_i),\eta_{i,k;(q+1):D})\|\Big)\Big),
\end{align*}
where $x_{a:b}$ denotes the subvector of $x$ from indices $a$ to $b$, the $w_1, \dots, w_D \geq 0$ are weights, and $e(x), \eta \in \mathbb{R}^D$ are expanded dimension reduction vector and noise variables from which sub-vectors of different lengths are extracted. The loss simultaneously considers different distributions of the total dimension $D$ to $e_{1:q}(x) \in \mathbb{R}^q$ and to the noise $\eta_{(q+1):D} \in \mathbb{R}^{D-q}$. Minimizing this mixture loss during training makes it possible to select the desired dimension reduction by extracting the first $q$ components, $e_{1:q}(x)$, and filling up the remaining $D-q$ components with noise $\eta_{(q+1):D}$.

\section{Additional Experiments} \label{app:experiments}

\subsection{Simulation Settings} \label{app:setting}

This section describes variations of the simulation examples from Section \ref{sec:simulations}.

\paragraph{Dense-linear.}
We let \(X\sim \mathcal{N}_{80}(0,I_{80})\), and take the sufficient dimension reduction
\[
Z=\Theta X,\qquad \Theta\in\mathbb{R}^{3\times 80}, \qquad \Theta\Theta^\top=I_3.
\]
Thus, all coordinates of $X$ contribute to the sufficient representation. The response is generated as
\[
Y=\mu(Z_{1:3})+\sigma(Z_{1:3})\odot \varepsilon,
\qquad \varepsilon\sim \mathcal{N}_4(0,I_4),
\]
where $\mu(\cdot)$ and $\sigma(\cdot)$ are the same mean and scale functions as in Section \ref{sec:Engression}, with $X_{1:3}$ replaced by $Z_{1:3}$.

\paragraph{Dense-nonlinear.}
Similarly to ``dense-linear'', we first form five dense linear projections
\[
Z=\Theta X,\qquad \Theta\in\mathbb{R}^{5\times 100}, \qquad \Theta\Theta^\top=I_5.
\]
The response depends on $X$ through the nonlinear three-dimensional representation
\[
S_1=Z_1^2+0.5\sin(Z_2),\qquad
S_2=Z_2Z_4,\qquad
S_3=\cos(Z_3-Z_5)+0.25Z_4^2.
\]
Then
\[
Y=\mu(S_{1:3})+\sigma(S_{1:3})\odot \varepsilon,
\qquad \varepsilon\sim \mathcal{N}_4(0,I_4),
\]
using the same mean and scale functions as before, with $Z_{1:3}$ replaced by $S_{1:3}$. This creates a dense nonlinear sufficient representation while avoiding dependence on only a few raw coordinates of $X$.

\paragraph{Heavy-tailed.}
For the simulation in Section \ref{sec:comp}, the data are generated by
$$
Y = \cos(b_1^\top X) + \varepsilon, \quad \text{where }
b_1 = \frac{1}{\sqrt{6}}(1,1,1,1,1,1,0,\ldots,0)^\top, \quad
X \sim \mathcal{N}_{20}(0,\Sigma), \quad \Sigma_{ij} = 0.5^{|i-j|}.
$$
\(\varepsilon\) follows a generalized normal distribution with shape parameter \(c=0.5\) and variance \(0.25\). 

\paragraph{Mixture.} 
The mixture design creates a genuinely distributional prediction problem: \(Y\mid X\) is multimodal, so a method must learn more than the conditional mean to achieve a good energy score. Since the mixture weights, component means, and variances all depend only on the low-dimensional representation, this setting tests whether SRDR can recover the sufficient structure for the full conditional distribution.

For the mixture simulation, the data are generated as 
\(X \sim \mathcal{N}_{40}(0,I_{40}), Z=(X_1,X_2,X_3), \) and 
\[ 
    C\mid X \sim \mathrm{Bernoulli}\{\pi(X_{1:3})\}, \qquad \pi(Z)=\mathrm{sigmoid}(0.9Z_1-0.7Z_2+0.5Z_3). 
\] 
Conditional on \(C\), the response is generated by
\[
    Y = \mu_C(X_{1:3})+\sigma_C(X_{1:3})\odot \varepsilon,
    \qquad \varepsilon\sim \mathcal{N}_4(0,I_4).
\]
The two component means are
\[
\mu_0(X_{1:3})=
\begin{pmatrix}
\sin(X_1)\\
\cos(X_2)\\
X_3\\
0.5X_1X_2
\end{pmatrix},
\qquad
\mu_1(X_{1:3})=
\begin{pmatrix}
\sin(X_1)+1.5+0.25X_3\\
\cos(X_2)-1.0+0.25X_1\\
-X_3+0.5\sin(X_2)\\
0.5X_1X_2+0.75\tanh(X_3)
\end{pmatrix}.
\]
The two component standard deviations are
\[
\sigma_0(X_{1:3})
=
0.12+0.18\,\mathrm{sigmoid}
\begin{pmatrix}
X_1\\
X_2\\
X_3\\
X_1+X_2
\end{pmatrix},
\qquad
\sigma_1(X_{1:3})
=
0.18+0.22\,\mathrm{sigmoid}
\begin{pmatrix}
X_2-X_3\\
X_1+X_3\\
X_1-X_2\\
X_3
\end{pmatrix},
\]
where the sigmoid function is applied componentwise. Hence \(Y\mid X\) is multimodal and depends on \(X\) only through \(X_{1:3}\).

\paragraph{Nuisance.} 
The correlated nuisance design is of interest because the irrelevant covariates are correlated with the true signal, making it harder to distinguish predictive structure from nuisance variation. This tests whether SRDR can focus on the response-relevant reduced representation rather than overusing the correlated noise.

We first generate the signal coordinates \(S=(S_1,S_2,S_3)\sim \mathcal{N}_3(0,I_3).\) The observed covariates are \(X_1=S_1, X_2=S_2, X_3=S_3\), and for 
\(j=4,\ldots,40\), 
\[ 
    X_j = A_j^\top S+\eta_j, 
    \qquad A_j\sim \mathcal{N}_3(0,0.35^2 I_3), 
    \qquad \eta_j\sim \mathcal{N}(0,1). 
\] 
The response is generated as 
\[ 
Y = \mu(S)+\sigma(S)\odot \varepsilon, \qquad \varepsilon\sim \mathcal{N}_4(0,I_4), 
\] 
where 
\[ 
\mu(S)= \begin{pmatrix} \cos(S_1+S_2)\\ \sin(S_2S_3)\\ S_1-0.5S_3\\ \tanh(S_1S_2)+0.25S_3 \end{pmatrix}, 
\qquad
\sigma(S)= \begin{pmatrix} 0.10+0.20\,\mathrm{sigmoid}(S_1)\\ 0.15+0.20\,\mathrm{sigmoid}(S_2)\\ 0.20+0.20\,\mathrm{sigmoid}(S_3)\\ 0.25+0.20\,\mathrm{sigmoid}(S_1-S_2+S_3) \end{pmatrix}. 
\] 
Although the nuisance coordinates \(X_4,\ldots,X_{40}\) are correlated with the signal, they contain no additional information about \(Y\) once \((X_1,X_2,X_3)\) is known.

\paragraph{Heteroscedastic.}
This setting is adapted from \citet{TangLi2025}. We draw $X \in \mathbb{R}^p$ with $p=100$, where each coordinate $X_j$ is independently distributed as $\mathcal{N}(0.2, 0.5)$. The true sufficient reduction is $f(X) = (f_1(X), f_2(X))$, where
\[
f_1(X) = \sin\!\left(\frac{(X_1 + X_2)\pi}{10}\right) + X_1^2, 
\quad
f_2(X) = 2\sin^2\!\left(\frac{(X_3 + X_4)\pi}{10}\right) + X_3^2.
\]
The response is generated as $Y = f_1(X) + f_2(X)\,\varepsilon$, where $\varepsilon \sim \mathcal{N}(0,1)$ is independent of $X$. 

\subsection{Results} \label{app:Engression}

We compare the test energy score of SRDR-linear, SRDR-nonlinear, and engression in four additional settings. The results are shown in Figure~\ref{fig:synthetic_additional}. The conclusion remains similar: SRDR improves upon engression at smaller sample sizes and remains competitive at larger sample sizes when \(q\) is at least $q^*$. SRDR also exhibits an elbow pattern, except that SRDR-linear suffers from misspecification and struggles in identifying $q^*$ in the dense-nonlinear setting.

\begin{figure}[!htp]
\centering
\begin{subfigure}[t]{\textwidth}
    \centering
    \includegraphics[width=\textwidth]{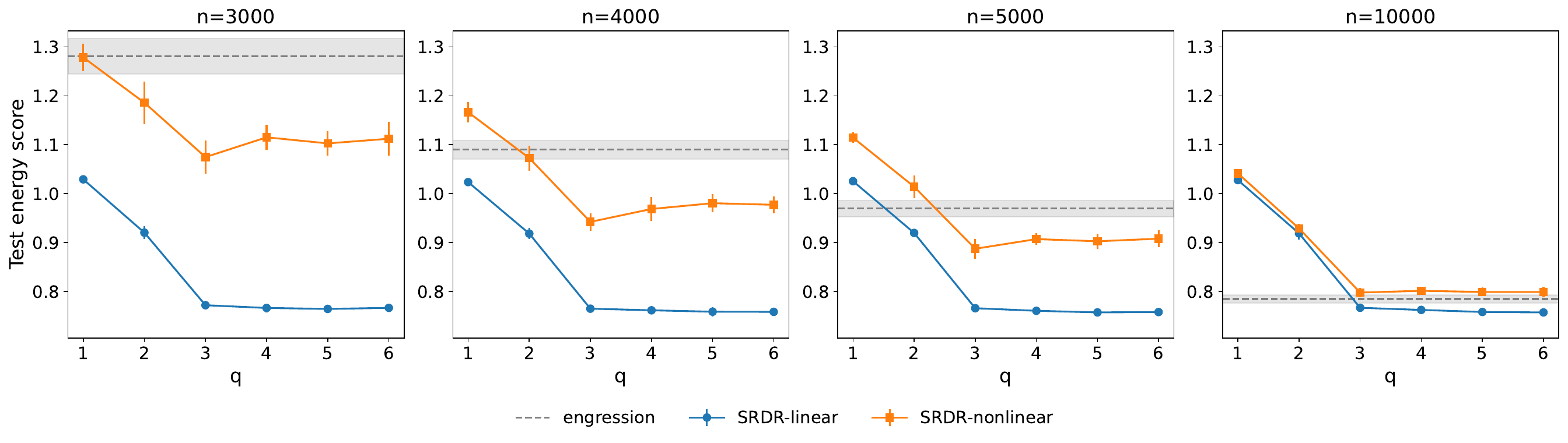}
    \caption{Mixture.}
    \label{fig:mixture}
\end{subfigure}
\vspace{0.5em}
\begin{subfigure}[t]{\textwidth}
    \centering
    \includegraphics[width=\textwidth]{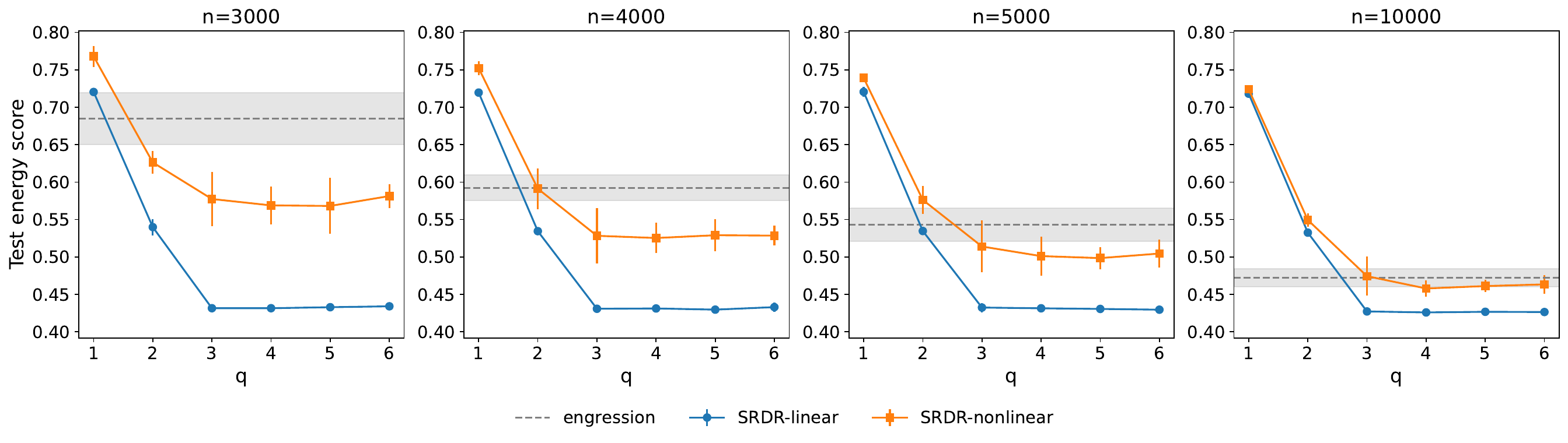}
    \caption{Nuisance.}
    \label{fig:nuisance}
\end{subfigure}
\vspace{0.5em}
\begin{subfigure}[t]{\textwidth}
    \centering
    \includegraphics[width=\textwidth]{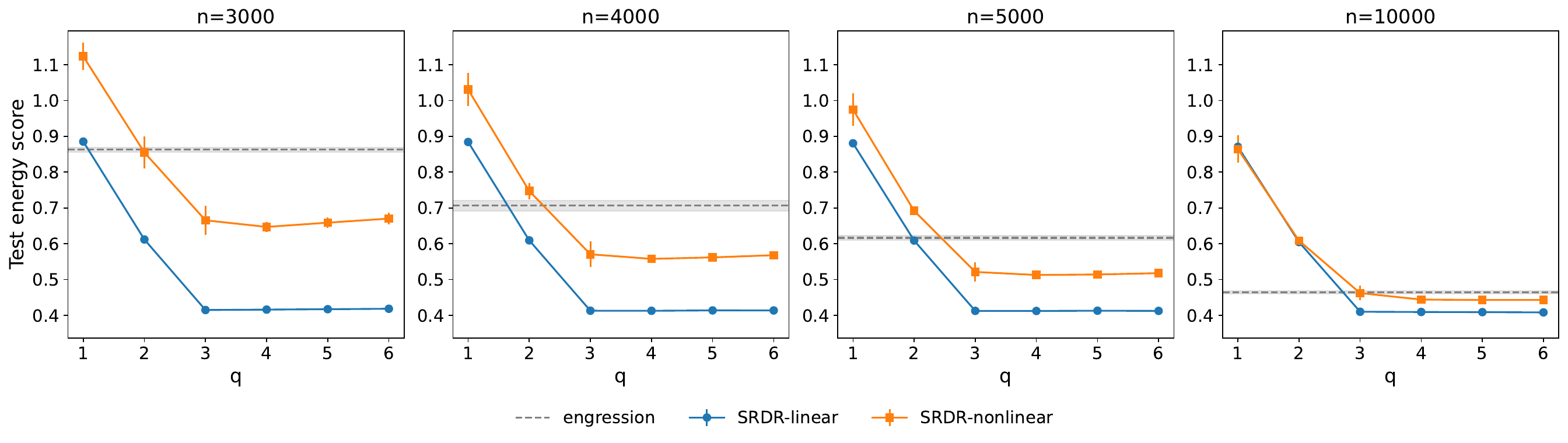}
    \caption{Dense-linear.}
    \label{fig:dense}
\end{subfigure}
\vspace{0.5em}
\begin{subfigure}[t]{\textwidth}
    \centering
    \includegraphics[width=\textwidth]{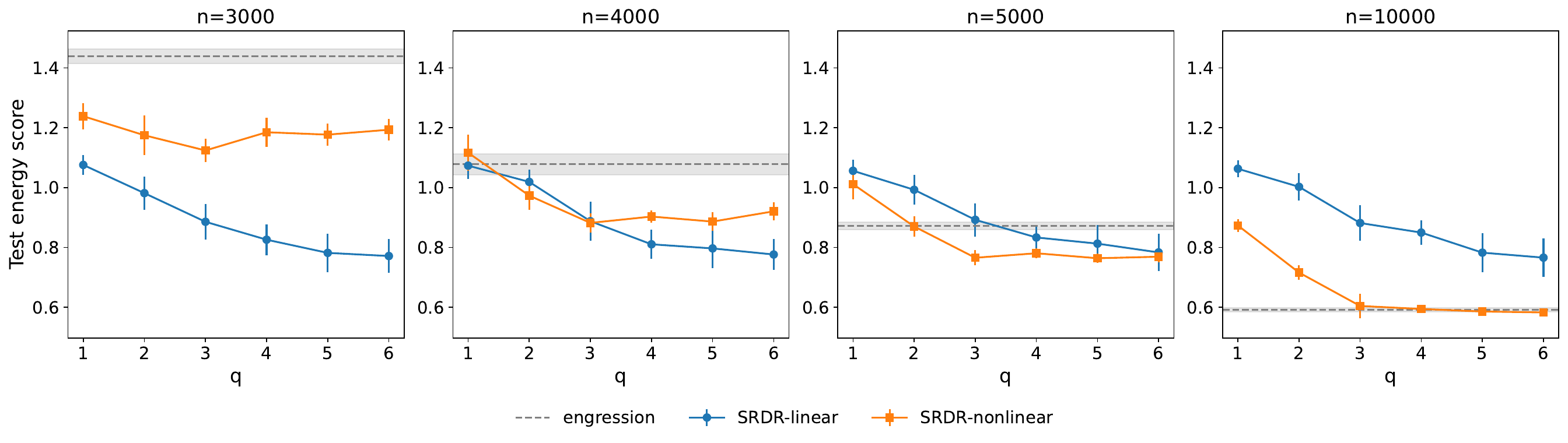}
    \caption{Dense-nonlinear.}
    \label{fig:dense-nonlinear}
\end{subfigure}

\caption{Additional results of the synthetic experiments.}
\label{fig:synthetic_additional}
\end{figure}

\begin{figure}[!ht]
\centering
\includegraphics[width=0.6\textwidth]{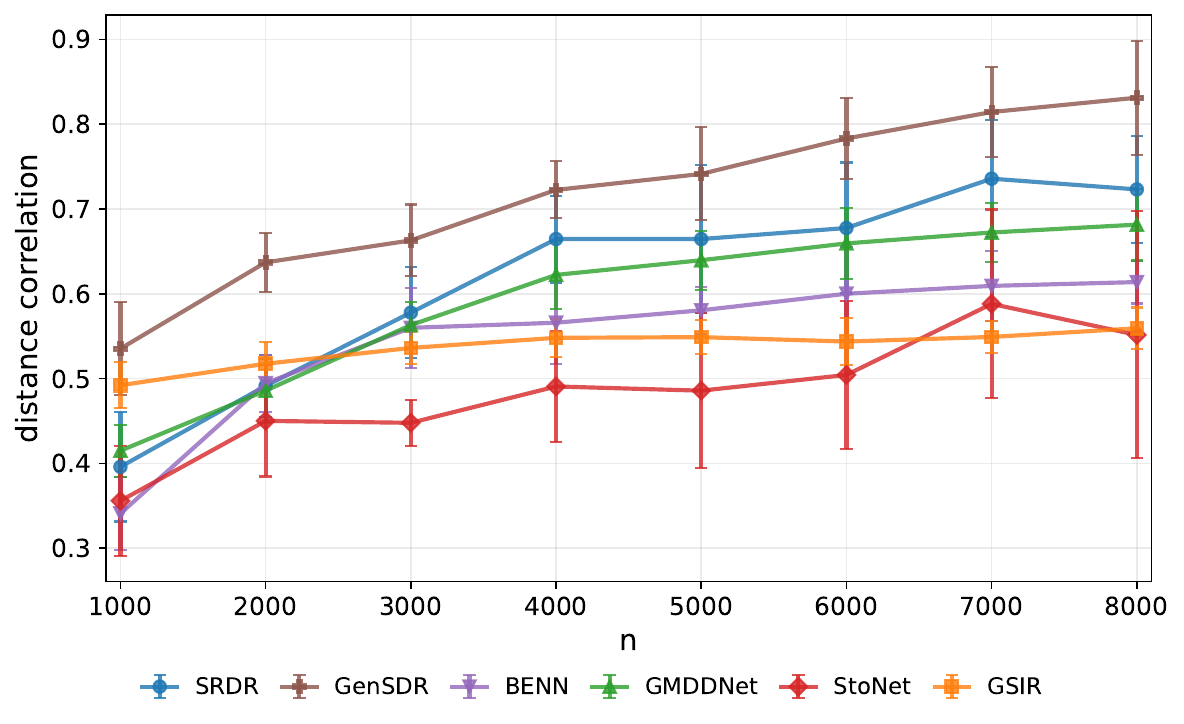}
\caption{$\mathrm{dCor}((f_1(X),f_2(X)), \hat e(X))$ for the heteroscedastic simulation model. Results are reported as mean with error bar over 10 replicates. }
\label{fig:div_dcor}
\end{figure}

We also compare the performance of GSIR, GMDDNet, StoNet, BENN, GenSDR, and SRDR in recovering sufficient reductions under the heteroscedastic setting. For each method, we learn a two-dimensional reduction and evaluate its quality by computing the distance correlation with the true sufficient reduction $f(X)$ on an independent test set of $2000$ data points. We consider sample sizes $n \in \{1000, 2000, \ldots, 8000\}$. For SRDR and GenSDR, we use a three-layer reduction network with 20 hidden units per layer. For the other methods, we adopt the settings of \citet{TangLi2025}. The results are summarized in Figure~\ref{fig:div_dcor}. Overall, for relatively large sample sizes, SRDR is among the most competitive methods.

\subsection{Additional results for Superconductivity} \label{app:supercond}

We present two additional empirical results on the superconductivity dataset.

\begin{table}[!t]
\centering
\caption{Superconductivity prediction-interval results for SRDR, GenSDR, and engression. Entries are reported as mean (sd) across replicates.}
\label{tab:supercond_pi}
\begin{tabular}{cccccc}
\toprule
Method & $q$ & 90\% cov. & 95\% cov. & 90\% width & 95\% width \\
\midrule
SRDR & 1 & 0.878 (0.007) & 0.928 (0.005) & 29.945 (0.483) & 36.212 (0.755) \\
SRDR & 2 & 0.873 (0.006) & 0.923 (0.005) & 29.080 (0.532) & 35.015 (0.915) \\
SRDR & 4 & 0.873 (0.008) & 0.924 (0.005) & 28.541 (0.693) & 34.440 (0.995) \\
SRDR & 8 & 0.863 (0.007) & 0.914 (0.006) & 28.069 (0.553) & 33.847 (0.655) \\
SRDR & 16 & 0.865 (0.010) & 0.915 (0.007) & 27.626 (0.790) & 33.305 (0.910) \\
SRDR & 32 & 0.849 (0.010) & 0.902 (0.008) & 26.165 (0.666) & 31.393 (0.823) \\
SRDR & 64 & 0.850 (0.008) & 0.900 (0.007) & 25.746 (0.701) & 30.764 (0.822) \\
SRDR & 82 & 0.845 (0.010) & 0.897 (0.006) & 24.942 (0.701) & 29.742 (0.860) \\
\midrule
GenSDR & 1 & 0.874 (0.019) & 0.934 (0.011) & 29.920 (1.493) & 36.592 (1.756) \\
GenSDR & 2 & 0.876 (0.018) & 0.932 (0.011) & 29.285 (1.542) & 35.358 (1.671) \\
GenSDR & 4 & 0.873 (0.022) & 0.929 (0.015) & 28.401 (1.053) & 34.356 (1.303) \\
GenSDR & 8 & 0.876 (0.023) & 0.931 (0.016) & 28.430 (1.176) & 34.361 (1.450) \\
GenSDR & 16 & 0.877 (0.017) & 0.929 (0.011) & 28.675 (1.291) & 34.402 (1.385) \\
GenSDR & 32 & 0.867 (0.023) & 0.925 (0.016) & 28.148 (2.102) & 33.769 (2.269) \\
GenSDR & 64 & 0.877 (0.007) & 0.931 (0.005) & 28.811 (1.468) & 34.457 (1.395) \\
GenSDR & 82 & 0.873 (0.028) & 0.928 (0.020) & 27.939 (1.109) & 33.491 (1.363) \\
\midrule
engression & 82 & 0.844 (0.008) & 0.895 (0.004) & 25.450 (0.973) & 30.284 (1.117) \\
\bottomrule
\end{tabular}
\end{table}

Table~\ref{tab:supercond_pi} reports the prediction interval performance for the three generative methods on the superconductivity data. SRDR and GenSDR both produce intervals whose empirical coverages are slightly below the nominal levels, with coverage decreasing as $q$ increases. Overall, the three methods exhibit similar prediction interval performance, with no clear advantage for any single method.

\begin{figure}[!htp]
\centering

\begin{subfigure}[t]{\textwidth}
    \centering
    \includegraphics[width=0.95\linewidth]{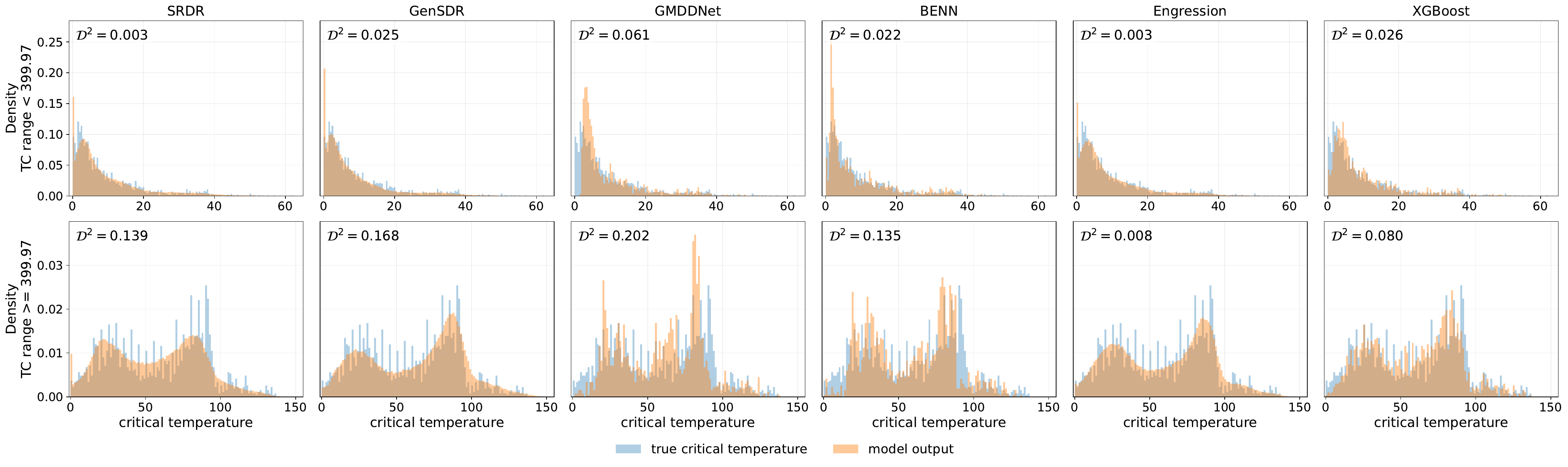}
    \caption{$q = 2$}
    \label{fig:supercond_hist_2}
\end{subfigure}
\vspace{0.5em}
\begin{subfigure}[t]{\textwidth}
    \centering
    \includegraphics[width=0.95\linewidth]{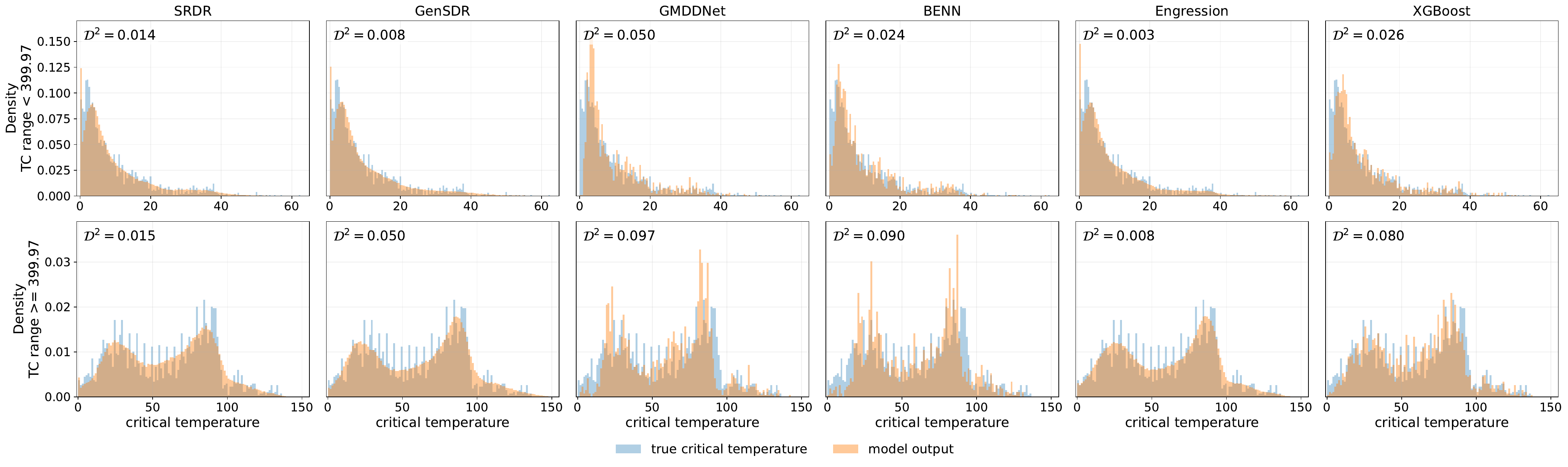}
    \caption{$q = 4$}
    \label{fig:supercond_hist_4}
\end{subfigure}
\vspace{0.5em}
\begin{subfigure}[t]{\textwidth}
    \centering
    \includegraphics[width=0.95\linewidth]{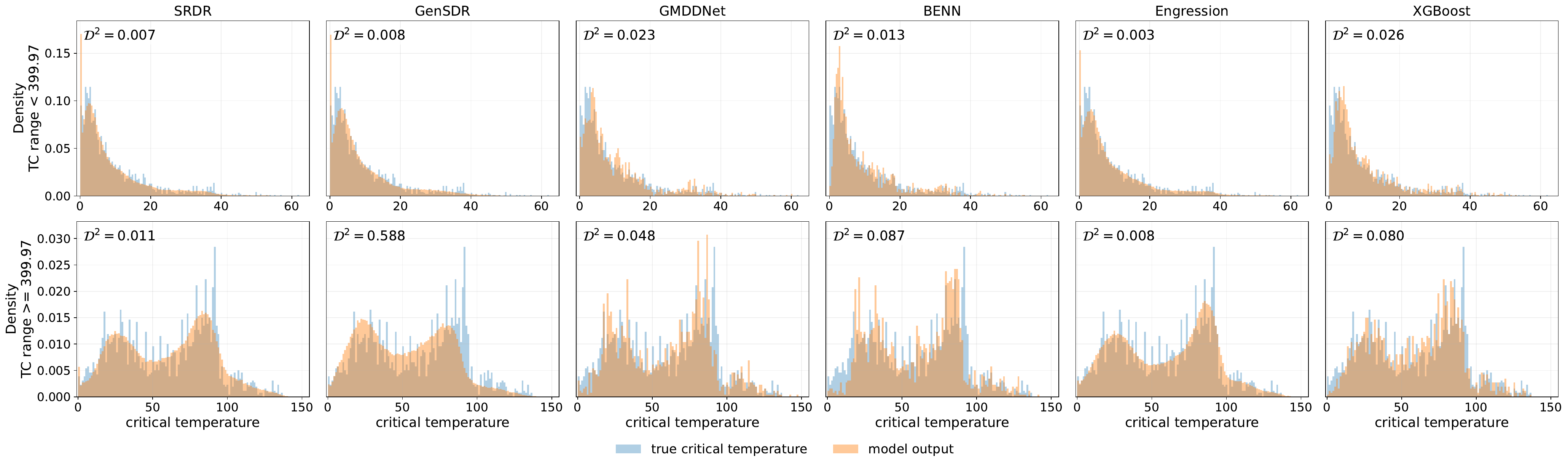}
    \caption{$q = 8$}
    \label{fig:supercond_hist_8}
\end{subfigure}

\caption{More histograms of the model outputs for the superconductivity experiment.}
\label{fig:supercond_hist_more}
\end{figure}

\begin{figure}[!htp]
\centering
\begin{subfigure}[t]{\textwidth}
    \centering
    \includegraphics[width=0.95\linewidth]{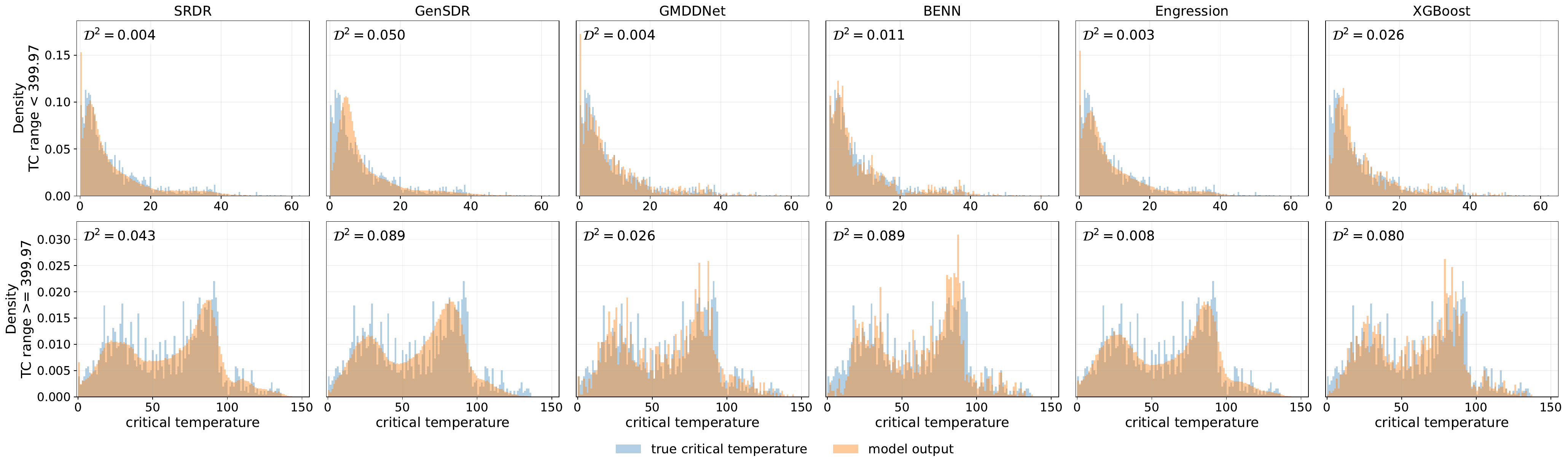}
    \caption{$q = 16$}
    \label{fig:supercond_hist_16}
\end{subfigure}
\begin{subfigure}[t]{\textwidth}
    \centering
    \includegraphics[width=0.95\linewidth]{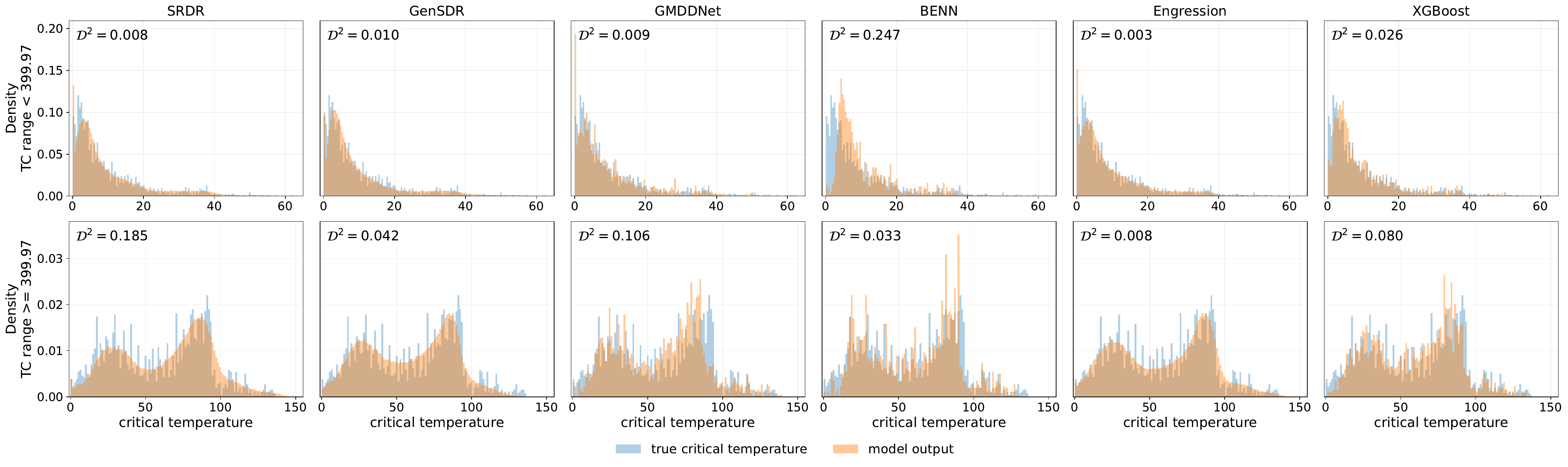}
    \caption{$q = 32$}
    \label{fig:supercond_hist_32}
\end{subfigure}
\begin{subfigure}[t]{\textwidth}
    \centering
    \includegraphics[width=0.95\linewidth]{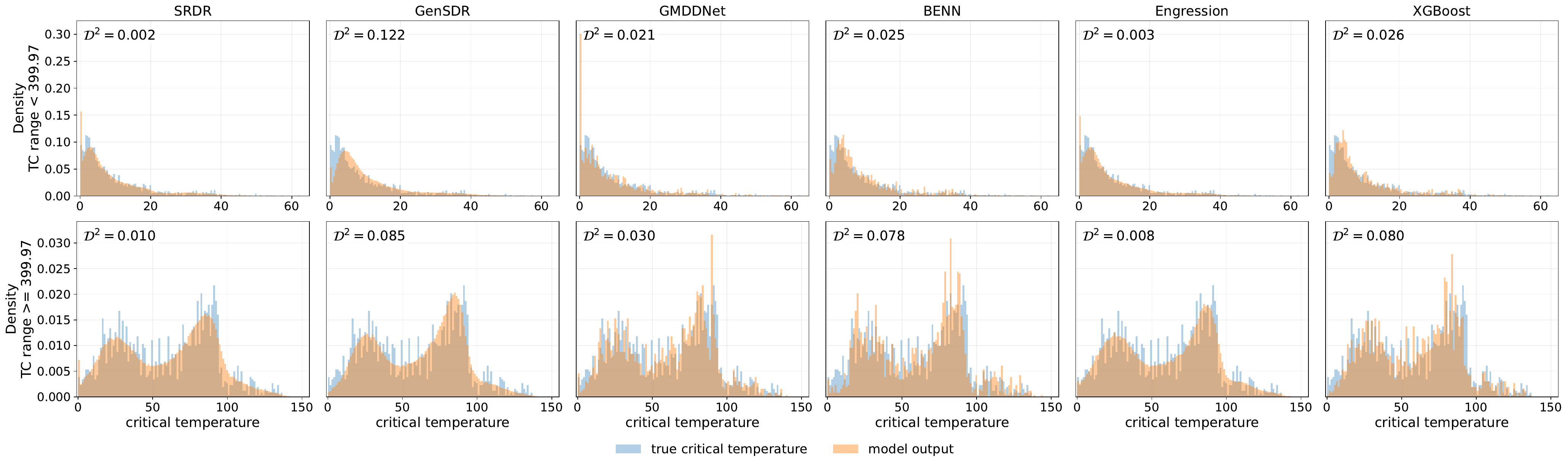}
    \caption{$q = 64$}
    \label{fig:supercond_hist_64}
\end{subfigure}
\begin{subfigure}[t]{\textwidth}
    \centering
    \includegraphics[width=0.95\linewidth]{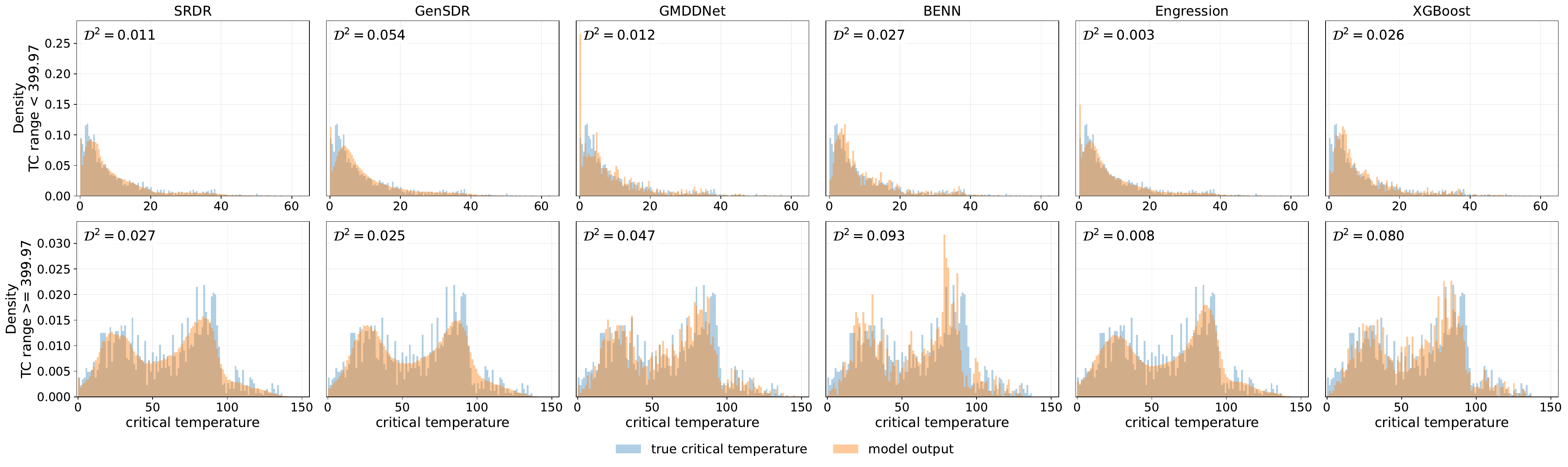}
    \caption{$q = 82$}
    \label{fig:supercond_hist_82}
\end{subfigure}

\caption{Further histograms of the model outputs for the superconductivity experiment, for larger $q$.}
\label{fig:supercond_hist_moremore}
\end{figure}

Figures~\ref{fig:supercond_hist_more} and~\ref{fig:supercond_hist_moremore} provide additional histogram comparisons for other choices of \(q\), complementing the main results in Figure~\ref{fig:supercond_hist}. The overall pattern is consistent across representation dimensions: SRDR continues to recover the main distributional features of the observed outcomes, while the full-feature engression and XGBoost benchmarks remain competitive. In contrast, GMDDNet and BENN often produce more concentrated output distributions.

\subsection{Additional results for CT Slice Localization} \label{app:ct}

We provide additional distributional comparisons for the CT slice localization dataset in Figures~\ref{fig:ct_hist_more} and \ref{fig:ct_hist_moremore}. These figures show patient-level histograms for additional choices of \(q\). The results are consistent with the main text: SRDR generally captures the support and shape of the observed patient-level distributions, while some competing methods produce overly concentrated predictions for smaller \(q\).

We also compare the 90\% and 95\% prediction intervals produced by GenSDR, SRDR and engression. The average coverages and prediction interval widths are reported in Table~\ref{tab:ct_pi}. SRDR yields shorter intervals while achieving similar coverage as engression. GenSDR, on the other hand, is more conservative and tends to overcover with relatively wider intervals.

\begin{table}[!t]
\centering
\caption{CT prediction-interval results for SRDR, GenSDR, and engression. Entries are reported as mean (sd) across replicates.}
\label{tab:ct_pi}
\begin{tabular}{cccccc}
\toprule
Method & $q$ & 90\% cov. & 95\% cov. & 90\% PI width & 95\% PI width \\
\midrule
SRDR & 3 & 0.857 (0.020) & 0.903 (0.016) & 1.283 (0.054) & 1.520 (0.064) \\
SRDR & 6 & 0.849 (0.022) & 0.897 (0.016) & 1.217 (0.065) & 1.441 (0.077) \\
SRDR & 12 & 0.848 (0.018) & 0.897 (0.014) & 1.190 (0.067) & 1.409 (0.080) \\
SRDR & 24 & 0.825 (0.016) & 0.877 (0.015) & 1.167 (0.087) & 1.381 (0.102) \\
SRDR & 48 & 0.836 (0.020) & 0.887 (0.016) & 1.101 (0.053) & 1.303 (0.062) \\
SRDR & 96 & 0.838 (0.015) & 0.889 (0.012) & 1.104 (0.070) & 1.306 (0.083) \\
SRDR & 192 & 0.837 (0.018) & 0.888 (0.014) & 1.108 (0.058) & 1.311 (0.069) \\
\midrule
GenSDR & 3 & 0.976 (0.009) & 0.991 (0.004) & 4.596 (0.092) & 5.717 (0.144) \\
GenSDR & 6 & 0.976 (0.009) & 0.991 (0.005) & 4.381 (0.293) & 5.400 (0.357) \\
GenSDR & 12 & 0.980 (0.008) & 0.993 (0.003) & 4.362 (0.177) & 5.363 (0.214) \\
GenSDR & 24 & 0.984 (0.005) & 0.994 (0.002) & 4.518 (0.193) & 5.540 (0.220) \\
GenSDR & 48 & 0.977 (0.009) & 0.992 (0.004) & 4.651 (0.160) & 5.704 (0.195) \\
GenSDR & 96 & 0.975 (0.009) & 0.990 (0.004) & 4.651 (0.306) & 5.710 (0.359) \\
GenSDR & 192 & 0.960 (0.021) & 0.984 (0.011) & 5.123 (0.525) & 6.274 (0.624) \\
\midrule
engression & 192 & 0.855 (0.015) & 0.902 (0.011) & 1.998 (0.102) & 2.363 (0.121) \\
\bottomrule
\end{tabular}
\end{table}

\begin{figure}[!t]
\centering

\begin{subfigure}[t]{\textwidth}
    \centering
    \includegraphics[width=0.95\linewidth]{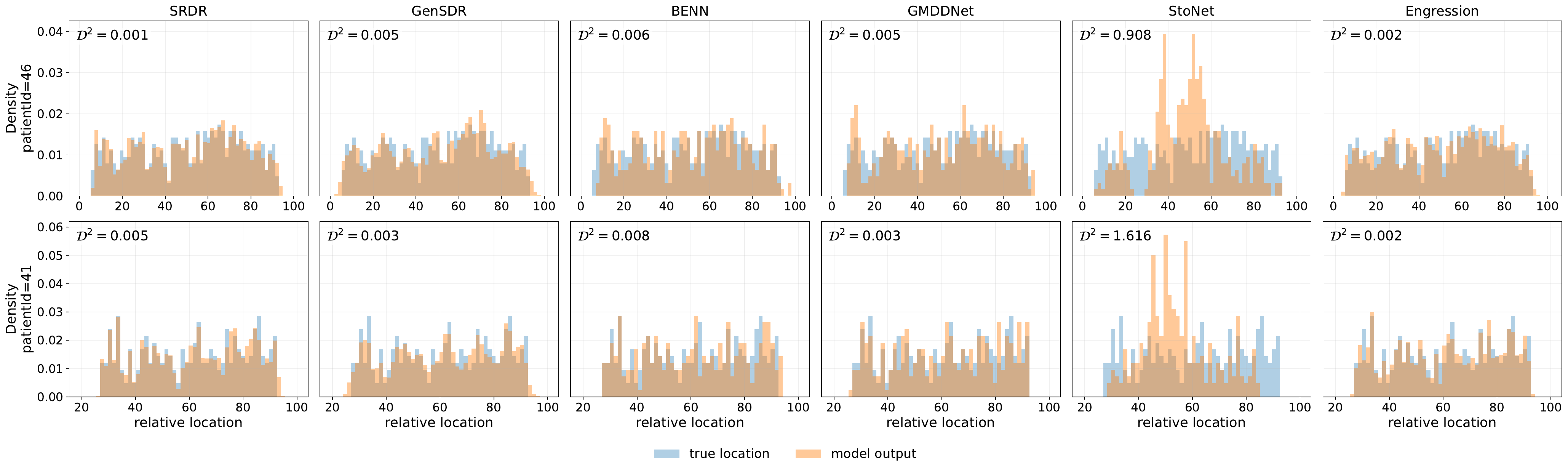}
    \caption{$q=6$}
    \label{fig:ct_hist_6}
\end{subfigure}
\vspace{0.5em}
\begin{subfigure}[t]{\textwidth}
    \centering
    \includegraphics[width=0.95\linewidth]{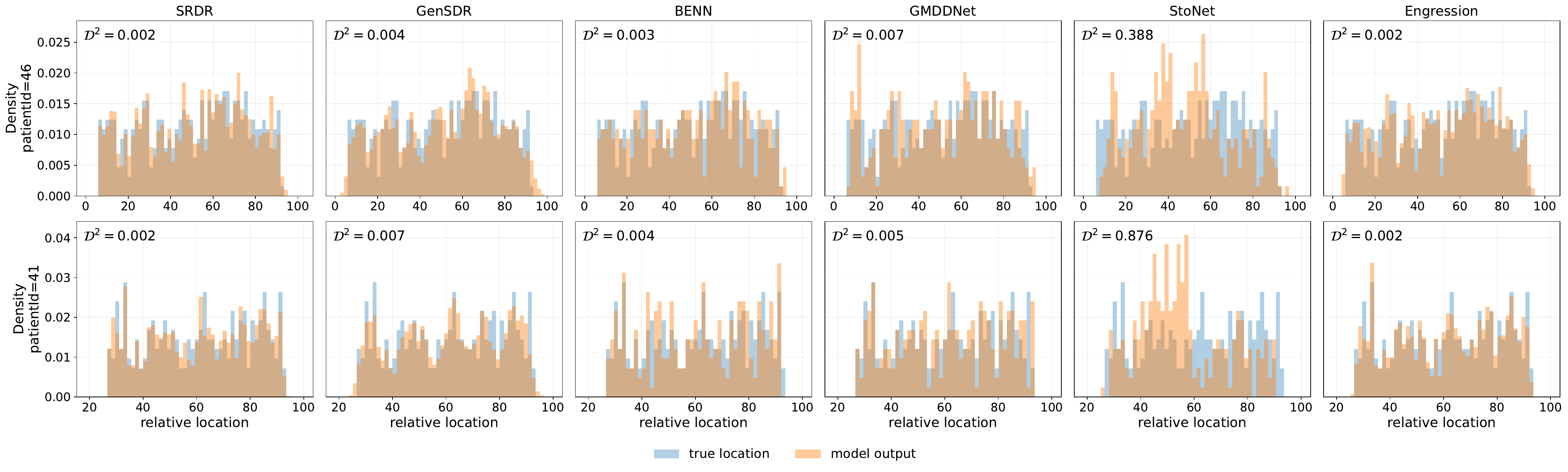}
    \caption{$q=12$}
    \label{fig:ct_hist_12}
\end{subfigure}
\vspace{0.5em}
\begin{subfigure}[t]{\textwidth}
    \centering
    \includegraphics[width=0.95\linewidth]{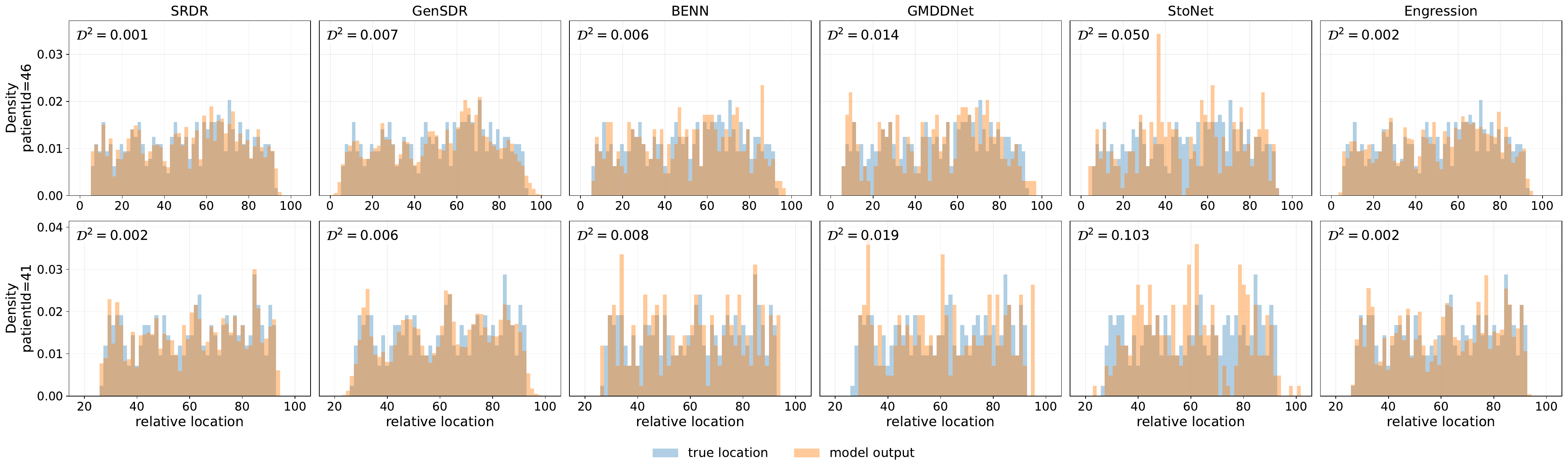}
    \caption{$q = 24$}
    \label{fig:ct_hist_24}
\end{subfigure}

\caption{Histograms for the CT localization experiment.}
\label{fig:ct_hist_more}
\end{figure}

\begin{figure}[!t]
\centering

\begin{subfigure}[t]{\textwidth}
    \centering
    \includegraphics[width=0.95\linewidth]{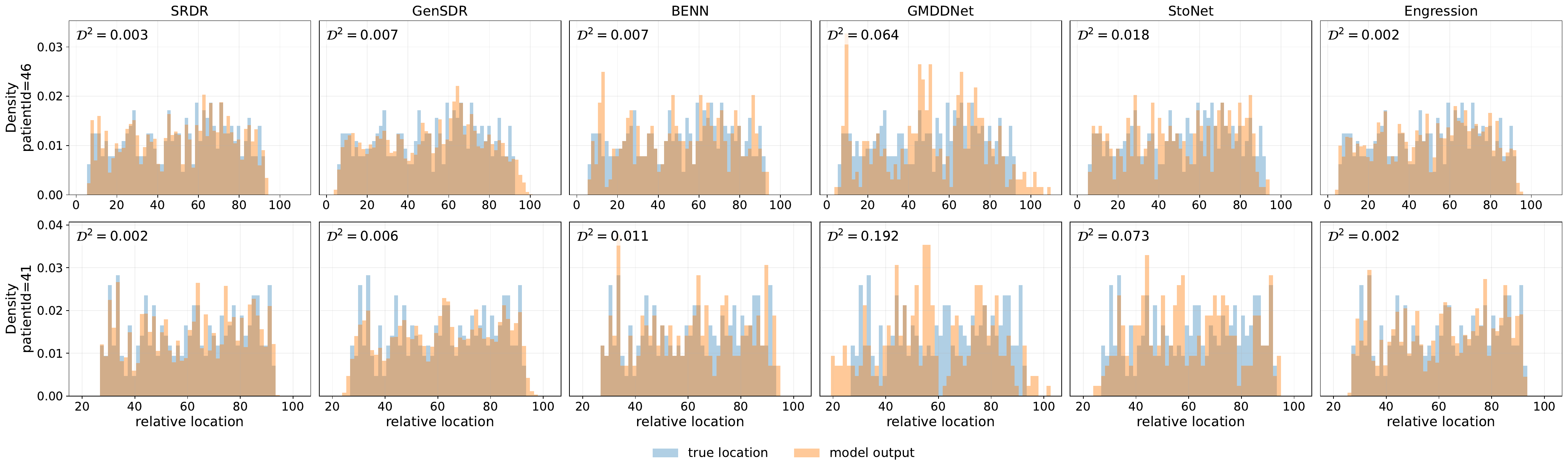}
    \caption{$q=48$}
    \label{fig:ct_hist_48}
\end{subfigure}
\vspace{0.5em}
\begin{subfigure}[t]{\textwidth}
    \centering
    \includegraphics[width=0.95\linewidth]{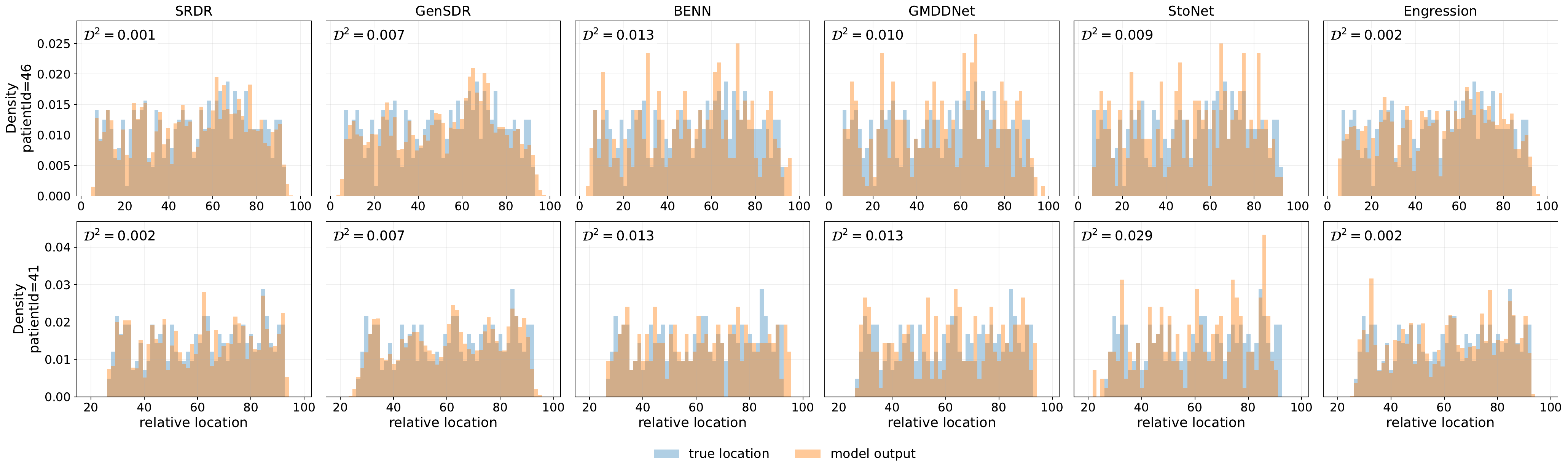}
    \caption{$q=96$}
    \label{fig:ct_hist_96}
\end{subfigure}
\vspace{0.5em}
\begin{subfigure}[t]{\textwidth}
    \centering
    \includegraphics[width=0.95\linewidth]{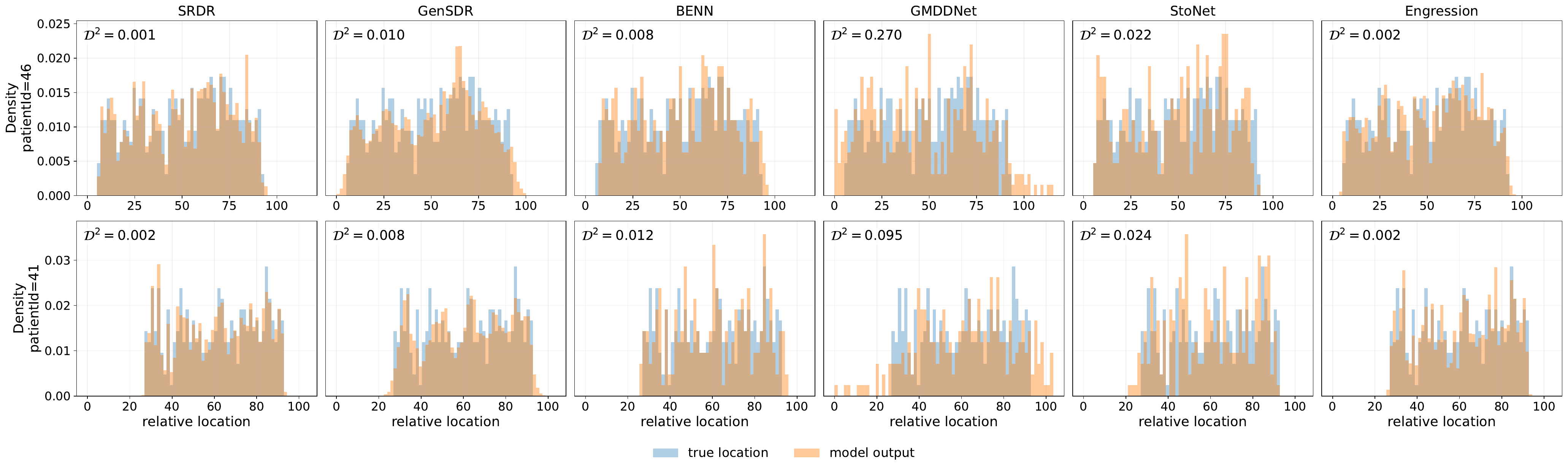}
    \caption{$q=192$}
    \label{fig:ct_hist_192}
\end{subfigure}

\caption{More histograms for the CT localization experiment.}
\label{fig:ct_hist_moremore}
\end{figure}

\subsection{Transfer Learning} \label{app:transfer}

\begin{table}[!ht]
\centering
\caption{CIFAR-10-to-PACS transfer learning classification results. Entries are reported as mean (sd) across seeded PACS resplits. Transfer gap is target accuracy minus scratch accuracy on the same held-out PACS test split.}
\label{tab:transfer}
\begin{tabular}{lcccc}
\toprule
Method & Source val. & Transfer & Scratch & Gap \\
\midrule
\multicolumn{5}{l}{\textit{CIFAR-10 $\rightarrow$ art painting}} \\
SRDR & 0.543 (0.002) & \textbf{0.733 (0.017)} & 0.724 (0.020) & 0.010 (0.017) \\
SRDR-Frozen & 0.543 (0.002) & 0.294 (0.011) & 0.725 (0.008) & -0.431 (0.012) \\
TESR & 0.481 (0.005) & 0.572 (0.032) & 0.779 (0.016) & -0.207 (0.039) \\
\midrule
\multicolumn{5}{l}{\textit{CIFAR-10 $\rightarrow$ cartoon}} \\
SRDR & 0.543 (0.002) & \textbf{0.835 (0.011)} & 0.824 (0.022) & 0.011 (0.019) \\
SRDR-Frozen & 0.543 (0.002) & 0.426 (0.017) & 0.821 (0.015) & -0.395 (0.020) \\
TESR & 0.479 (0.006) & 0.779 (0.037) & 0.882 (0.008) & -0.103 (0.043) \\
\midrule
\multicolumn{5}{l}{\textit{CIFAR-10 $\rightarrow$ photo}} \\
SRDR & 0.543 (0.002) & \textbf{0.827 (0.012)} & 0.824 (0.016) & 0.003 (0.014) \\
SRDR-Frozen & 0.543 (0.002) & 0.490 (0.028) & 0.823 (0.016) & -0.333 (0.039) \\
TESR & 0.480 (0.006) & 0.731 (0.018) & 0.863 (0.013) & -0.131 (0.024) \\
\midrule
\multicolumn{5}{l}{\textit{CIFAR-10 $\rightarrow$ sketch}} \\
SRDR & 0.543 (0.002) & \textbf{0.829 (0.011)} & 0.822 (0.018) & 0.006 (0.022) \\
SRDR-Frozen & 0.543 (0.002) & 0.407 (0.022) & 0.817 (0.021) & -0.410 (0.035) \\
TESR & 0.473 (0.021) & 0.737 (0.026) & 0.878 (0.010) & -0.141 (0.022) \\
\bottomrule
\end{tabular}
\end{table}

We evaluate transfer from CIFAR-10 to the four PACS domains: art painting, cartoon, photo, and sketch. This setting follows the image classification experiment in \citet{GeZhouHuang2025}, where CIFAR-10 provides a single source task with 10 classes and each PACS domain provides a 7-class target task. For each target domain, we repeat the procedure over 10 random target resplits and report the average results in Table~\ref{tab:transfer}.

We compare TESR with two SRDR variants. Enhanced SRDR, listed as ``SRDR'' in Table \ref{tab:transfer}, first pretrains a source model on CIFAR-10 and then adapts it to the PACS target using a frozen source reduction, a trainable target reduction, and a generative prediction network.
SRDR-Frozen instead keeps the pretrained source reduction fixed and trains only the target prediction network on the target domain. For both methods, we also report the accuracies of a model that is trained from scratch on the target domain (column ``Scratch'' in Table \ref{tab:transfer}). The scratch accuracies of TESR are higher than those of SRDR because TESR adopts a convolutional network that exploits the spatial structure of images, whereas SRDR uses flattened image inputs with an MLP dimension reduction network. Scratch accuracies are therefore not directly comparable across methods, and the transfer gap, the difference between the transfer and scratch accuracies of the same method, is the fairer measure of transfer effectiveness.

Table~\ref{tab:transfer} shows that the enhanced SRDR consistently achieves small positive transfer gaps across all four target domains, indicating that transferring the pretrained reduction provides a modest but consistent improvement over training the same SRDR architecture from scratch. In contrast, SRDR-Frozen exhibits substantial negative transfer gaps on every target domain, suggesting that target-specific adaptation is essential for heterogeneous image transfer. TESR also exhibits negative transfer gaps relative to its target-only baseline under our implementation.

\subsection{Implementation Details} \label{app:implementation}

For all experiments involving GMDDNet, we use the successive procedure, which \citet{Chen2024} report to improve performance.

\paragraph{Synthetic (SRDR and engression).}
Across all scenarios, we use the same network architecture and optimization hyperparameters for a fair comparison. For SRDR, the nonlinear dimension reduction variants use a two-layer dimension reduction network, while the linear variant uses a single linear layer; the prediction network has two layers, with hidden width \(100\). Engression is implemented with a four-layer network of hidden width \(100\). All methods use \(s=10\), Adam optimization with learning rate \(10^{-4}\), and \(20{,}000\) training steps. The batch size is capped at \(256\). For each setting, we run \(10\) Monte Carlo replicates and evaluate test energy score using \(100\) generated samples per test point on an independent test set of size \(5000\).

\paragraph{Heavy-tailed.}
All neural-network-based methods use the same \(2\times 20\) dimension reduction network for comparability. Specifically, GMDDNet uses a two-hidden-layer ReLU network. SRDR uses a two-layer prediction network with hidden width \(20\) and \(s=5\). StoNet has a one-dimensional bottleneck after the dimension reduction network. BENN uses \(m=3\) test functions. GenSDR uses a matching two-hidden-layer ReLU velocity network, which takes \((e(X),Y_t,t)\) as input and outputs a scalar velocity.

For optimization, GMDDNet is trained with Adam using learning rate \(10^{-3}\), batch size at most \(100\), and \(100\) training iterations. BENN is trained for \(100\) epochs with learning rate \(10^{-3}\). SRDR is trained for \(50{,}000\) steps using learning rate \(10^{-4}\), weight decay \(10^{-5}\), and batch size capped at \(256\). StoNet is trained for \(200\) epochs with step size \(0.005\), momentum \(0.9\), weight decay \(0.01\), and then fine-tuned for \(30\) epochs with step size \(0.001\). GenSDR is trained for \(50\) epochs using Adam with learning rate \(10^{-3}\), batch size capped at \(256\), \(\tau=0.001\), and no weight decay, using the stochastic-interpolation velocity-matching objective. All experiments use an independent test set of size \(2000\).

\paragraph{CT.}
All experiments use a fixed random \(80/20\) train/test split. All neural-network dimension reduction methods use a one-hidden-layer network with hidden width \(400\). SRDR uses a one-hidden-layer generative prediction network with hidden width \(400\). Engression is used as a full-feature benchmark with a two-layer network of hidden width \(400\). GenSDR uses a one-hidden-layer ReLU velocity network and generates response samples directly from the learned conditional flow. GMDDNet is trained using a one-layer ReLU network, and BENN uses \(m=50\). For GMDDNet, BENN, and StoNet, the downstream regressors are trained with $L_2$ loss for MSE evaluation, or with $L_1$ loss for the MAE report.

SRDR is trained for \(25{,}000\) steps using Adam with batch size \(256\), learning rate \(2\times10^{-4}\), and weight decay \(10^{-4}\). Engression uses the same energy-score training objective on the full feature vector. GenSDR is trained for \(50\) epochs using Adam with learning rate \(10^{-3}\), batch size \(256\), \(\tau=0.001\), and no weight decay. GMDDNet is trained for \(100\) iterations with learning rate \(10^{-3}\) and batch size \(256\). BENN is trained for \(400\) epochs with learning rate \(10^{-3}\). StoNet uses proposal learning rates \((5\times10^{-9},5\times10^{-10})\), sigma list \((10^{-7},10^{-8},10^{-8})\), \(100\) epochs, subsample size \(800\), step size \(0.01\), momentum \(0.9\), and weight decay \(0.01\). Prediction from SRDR, GenSDR, and engression uses Monte Carlo summaries of generated samples.

\paragraph{Superconductivity.}
We use a fixed random 80/20 train/test split and repeat model fitting over 10 replicates on this split. All methods use a two-layer dimension reduction network with hidden width \(200\). SRDR uses a two-layer generative prediction network with width \(200\) and \(s=10\). Engression is used as a full-feature benchmark with a four-layer network of hidden width \(200\). GenSDR uses a matching two-hidden-layer ReLU velocity network. GMDDNet is trained using a ReLU network, and BENN uses \(m=2\). For GMDDNet and BENN, a two-layer DNN regressor with hidden width \(200\) is fitted from the learned reduction to the response; separate regressors are trained with $L_2$ loss for MSE evaluation and $L_1$ loss for MAE evaluation.

All neural generative models are trained with the same main optimization budget whenever applicable. SRDR and engression are trained for \(30{,}000\) steps using Adam with learning rate \(2\times10^{-4}\), batch size \(256\), and weight decay \(5\times10^{-5}\). GenSDR is trained for \(50\) epochs using Adam with learning rate \(10^{-3}\), batch size \(256\), \(\tau=0.001\), and no weight decay. BENN is trained for \(500\) epochs with learning rate \(10^{-3}\), while GMDDNet is trained for \(100\) iterations with learning rate \(10^{-3}\) and batch size \(256\). Prediction from GenSDR, SRDR and engression uses Monte Carlo averaging with \(500\) generated samples and prediction batch size \(512\). The XGBoost baseline follows the full-feature setting with \(374\) trees, learning rate \(0.02\), maximum depth \(16\), subsampling rate \(0.5\), and column subsampling rate \(0.5\). We fit XGBoost separately for the two evaluation criteria, using squared-error loss for MSE and absolute-error loss for MAE.

\paragraph{MNIST.}
GMDDNet, StoNet, BENN, GenSDR, and SRDR are all implemented with a dimension reduction network consisting of four hidden layers of width \(200\). A downstream logistic regression classifier is trained on the learned reduction for GMDDNet, StoNet, and BENN. For SRDR, the \(4\times 200\) generative prediction network uses a softmax output, so generated responses lie on the probability simplex; predictions are obtained by averaging the generated vectors and taking the largest coordinate. For GenSDR, the one-hot labels are treated as Euclidean response vectors: the \(4\times 200\) learned velocity network generates unconstrained \(10\)-dimensional response samples, which are not forced to be one-hot or simplex-valued. We average these generated response vectors and predict by the largest coordinate. All experiments use the full MNIST training set of 60,000 images and are evaluated on the standard test set of 10,000 images.

The optimization hyperparameters are chosen according to the original implementations of the competing methods whenever possible. GMDDNet is trained using Adam with learning rate \(10^{-3}\) for \(100\) iterations. BENN is trained for \(150\) epochs with learning rate \(10^{-3}\). StoNet is trained for \(20\) epochs using stochastic gradient MCMC with \(25\) Metropolis--Hastings steps per iteration, followed by a \(30\)-epoch fine-tuning stage. GenSDR is trained for \(50\) epochs using Adam with learning rate \(10^{-3}\), batch size \(256\), \(\tau=0.001\), and no weight decay. SRDR is trained for \(40{,}000\) optimization steps using Adam with learning rate \(2\times10^{-4}\), weight decay \(10^{-4}\), batch size \(256\), and \(s=10\).

\paragraph{Heteroscedastic.}
All methods are evaluated with \(s=2\). GMDDNet is implemented with a \(2\times 50\) ReLU network, and is trained for \(100\) iterations using Adam with learning rate \(10^{-3}\) and batch size \(100\). StoNet follows the bottleneck architecture \(25 \rightarrow 2 \rightarrow 1\), where the two-dimensional hidden layer is used as the learned reduction. It is pretrained for \(60\) epochs with batch size \(64\), step size \(0.01\), proposal learning rates \((7\times 10^{-4},7\times 10^{-4})\), sigma list \((10^{-2},10^{-4})\), Metropolis--Hastings step \(25\), momentum \(0.9\), and no weight decay. The StoNet reduction is then fine-tuned for \(30\) epochs with step size \(0.001\), followed by a \(\tanh\) transformation of the bottleneck output.

BENN uses the kernel transform with \(m=1000\), \(150\) training epochs, and learning rate \(10^{-3}\). SRDR uses a three-layer dimension reduction network with hidden width \(20\), a two-layer generative network with hidden width \(20\) and noise dimension \(5\). It is trained for \(20{,}000\) steps using learning rate \(2\times 10^{-4}\) and batch size \(\min(256,n)\). GenSDR uses a three-hidden-layer ReLU dimension reduction network with hidden width \(20\), together with a matching three-hidden-layer ReLU velocity network. GenSDR is trained for \(50\) epochs using Adam with learning rate \(10^{-3}\), batch size \(\min(256,n)\), \(\tau=0.001\), and no weight decay. Performance for each method is summarized using distance correlation between the learned representation and the true sufficient reduction \(f(X)\) on an independently generated test set of size \(2000\).

\paragraph{PACS.}
We make a stratified $80/20$ train/validation split on CIFAR-10. For each PACS target domain, we construct a stratified $60/20/20$ train/validation/test split and repeat this procedure over seeded target resplits.

The TESR implementation follows the two-stage construction of \citet{GeZhouHuang2025}. PACS images are normalized to $[-1,1]$, and resized to $32\times32$. Both the source representation $R_c$ and the target augmentation $R_t$ use the paper-style PACS CNN architecture: two $3\times3$ convolutional layers with widths $48$ and $96$, batch normalization and LeakyReLU$(0.2)$ activations, a $2\times2$ max-pooling layer, two additional $3\times3$ convolutional layers with widths $192$ and $256$, a second $2\times2$ max-pooling layer, and a linear projection from $256\cdot 8\cdot 8$ to $1024$, followed by a final representation head. We set $\dim(R_c)=\dim(R_t)=64$, use batch size $128$, and optimize with RMSprop using learning rate $5\times10^{-4}$ and weight decay $10^{-4}$. Each phase is trained for $300$ epochs. The regularization weights are $\lambda_E=\lambda_Z=\lambda_{E,0}=0.01$ and $\lambda_C=1$. 

SRDR uses a six-layer dimension reduction network with hidden width $320$, paired with a two-layer generative prediction network with hidden width $128$ and noise dimension $10$. Both the enhanced and frozen SRDRs are optimized with Adam using learning rate $5\times10^{-4}$ and weight decay $10^{-5}$.

\end{document}